\documentclass[final,4p]{elsarticle}
\usepackage[a4paper, top=2.5cm, bottom=2.5cm, left=2.5cm, right=2.5cm]{geometry}
\usepackage{natbib}
\setcitestyle{authoryear,round}
\usepackage{lipsum}
\makeatletter
\def\ps@pprintTitle{%
 \let\@oddhead\@empty
 \let\@evenhead\@empty
 \def\@oddfoot{}%
 \let\@evenfoot\@oddfoot}
\makeatother
\usepackage{amsmath}
\usepackage{graphicx}
\usepackage{ragged2e}
\usepackage{float}
\usepackage{bigints}
\usepackage{appendix}
\usepackage{booktabs}

\DeclareMathOperator*{\argmin}{arg\,min}
\usepackage{lscape} 
\usepackage{cases}
\usepackage{todonotes}

\newcommand{\E}{\mathbb{E}}
\newcommand{\var}{\mathbb{V}ar}
\newcommand{\M}{\mathcal{M}}
\renewcommand{\P}{\mathcal{P}}
\usepackage{arydshln}
\usepackage{mathtools}
\usepackage{afterpage}
\usepackage{amsfonts}
\usepackage{algorithm,algpseudocode}
\usepackage{arydshln}
\usepackage[scr]{rsfso}
\usepackage{multirow}
\usepackage[flushleft]{threeparttable}
\usepackage{siunitx}
\usepackage{caption}
\usepackage{rotating}
\usepackage{amsfonts}
\usepackage{booktabs}
\usepackage{siunitx}

\usepackage{tabularx}
\usepackage{subcaption}
\usepackage{amsthm}
\usepackage{amssymb,bm}
\usepackage{amsfonts}
\definecolor{applegreen}{rgb}{0.55, 0.71, 0.0}
\definecolor{ao(english)}{rgb}{0.0, 0.5, 0.0}
\usepackage{booktabs,caption}
\usepackage[pdftex,colorlinks,
  citecolor=black]{hyperref}
\usepackage{setspace}
\usepackage{bbold}

\newcommand{\norm}[1]{{\left\lVert#1\right\rVert}}

\newcommand{\bA}{\boldsymbol A}

\newcommand{\bB}{\boldsymbol B}
\newcommand{\bb}{\boldsymbol b}
\newcommand{\bC}{\boldsymbol C}

\newcommand{\bD}{\boldsymbol D}

\newcommand{\be}{\boldsymbol e}
\newcommand{\bF}{\boldsymbol F}

\newcommand{\bG}{\boldsymbol G}

\newcommand{\bI}{\boldsymbol I}

\newcommand{\bM}{\boldsymbol M}

\newcommand{\bP}{\boldsymbol P}

\newcommand{\bQ}{\boldsymbol Q}

\newcommand{\bR}{\boldsymbol R}
\newcommand{\br}{\boldsymbol r}
\newcommand{\bS}{\boldsymbol S}

\newcommand{\bu}{\boldsymbol u}
\newcommand{\bV}{\boldsymbol V}
\newcommand{\bv}{\boldsymbol v}
\newcommand{\bW}{\boldsymbol W}
\newcommand{\bw}{\boldsymbol w}
\newcommand{\bX}{\boldsymbol X}
\newcommand{\bx}{\boldsymbol x}
\newcommand{\bY}{\boldsymbol Y}
\newcommand{\by}{\boldsymbol y}
\newcommand{\bZ}{\boldsymbol Z}
\newcommand{\bz}{\boldsymbol z}

\newcommand{\bzero}{\boldsymbol 0}
\newcommand{\abs}[1]{{\left\lvert#1\right\rvert}}

\newcommand{\bbeta}{\boldsymbol \beta}

\newcommand{\bdelta}{\boldsymbol \delta}

\newcommand{\bzeta}{\boldsymbol \zeta}
\newcommand{\boeta}{\boldsymbol \eta}

\newcommand{\bkappa}{\boldsymbol \kappa}

\newcommand{\bnu}{\boldsymbol \nu}
\newcommand{\bxi}{\boldsymbol \xi}

\newcommand{\bLambda}{\boldsymbol \varLambda}

\newcommand{\bPi}{\boldsymbol \varPi}
\newcommand{\bSigma}{\boldsymbol \varSigma}

\newcommand{\bPhi}{\boldsymbol \varPhi}

\newcommand{\bOmega}{\boldsymbol \varOmega}

\usepackage{dsfont}
\usepackage{comment}
\usepackage{color}
\usepackage{xcolor}

\usepackage[normalem]{ulem}

\newcommand{\la}{-(p+d)}
\renewcommand{\S}{\mathcal{S}}
\newcommand{\hS}{\hat{\S}}
\newcommand{\hSmj}{\hS_j^+}
\newcommand{\Smj}{\S_j^+}
\newcommand{\Wald}{\text{Wald}}
\newcommand{\LM}{\text{LM}}

\newcommand{\pref}[1]{(\ref{#1})}
\newcommand{\partref}[2]{\ref{#1}(\ref{#2})}

\theoremstyle{definition}

\newtheorem{assumption}{Assumption}

\newtheorem{remark}{Remark}
\newtheorem{theorem}{Theorem}
\newtheorem{lemma}{Lemma}

\theoremstyle{plain}
\makeatletter
\def\adl@drawiv#1#2#3{%
        \hskip.5\tabcolsep
        \xleaders#3{#2.5\@tempdimb #1{1}#2.5\@tempdimb}%
                #2\z@ plus1fil minus1fil\relax
        \hskip.5\tabcolsep}
\newcommand{\cdashlinelr}[1]{%
  \noalign{\vskip\aboverulesep
           \global\let\@dashdrawstore\adl@draw
           \global\let\adl@draw\adl@drawiv}
  \cdashline{#1}
  \noalign{\global\let\adl@draw\@dashdrawstore
           \vskip\belowrulesep}}
\makeatother


\allowdisplaybreaks

\begin{document}
\bibliographystyle{apa}
\begin{frontmatter}

\title{Inference in Non-stationary High-Dimensional VARs}
\author{Hecq, A.$^{\dagger}$, Margaritella, L.$^{\ddagger}$\footnote{The second and third author thank NWO for financial support under grant number 452-17-010. Previous versions of this paper have been presented at the (EC)$^2$ 2020, ICEEE 2021, IAAE 2021 and NESG 2022 conferences as well as the 2022 Maastricht Dimensionality Reduction and Inference in High-dimensional Time Series and the 2022 Maastricht-York Economics workshops. We gratefully acknowledge the comments by participants at these conferences. In addition, we thank Anders B. Kock, Joakim Westerlund, Gianluca Cubadda, Etienne Wijler for their valuable feedback. All remaining errors are our own.
Address correspondence to: Luca Margaritella, Department of Economics, Lund University, P.O. Box 7082, S-220 07 Lund, Sweden. E-mail: \href{mailto:luca.margaritella@nek.lu.se}{\textcolor{blue}{luca.margaritella@nek.lu.se}}}, Smeekes, S.$^{\dagger}$}
\address{$^{\dagger}$Maastricht University, $^{\ddagger}$Lund University\bigbreak \normalsize{\today}}
\date
\maketitle
\begin{abstract}
In this paper we construct an inferential procedure for Granger causality in high-dimensional non-stationary vector autoregressive (VAR) models. Our method does not require knowledge of the order of integration of the time series under consideration. We augment the VAR with at least as many lags as the suspected maximum order of integration, an approach which has been proven to be robust against the presence of unit roots in low dimensions. We prove that we can restrict the augmentation to only the variables of interest for the testing, thereby making the approach suitable for high dimensions. We combine this lag augmentation with a post-double-selection procedure in which a set of initial penalized regressions is performed to select the relevant variables for both the Granger causing and caused variables. We then establish uniform asymptotic normality of a second-stage regression involving only the selected variables. Finite sample simulations show good performance, an application to investigate the (predictive) causes and effects of economic uncertainty illustrates the need to allow for unknown orders of integration.

\bigbreak
\end{abstract}
\end{frontmatter}

\doublespacing

\section{Introduction}
In this paper we construct an inferential procedure for Granger causality in high-dimensional non-stationary vector autoregressive (VAR) models. Investigating causes and effects in time series models has a long and rich history, dating back to the seminal work of \citet{granger1969investigating}. The statistical assessment of the directional predictability among two (or blocks of) time series can have important consequences for decision making processes. Applications of Granger causality range over from macroeconomics, finance, network theory, climatology and even the neuroscience. For instance, the evidence of causality between money and gross domestic product is a  long-debated issue in the macroeconomic literature. This was first discussed in \citet{sims1990inference} and \citet{stock1989interpreting} and it is still argued to this day, see e.g., \citet{miao2020high}. Financial applications of Granger causality include, among others, \citet{billio2012econometric} who study the connectedness among monthly returns of hedge funds, banks, broker/dealers, and insurance companies. \citet{hecq2021granger} build Granger causality networks from high-dimensional vector autoregressive models, describing the dynamic volatility spillovers among a large set of stock returns. Many applications are also found in climate science, for instance in trying to understand and disentangle the causes of climate change. Among others, \citet{stern2014anthropogenic} investigate causality between greenhouse gas emissions and temperature. In neuroscience, Granger causality is employed to understand mechanisms underlying complex brain function and cognition, with examples in the field of functional neuroimaging \citep[see e.g.,][for some reviews]{seth2015granger, friston2013analysing}.

With the increased availability of larger and richer datasets, these causality concepts have recently been extended to a high-dimensional setting where they can benefit from the inclusion of many more series within the information set. In the original concept of Granger causality \citep{granger1969investigating}, conditioning on a given information set plays a central role.\footnote{Throughout the paper we focus on Granger causality in mean. While there are clearly many other forms of causal or predictive relations possible, Granger causality in mean is the most prominent and therefore the focus of our attention. It should therefore be understood that whenever we mention Granger causality in the paper, we refer to Granger causality in mean.} After all, Granger causality is a study of predictability. Only by considering predictability given a specific conditioning set, is possible to attach some sort of causal meaning to the outcome (where the nature of that causality is still up for debate, of course). To identify direct `truly causal' links between variables would require one to condition upon all possible variables that may be related to the variables of interest. Otherwise, any discovered relation may simply be an artifact of omitted variable bias which would invalidate a causal interpretation. Indeed, Granger himself envisioned the information set as ``all the knowledge in the universe available at that time" \citep[][p.330]{granger1980testing}. While this concept cannot be operationalized with any finite amount of data, the availability of increasingly high-dimensional datasets, along with the econometric techniques to analyze them, provide a great opportunity for turning Granger causality from `just predictability' into a concept to which at least some form of causal interpretation can be attached. 

Recently, \citet{hecq2021granger} (HMS henceforth) proposed a test for Granger causality in high-dimensional vector autoregressive models. This combines dimensionality reduction techniques based on penalized regression such as the lasso of \citet{tibshirani1996regression}, with the post-double selection procedure of \citet{belloni2014inference} designed to guarantee uniform asymptotic validity of the post-selection least squares estimator. However, HMS assumed stationarity of all the time series considered. This is a typical assumption in much of the literature on regularization methods, in particular when inference is considered. While the literature on estimation and forecasting with high-dimensional non-stationary processes is growing \citep[see e.g.][and the references therein]{SmeekesWijler20}, this is not the case for inference due to the complexities arising with unit roots and cointegration, which already have severe effects in low-dimensional settings.

Working with (data possibly transformed to yield) stationary time series avoids these complications in the asymptotic analysis and allows to invoke standard (Gaussian) limit theory, thereby enabling the use of standard inferential procedures such as $t$ and $F$-tests. On the other hand, it assumes prior knowledge of the order of integration of all the time series entering the model. This prior knowledge is usually acquired via unit root tests such as the augmented Dickey-Fuller (ADF) test \citep[see][]{dickey1979distribution}. However, these tests are sensitive to their specifications, such as the inclusion of a deterministic time trend in the regression equation and the choice of the lag length. Even when performing what seems to be the same test in, for example, different R packages, it may actually lead to different outcomes \citep{SmeekesWilms20}. That does not even take into account that there are many different unit root tests, none of which is clearly preferred to the others, as well as how specific data features (and their treatment) such as seasonality (adjustments), structural breaks and outliers may affect these tests. As such, the outcome of a unit root test is ambiguous at best; even more so if one also takes into account their well-documented low power \citep{cochrane1991critique} on the one hand and the accumulation of the probability of obtaining false positives through multiple testing -- especially in high dimensions -- on the other hand.

But even if one did know the true orders of integration, transformations such as differencing to achieve stationarity are not innocuous. In particular, information about long-run relations such as cointegration is deleted from the series. Yet the error correction mechanisms at play for the movement towards long-run equilibria may well induce (Granger) causal relations that are not apparent anymore in the transformed data. While this is not necessarily a problem if one is interested in the transformed data as such (for example the growth rate of economic series may be an object of interest in themselves), in many applications we are interested in the level of the series and the transformations are only done to avoid issues with non-stationarity. It would therefore be very beneficial to practitioners to have high-dimensional inference methods available that are robust to (unknown) orders of (co)integration.

In this paper we develop a method which allows for testing Granger causality in high-dimensional VAR models, irrespective of the (co)integration properties of the time series in the VAR. We avoid any bias coming from unit root and cointegration pre-testing and instead use the VAR in levels directly to perform inference on the (Granger causality) parameters of interest. The procedure we design builds on the work of \citet{toda1995statistical} who first considered a simple lag augmentation of the system, which they showed provides asymptotically normal estimators irrespective of potential unit roots and cointegration \citep[see also][]{dolado1996making}. The same approach has also recently been used to develop inference on impulse response functions \citep{inoue2020uniform,montiel2021local} that is robust to unit roots. We modify the approach outlined above by confining lag augmentation to only the variable(s) tested as Granger causing, instead of adding lags of all variables. We show that this modification, which makes it suitable for application to high-dimensional VARs, does not affect the asymptotic properties. In addition, we modify the post-double selection procedure of HMS to prevent the possibility of spurious regression, thereby extending its uniform validity to data with potential (co)integrated time series.

The remainder of the paper is organized as follows: Section \ref{sec:model} introduces the model, the Granger causality testing and our lag-augmented post-double-selection approach. The theoretical properties of our method are studied in Section \ref{sec:theory}. In Section \ref{sec:sim} we investigate finite-sample performance through a simulation study, while Section \ref{sec:choicep} proposes a data-driven method to find a sensible upper bound for the lag length in the VAR. Section \ref{sec:emp_appl} uses the proposed testing framework to investigate the causes and effects of economic uncertainty in the context of the FRED-MD dataset. Section \ref{sec:conclusion} concludes. In Appendix \ref{sec:app_theory} the theory underlying the lag augmentation and the asymptotic properties is developed. Introductory Lemmas and complementary results to the following two appendices are also given. Appendix \ref{sec:verifyUR} is devoted to the validation of high-level assumptions, while additional empirical results are presented in Appendix \ref{sec:vix}.

A few words on notation. For any $n$-dimensional vector $\bx$, we let $\norm{\bx}_p = \left(\sum_{i=1}^n |x_i|^p \right)^{1/p}$ denote the $\ell_p$-norm. For any index set $S \subseteq \{1, \ldots, n\}$, let $\bx_{S}$ denote the sub-vector of $\bx_t$ containing only those elements $x_i$ such that $i \in S$. $|S|$ denotes the cardinality of the set $S$. We use $\xrightarrow{p}$ and $\xrightarrow{d}$ to denote convergence in probability and distribution, respectively.

\section{Granger Causality Tests for Nonstationary High-Dimensional Data}\label{sec:model}
In this section we propose our model, our strategy for lag augmentation of the Granger causality tests, and the post-double-selection procedure needed to achieve uniformly valid inference. Section \ref{sec:GClag} first presents the model and lag-augmented Granger causality test. Next, \ref{sec:pds} sets out the post-double-selection procedure.

\subsection{Lag-Augmented Granger Causality Testing}\label{sec:GClag}
Let $\bz_1,\ldots,\bz_T$ be a $K$-dimensional multiple time series process, where $\bz_t=(y_t, x_t, \bw_t)^{'}$. Here $y_t$ is the series we would like to test being Granger caused, $x_t$ the potentially Granger causing series and $\bw_t$ is a $K-2$ dimensional vector of controls constituting the \textit{information set}. We allow for $K$ to be large, potentially larger than (and growing with) the sample size $T$. We assume $\bz_t$ is generated by a VAR($p$) process as
\begin{equation} \label{eq:dgp}
\bz_t=\bA_1 \bz_{t-1}+\cdots+\bA_{p}\bz_{t-{p}}+\bu_t, \qquad t=p+1,\ldots,T\;,
\end{equation}
where $\bA_1,\ldots,\bA_{p}$ are $K\times K$ parameter matrices and $\bu_t$ is a $K\times 1$ vector of error terms.
\begin{assumption} \label{ass:basic}
The VAR model in \eqref{eq:dgp} satisfies:
\begin{enumerate}[(a)]
\item\label{ass:basic1} $\{\bu_t\}_{t=1}^T$ is an mds with respect to  the filtration $\mathcal{F}_t = \sigma(\bz_t, \bz_{t-1}, \bz_{t-2}, \ldots)$ $\bu_t$ such that $\E (\bu_t| \mathcal{F}_{t-1}) = \bzero$ for all $t$; the $K\times K$ covariance matrix $\bSigma_{u} = \E (\bu_t \bu_t^\prime)$ is positive definite and $\mathbb{E}|\boldsymbol{u}_{t}|^{2+\delta}\leq \infty$, for $\delta>0$. 
\item\label{ass:basic2} The roots of $\det(\bI_{K}-\sum_{j=1}^{p} \bA_j z^j)$ can either lie on the unit disc or outside.
\end{enumerate}
\end{assumption}

Note that Assumption \ref{ass:basic}(\ref{ass:basic2}) allows for the time series to have unit roots and be cointegrated. We specifically allow elements of $\bz_t$  to be integrated of order $d$: $I(d)$ for $d=0,1,2$ and possibly cointegrated of order $d,b$: $CI(d,b)$ with $0<b\leq d$. We formulate more specific assumptions on (co)integration properties in Section \ref{sec:theory} and Appendix \ref{sec:verifyUR}.

We are interested in testing the null hypothesis of Granger non-causality in mean between the Granger causing series $x_t$ and the Granger caused $y_t$, conditional on all the series in $\bw_t$. For the moment we assume the lag-length $p$ in \eqref{eq:dgp} to be known; we shall further elaborate on data-driven ways to select $p$ in Section \ref{sec:choicep}. Also, in order to ease the notation, we omit both the intercept and any polynomial time trend from the model; the results we derive easily extend to those cases as well. It is convenient to introduce the following stacked notation. Let $\bX_{-p} = (\bx_{-1}, \ldots, \bx_{-p})$ denote the $T\times p$ matrix containing the $p$ lags of the Granger causing variable $x_t$, where $\bx_{-j}$ is the vector containing the observations corresponding to its $j$-th lag.\footnote{To keep $T$ observations for the lags, one can replace the missing values for $t\leq 0$ with zeros. Alternatively, the vectors are shortened to contain $T-p$ observations only. Both approaches are equivalent asymptotically, and for notational simplicity in the following we do not explicitly distinguish between them.} Similarly, we define $\bY_{-p}$ containing the $p$ lags of the Granger caused variable $y_t$ and $\bW_{-p}$ for the conditioning set such that $\bW_{-p}$ is a $T\times (K-2)p$ matrix. Then, our equation of interest for the Granger causality testing can be stated as
\begin{equation}\label{eq:dgp_beta}
\by = \bX_{-p} \bbeta + \bY_{-p} \bdelta_1 + \bW_{-p} \bdelta_2 + \bu = \bX_{-p} \bbeta + \bV \bdelta + \bu,
\end{equation} 
where for notational simplicity we define $\bV = (\bY_{-p}, \bW_{-p})$ as the matrix containing all control variables and $\bdelta = (\bdelta_1', \bdelta_2')'$ its coefficients.

Testing for no Granger causality is then equivalent to testing the following null hypothesis: 
\begin{equation}\label{eq:H0}
H_0: \bbeta=\bzero\quad \text{against}\quad H_1: \bbeta\neq\bzero.
\end{equation}

To account for the potential unit roots in the system, we follow the approach pioneered by \citet{toda1995statistical} and \citet{dolado1996making} of augmenting the regression of interest with redundant lags of the variables. However, in contrast to the existing approaches, we only augment the lags of the Granger causing series $x_t$. That is, we consider the lag-augmented regression
\begin{equation}\label{eq:lag_augmented}
\by = \bX_{-p} \bbeta + \sum_{j=1}^{d} \beta_{p+j} \bx_{-(p+j)} + \bV \bdelta + \bu = \bX_{\la} \bbeta_{+} + \bV \bdelta + \bu,
\end{equation}
where $\bX_{\la} = (\bX_{-p}, \bx_{-(p+1)}, \ldots, \bx_{-(p+d)})$ contains the $p+d$ lags of $x_t$. Here $d$ represents the maximum order of integration one suspects the series are having. Note that in fact $\beta_{p+j} = 0$ for all $j\geq 1$ such that $\bbeta_{+}=(\bbeta', \bzero_d')'$, as we are adding redundant variables. It should be clear that by testing whether the first $p$ elements of $\bbeta_{+}$ -- those corresponding to $\bbeta$ in \eqref{eq:dgp} -- we can perform the same test of Granger non-causality as in the original setup.

By adding ``free'' lags of $x_t$, we in essence allow this variable to ``difference itself'' into the correct order to remove the unit roots. As a simple illustration, consider the following DGP where $x_t$ may be $I(1)$:
\begin{equation} \label{eq:simple}
y_t = \beta x_{t-1} + u_{1,t}, \qquad x_t = \rho x_{t-1} + u_{2,t}, \quad \abs{\rho} \leq 1.
\end{equation}
Testing whether $\beta=0$ by estimating this regression directly, is complicated as the limit distribution of the least squares estimator and the corresponding test changes depending on whether $x_t$ is $I(1)$ or $I(0)$. By adding a redundant lag we can write \eqref{eq:simple} as
\begin{equation} \label{eq:simple2}
y_t = \beta x_{t-1} + 0 x_{t-2} + u_{1,t} = \beta u_{2,t-1} + \beta \rho x_{t-2} + u_{1,t},
\end{equation}
suggesting that regressing on $x_{t-1}$ and $x_{t-2}$ is equivalent to regressing on $u_{2,t-1}$ and $x_{t-2}$. Moreover, if $u_{2,t}$ were observed, we could test $\beta=0$ directly using its regression coefficient. In Appendix \ref{sec:LA} we show formally that not only the equivalence in \eqref{eq:simple2} continues to hold in the presence of higher order lags as well as control variables, but also that testing for Granger causality in the lag-augmented regression is in fact equivalent to testing on the appropriately transformed variables. This, in turn, allows us to retrieve the asymptotic normality of the least squares estimator regardless of the order of integration, provided that the lag order $p$ of the VAR and \textit{maximum} order of integration $d$ are correctly specified.

The main difference compared to the original results of \citet{toda1995statistical} is that we only augment the Granger causing series $x_t$. While this difference is not of importance when $K$ is small, it opens the door for high-dimensional applications where $K$ is large, and even larger than $T$. As we cannot estimate such regressions with least squares anymore, we combine lag-augmentation with the post-double-selection framework of HMS to construct Granger causality tests in high-dimensional models that are robust to unknown orders of integration. We describe the method in the next section.

\begin{remark}\label{rmk:bivariate}
For ease of exposition we confine our attention to the study of bivariate Granger causality relations, conditional on a large information set. At the cost of more involved notation and algebra, our approach can be extended to Granger causality between multiple variables. In that case, lag augmentation is needed for all Granger causing variables, which is only feasible if the block of Granger causing variables is not too large. HMS provide details about this in the stationary setting; the same approach can be adapted here.
\end{remark}

\begin{remark}\label{rmk:overdiff}
One important aspect of the current framework is that the lag-length $p$ of the VAR is necessarily larger than, or at least equal to, the suspected maximum order of integration $d$. It is therefore important to specify $p$ correctly in practice. We return to this issue in Section \ref{sec:choicep}. One might also worry that specifying $d$ too high, or having mixed orders of integration among multiple Granger causing variables would lead to over-differencing, i.e., moving average unit roots being introduced by differencing stationary time series \citep[see e.g.,][]{chang1994recognizing}. This, however, does not happen here as the additional lags of the Granger causing and Granger caused variables are used ``at convenience"; if they are not needed because the variables are already stationary, inclusion of these redundant will at most marginally decrease the power of the test given the small over-specification of the lag length.
\end{remark}

\begin{remark}\label{rmk:d2}
Even if the variables in a particular dataset are thought to be at most of order $I(1)$, it may still pay off to take $d=2$. The choice of $d=2$ is supported by simulations reported later in Section \ref{sec:sim}. It is well known that when one or more roots of the characteristic polynomial are close to unity, the distribution of the least squares estimator becomes skewed, yielding estimators which tend to underestimate the true autoregressive parameters \citep[see e.g.][]{fuller2009introduction}. This also causes difficulties in performing inference on these parameters. By augmenting with $d=2$ lags, one can avoid any issues with near unit roots. The simulations reported in Section \ref{sec:sim} confirm that in the presence of substantial autocorrelation beyond the first lag, augmenting with $d=2$ lags is an easy way to improve the finite sample behaviour of the test.
\end{remark}

\begin{remark}\label{rmk:coint}
If the $I(d)$ series object of the Granger causality test are also actually cointegrated, then the lag-augmentation does not serve any purpose as no spurious relation would be estimated \citep{dolado1996making}. This causes a slight loss in power which however is in practice minimal, as shown in Section \ref{sec:sim}. If, however, the series in object are not cointegrated, then the lag-augmentation becomes paramount to render the series difference-stationary before the estimation and testing. Note that also all sorts of situations in between are allowed: some variables may be cointegrated, others contain `pure' stochastic trends, and others may be stationary. In fact, in our high-dimensional setup where we apply the test to a large dataset containing many variables, such mixed properties seem likely. 
\end{remark}

\subsection{Inference after selection by the lasso} \label{sec:pds}
Assuming that $\bdelta$ is sparse, in the sense that many elements of the vector are equal to zero, one might consider estimating \eqref{eq:lag_augmented} with a technique that does variable selection such as the lasso, and then re-estimating the equation including only the selected control variables. However, doing so we run into the problems with inference after model selection \citep{leeb2005model}, where even if selection is done consistently, the (diminishing) probability of omitting a relevant variable causes a sufficiently large bias to prevent uniform convergence of the post-selection least squares estimator to a Gaussian limit distribution. To circumvent this problem \citet{belloni2014inference} introduced the post-double-selection (PDS) method that not only selects relevant variables on the outcome variable, but also on the treatment variable. This double selection reduces the probability of admitting relevant variables to such an extent that it does not affect the asymptotic distribution anymore.

The PDS framework was used by HMS to develop high-dimensional tests for Granger causality in sparse VAR models. In this section we show how to adapt the post-double selection framework of HMS to the unit root non-stationary framework with lag augmentation. In this framework we first perform a set of initial regressions -- to be estimated with variable selection techniques such as the lasso -- of the dependent variable plus the explanatory variables of interest (here the $p$ lags of $x_t$) on all other variables. To properly account for potential unit roots and avoid spurious regression in these initial regressions, we need to slightly adapt the setup of HMS.

Let $\bX_{-p,\backslash\{j\}}$ denote the matrix $\bX_{-p}$ from which the $j$-th column, corresponding to the $j$-th lag of $x_t$, has been removed. Let $\bZ_{-p} =(\bX_{-p}, \bY_{-p}, \bW_{-p})$ denote the matrix of all (non-augmented) variables and $\bZ_{-p,\backslash\{j\}}=(\bX_{-p,\backslash\{j\}}, \bY_{-p}, \bW_{-p})$.
Then we consider the following regressions in the first step:
\begin{equation} \label{eq:lasso_steps}
\by = \bZ_{-p} \boeta_0 + \bu,\qquad
\bx_{-j} = \bZ_{-p,\backslash\{j\}} \boeta_j + \be_j, \quad j=1,\ldots,p.
\end{equation}
In contrast to HMS, who follow \citet{belloni2014inference} by only regressing $\by$ and $\bx_{-j}$ on the controls $\bY_{-p}$ and $\bW_{-p}$, we also add all other lags of $x_t$. This is done to avoid that the error terms in \eqref{eq:lasso_steps} contain unit roots and the regressions become spurious.

\begin{remark} \label{rem:p=d}
Note that in the situation where $p=d$, there may still be a risk of spurious regression in the first step. E.g.,~consider $p=d=1$ where the regression of $x_{t-1}$ on the remaining variables will not contain any lags (or leads) of $x_{t-1}$. This yields a spurious regression if $x_t$ is a non-cointegrated unit root process. We therefore recommend to always take $p \geq d+1$ in such cases. If this is not desirable, e.g.~if the sample size is too small to sustain such a large $p$, it is also possible to augment the first-stage regressions with an additional lag of $x_t$ to eliminate the possibility of spurious regression. In the following we work under the condition that the error terms $\be_j$ are $I(0)$, thereby implicitly assuming that one of the two measures suggested above have been taken if needed.
\end{remark}

Although the concept is tricky with integrated variables, the coefficients $\boeta_j$ in \eqref{eq:lasso_steps} can be thought of as best linear predictors. Indeed, as the error terms $\be_j$ are $I(0)$, these coefficients are ensured to exist and can be defined as the probability limits of the respective least squares estimators while keeping the dimension $K$ fixed.

As will be formalized in Assumption \ref{as:hlev}, we assume that the coefficients $\boeta_j$ are sparse. Let $\S_j = \{m>p: \eta_{m,j} \neq 0\}$, $j=0,1, \ldots, p$, be the sets of active variables in \eqref{eq:lasso_steps}.\footnote{We exclude the variables corresponding to the lags of $x_t$ from $\S_j$ as these will always be included in the second stage.} Then, we need that the cardinality of these sets is small; that is, smaller than $T$ and at most growing at a slow rate of $T$. In fact, the union over all these sets, $\S := \bigcup_{j=0}^{p} \S_j$ is our set of interest. This set represents all variables needed to control for when regressing $\by$ on $\bX_{-p}$, as they either have non-zero coefficients in \eqref{eq:dgp_beta} \emph{or} are correlated with $x_t$. In fact, variables would need to have both properties to cause omitted variable bias if they were not included in the final regression. By aiming to recover $\S$, we therefore have two opportunities to select relevant variables. This is sufficient to reduce the probability of missing them to be asymptotically negligible.

Let $\bV_{\S}$ denote the matrix consisting of only those columns of $\bV$ that corresponds to the selected variables in $\S$, and $\bdelta_{\S}$ the corresponding parameter vector. Our goal of the PDS regression is then to recover the regression
\begin{equation} \label{eq:sec_stage_trueS}
\by = \bX_{\la} \bbeta_+ + \bV_{\S} \bdelta_{\S} + \bu,
\end{equation}
in the second stage, and base inference on this.

To obtain an estimate for the set $\S$, one can use the lasso or any of its relatives that similarly ensure variable selection. The lasso \citep[see][]{tibshirani1996regression} simultaneously performs variable selection and estimation of the parameters in \eqref{eq:lasso_steps} by solving the following minimization problems 
\begin{equation}
\begin{split}\label{eq:lasso}
&\hat{\boeta}_{0}= \underset{\boeta}{\arg\min}\bigg(T^{-1}\norm{\by -  \bZ_{-p} \boeta}_2^2 + \lambda \norm{\boeta}_1 \bigg),\\
&\hat{\boeta}_{j}=\underset{\boeta}{\arg\min}\bigg(T^{-1}\norm{\bx_{-j} - \bZ_{-p,\backslash\{j\}} \boeta}_2^2 + \lambda \norm{\boeta}_1 \bigg), \quad j=1,\ldots,p,
\end{split}
\end{equation}
where $\lambda$ is a non-negative tuning parameter determining the strength of the penalty. Here, we follow the framework of HMS and use the Bayesian information criterion (BIC) in selecting the tuning parameter, coupled with a penalty lower bound ensuring a maximum of selected variables per estimated equation (see Remark \ref{rmk_5} for details). Minimizing an information criterion (IC) in order to determine an appropriate data-driven $\lambda$ is one way to deal with dependent data (see HMS for an overview of other methods and their finite sample behaviors).

The first step of our PDS method is then to estimate each of the regressions in \eqref{eq:lasso_steps} using a penalized regression technique such as in \eqref{eq:lasso}. Let $\hS_j = \{m>p: \hat\eta_{m,j} \neq 0\}$ represent the corresponding sets of active variables retained in each of these regressions, and let $\hS = \bigcup_{j=0}^{p} \hS_j$ denote the set of all active variables. Then, we base our inference on the second-stage regression
\begin{equation} \label{eq:sec_stage_hatS}
\by = \bX_{\la} \bbeta_+ + \bV_{\hS} \bdelta_{\hS} + \bu,
\end{equation}
which may now be treated as if no selection took place, and can simply be estimated by least squares. As we will show in the next section, the OLS estimator converges uniformly to a normal distribution, which in turn makes standard tests such as the Wald test or the LM test applicable with their regular limiting distributions.

Algorithm \ref{alg:pdslm} describes the main steps of our post-double selection lag-augmented (PDS-LA) Granger causality test implemented through the Lagrange Multiplier (LM) or Wald test.

\begin{algorithm}[H]
\caption{Post-double selection lag-augmented (PDS-LA) Granger causality test} \label{alg:pdslm}
\begin{itemize}
\item[\textbf{[1]}]
Estimate the initial partial regressions in \eqref{eq:lasso_steps} by an appropriate sparsity-inducing estimator such as the (adaptive) lasso, and let $\hat{\boeta}_0, \ldots, \hat{\boeta}_{p}$ denote the resulting estimators. Collect the selected variables in the sets $\hS_j = \{m>p: \hat\eta_{m,j} \neq 0\}$ for $j=0,1,\ldots,p$.

\item[\textbf{[2]}] Let $\hS = \bigcup_{j=0}^p \hS_j$ denote the set of all variables selected in at least one regression at Step [1]. Denote its cardinality by $\hS=\abs{\hS}$.

\medskip
Then, consider one of the following two tests:
\begin{center}
\hrule
PDS-LA-LM\medskip
\hrule
\end{center}

\item[\textbf{[3]}] Regress $\by$ on $\bV_{+,\hS} := (\bx_{-(p+1)}, \ldots, \bx_{-(p+d)}, \bV_{\hS})$ by OLS and obtain the residuals $\hat{\bxi}$.

\item[\textbf{[4]}] Regress $\hat{\bxi}$ onto $\bV_{+,\hS}$ plus the Granger causing variables $\bX_{-p}$, retaining the residuals $\hat{\bnu}$. Then obtain $R^2 = 1 - \hat{\bnu}^\prime \hat{\bnu} / \hat{\bxi}^\prime \hat{\bxi}$

\item[\textbf{[5a]}] 
Reject $H_0$ if $\LM = TR^2 > q_{\chi_{p}^2}(1-\alpha)$, where $q_{\chi_{p}^2}(1-\alpha)$ is the $1-\alpha$ quantile of the $\chi^2$ distribution with $p$ degrees of freedom.

\item[\textbf{[5b]}] 
Reject $H_0$ if \small $\bigg(\frac{T - \hS-(p+d)}{p} \bigg)\bigg(\frac{R^2}{1 - R^2} \bigg)>q_{F_{p, T - \hS-(p+d)}} (1-\alpha)$\normalsize, where $q_{F_{p, T - \hS-(p+d)}}$ is the $1-\alpha$ quantile of the $F$ distribution with $p$ and $T - \hS-(p+d)$ degrees of freedom.

\begin{center}
\hrule
PDS-LA-Wald\medskip
\hrule
\end{center}

\item[\textbf{[3]}] Regress $\by$ on $\bX_{-p}$ and $\bV_{+,\hS} := (\bx_{-(p+1)}, \ldots, \bx_{-(p+1)}, \bV_{\hS})$ by OLS, obtaining $\hat{\bbeta}_p$, the estimated coefficients corresponding to $\bX_{-p}$, and the residuals $\hat{\bu}_+$.

\item[\textbf{[4]}] Reject $H_0$ if $\Wald = \hat{\bbeta}_p' \left(\widehat{\var} \hat{\bbeta}_p \right)^{-1} \hat{\bbeta}_p > q_{\chi_{p}^2}(1-\alpha)$,
where $\widehat{\var} \hat{\bbeta}_p$ is the OLS standard error of $\hat{\bbeta}_p$.
\end{itemize}
\end{algorithm}

\begin{remark}
In Algorithm \ref{alg:pdslm}, the choice among Step [5a] or [5b] does not affect the finite sample results of the test whenever the sample size $T$ is large enough. The small sample correction in [5b] \citep[see][]{kiviet1986rigour} has a wider practical applicability since [5a] suffers from size distortion in small samples, therefore in Section \ref{sec:sim} we always use [5b] for the Monte-Carlo simulations of the PDS-LA-LM test. Unreported simulations show that the Wald test performs very similarly to the LM test, if with slightly bigger size distortions.
\end{remark}

\begin{remark}\label{rmk_5}
The algorithm designed in HMS employs a lower bound on the penalty to ensure that in each selection regression at most $cT$ terms gets selected, for some $0<c<1$. Similarly, we also employ a $c=0.5$ lower bound on the selected variables in Algorithm \ref{alg:pdslm}. This ensures that in each equation the lasso does not select too many variables, as this would render the union too large and hence infeasible for post-least-squares estimation. Our PDS procedure does not require consistent model selection but only consistency (see Assumption \partref{as:hlev}{as:cons} below). Mistakes are allowed to occur in the selection: variables might be incorrectly included as long as the estimator remains sufficiently sparse and consistency is guaranteed. Note that even with the lower bound it remains possible that the lasso selects sufficiently distinct variables at every selection step. Should such a case occur where the number of selected variables $\hS$ is larger than the sample size, post-OLS would be infeasible. One way to avoid this would be to impose an ad-hoc increase on the tightness of the bound on the selected variables. Alternatively one could switch to a sparser estimator, such as the adaptive lasso.
\end{remark}

\begin{remark} \label{rmk:include_y}
Although it is not necessary for the theory to hold, we advocate to always include the $p$ lags of $y_t$ in the second stage regression. Erroneously omitting them might induce spurious regression, which is to be avoided. We achieved this by including them without penalty in the first stage regressions, such that they will be included in the active set by default. Additionally, we do not penalise the lags of $x_t$ in the first-stage regressions, to further reduce the probability of spurious regression and improving finite sample behaviour.
\end{remark}

\section{Theoretical Properties}\label{sec:theory}
In this section we present the main theoretical result. We show the post-selection, lag-augmented, least squares estimator $\hat \bbeta_{p}$ is asymptotically Gaussian uniformly over the parameter space. Therefore, tests for Granger causality are $\chi^2$ distributed. For the PDS-LA method to deliver uniformly valid inference, a set of assumptions is necessary. Among others, sparsity needs to be assumed on the high-dimensional vector $\bdelta$. Also, a restricted (sparse) eigenvalue condition needs to be assumed on the (scaled) Gram matrix, bounding away its smallest eigenvalue over a subset (cone) of $\mathbb{R}^K$. This condition essentially guarantees that over a sufficiently large subset of the parameter space the Gram matrix is well behaved and thus invertible. An empirical process bound is also required. This states that the (scaled) process $\bZ_{-p}'\bu$ uniformly concentrates around zero. 

These conditions are standard in the literature to prove the consistency of the lasso. However, when nonstationary processes are considered, more refined results are required. In particular, scaling becomes more complicated as time series of different orders need to be scaled with different rates. In addition, when time series are cointegrated, we first need to separate the stochastic trends from the stationary components before the appropriate scaling can be applied.

For example, suppose that $\bz_{t}$ can be written as the cointegrated system
\begin{equation*}
\Delta \bz_t = \bA \bB' \bz_{t-1} + \bu_t,
\end{equation*}
where $\bu_t$ is $I(0)$ and $\bA$ and $\bB$ are $n\times r$ matrices with $r < n$. Then define $\bzeta_t := \bz_t \bQ$ where $\bQ := (\bB, \bA_{\perp})'$ and $\bA_\perp$ is the $n \times n - r$ orthognal complement of $\bA$, such that $\bA_\perp' \bA = \bzero$. Then we can write
\begin{equation*}
\Delta \bzeta_t = \bQ \bA \bB' \bQ^{-1} \bzeta_{t-1} + \bQ \bu_t
=
\begin{bmatrix}
\bB' \bA & \bzero \\ \bzero & \bzero
\end{bmatrix} 
\bzeta_{t-1} + \bQ \bu_t.
\end{equation*}
Now the first $r$ variables in $\bzeta_t$ are $I(0)$, while the remaining ones are $I(1)$. Further details on constructing such a matrix $\bQ$ can be found in, e.g., \citet[Chapter 7]{lutkepohl2005new}.

Given the existence of such a matrix $\bQ$, we can without loss of generality assume that we may partition $\bzeta_t$ as $\bzeta_t = (\bzeta_{0,t}, \bzeta_{1,t}, \bzeta_{2,t})'$, where $\bzeta_{i,t}$ contains the $I(i)$ variables. We then consider a scaling matrix
\begin{equation} \label{eq:D_T}
\bD_{T} :=
\begin{bmatrix}
T^{1/2} \bI_{n_0} & \bzero & \bzero \\
\bzero & T^{-1} \bI_{n_1} & \bzero \\
\bzero & \bzero & T^{-2} \bI_{n_2}
\end{bmatrix},
\end{equation}
where $n_i$ corresponds to the number of $I(i)$ variables in $\bzeta_t$. Finally, we let $\bG_T := \bQ \bD_T$ as the final `scaling and rotation' matrix.

\begin{assumption}\label{as:hlev}
Let $\delta_T$ and $\Delta_T$ denote sequences such that $\delta_T, \Delta_T \rightarrow 0$ as $T \rightarrow \infty$. For $\bD_T$ as defined in \eqref{eq:D_T}, assume there exists a $Kp \times Kp$ matrix $\bG_{T,\bZ} = \bQ \bD_{T}$ conformably with $\bZ_{-p}$ such that the following holds:
\begin{enumerate}[(a)]
\item\label{as:ep} \textbf{Deviation Bound: } With probability at least $1-\Delta_T$ we have that the empirical process satisfes $\norm{\bG_T^{-1\prime} \bZ_{-p}'\bu}_{\infty}\leq \bar{\gamma}_T$, for the deterministic sequence $\bar{\gamma}_T$. 

\item\label{as:bound} \textbf{Boundedness:} Let $\bbeta$ in \eqref{eq:dgp_beta} be in the interior of a compact parameter space $\mathbb{B}\subset \mathbb{R}^K$.

\item\label{as:cons} \textbf{Consistency:} 
With probability $1-\Delta_T$, the following error bounds for $\boeta_j$ and $\hat{\boeta}_j$, defined in \eqref{eq:lasso_steps} and \eqref{eq:lasso}, hold:
\begin{equation*}
\begin{split}
&\norm{\bZ_{-p} \left(\hat{\boeta}_0 - \boeta_0 \right)}_2^2 \leq  \delta_T^2 T^{1/2}, \qquad \norm{\hat{\boeta}_0 - \boeta_0}_\infty \leq \delta_T T^{-1/4},\\ &\norm{\bZ_{-p,\backslash\{j\}} \left(\hat{\boeta}_j - \boeta_j \right)}_2^2 \leq  \delta_T^2 T^{1/2}, \quad j=1,\ldots,p.
\end{split}
\end{equation*}

\item\label{as:spar} \textbf{Sparsity:} Let $s = \abs{\S}$ and $\hS = \abs{\hS}$ denote the number of active variables in population and sample, respectively. Then with probability at least $1 - \Delta_T$, we have that $\max(s, \hS) \leq \bar{s}_T$ for the deterministic sequence $\bar{s}_T$. 

\item\label{as:rsev} \textbf{Restricted Sparse Eigenvalues:} 
For any $\boeta \in \mathbb{R}^{Kp}$ with $\norm{\boeta}_0 \leq \bar{s}_T$, we have with probability at least $1 - \Delta_T$
\begin{equation*}
\norm{\boeta }_1 \leq \sqrt{\bar{s}_T} \norm{\bZ_{-p} \bG_{T, \bZ}^{-1} \boeta}_2 / \kappa_{T,\min},
\end{equation*}
for the deterministic sequence $\kappa_{T,\min} > 0$.
\item\label{as:rates}\textbf{Rate Conditions:} the deterministic sequences bounding sparsity ($\bar{s}_T$), thickness of empirical process tails ($\bar{\gamma}_T$) and minimum eigenvalue ($\kappa_{T,\min}$) must satisfy
\begin{equation} \label{rates_interplay}
T\frac{\bar{s}_T\bar{\gamma}_T}{\kappa_{T,\min}}\leq \delta_T.
\end{equation}
\end{enumerate}
\end{assumption}

Condition (\ref{as:ep}) bounds the empirical process with high probability. Such deviation bounds for non-stationary time series can be found in, among others, \citet[Lemma A.3]{smeekes2021automated}, \citet[Proposition 1]{mei2022lasso} and \citet[Lemma 2]{wijler2022restricted}. Condition \pref{as:bound} is standard and assumes compactness of the parameter space of the vector $\boldsymbol{\beta}$ which in turn implies the boundedness.

Condition \pref{as:cons} supposes consistency of the first-stage estimator. Consistency is thereby established using an error bound or ``oracle inequality". Such inequality gives the rate of convergence of the estimator as a function of: the tuning parameter $\lambda$, the deviation bound, the cardinality $s$ of the active set and the restricted (sparse) eigenvalue. As such, this condition is intimately related to \pref{as:ep}, \pref{as:spar} and \pref{as:rsev}. For stationary time series many results exist under a variety of settings; see e.g.~\citet[Theorem 1]{masini2022regularized} for VAR models or \citet[Corollary 1]{adamek2022lasso} for general time series models. Nonstationary time series are treated in \citet[Theorem 2]{smeekes2021automated}, \citet[Theorem 1 and 2]{mei2022lasso} and \citet[Theorem 2]{wijler2022restricted}.

Condition \pref{as:spar} requires sparsity of the population parameters and the estimator. Sparsity of the first-stage estimator is needed in our framework as we perform OLS on the selected variables from the first-stage regressions. If the selected variables are not sparse enough, too many variables will be selected for OLS to be feasible. For simplicity we work under the assumption of exact sparsity; however, at the expense of more complicated notation and proofs this can be relaxed to approximate sparsity following \citet{belloni2014inference}, where it is assumed that the exact sparse model is a (good) approximation to the true DGP, or weak sparsity following \citet{adamek2022lasso}, where many non-zero but small coefficients are allowed.

Condition \pref{as:rsev} requires that for sufficiently sparse vectors, the eigenvalues of the subset of the Gram matrix corresponding to their nonzero support do not decrease to zero too fast. While this is a standard (and easily verifiable) assumption for stationary time series (see e.g., \citealp[Lemma A.3 and A.5]{adamek2022lasso} and \citealp[Lemma 3 and Proposition 2]{masini2022regularized}), time series with unit roots require more care. Indeed,  \citet[Theorem B.2]{smeekes2021automated}, \citet[Lemma 2]{mei2022lasso} and \citet[Theorem 1]{wijler2022restricted} show for unit root regressors that the standard rate of $T^{-2}$ coming from the $\bD_T$ scaling still results in at least a factor of $\bar{s}_T^{-1}$ in $\kappa_{T,\min}$. As we allow $\kappa_{T,\min}$ to depend on the sample size, this can be accommodated, as long as the interplay of the rates in condition \pref{as:rates} is satisfied. This condition links the sparsity $\bar{s}_T$, the tails thickness of the empirical process $\bar{\gamma}_T$ and the minimum eigenvalue $\kappa_{T,\min}$. The exact rates allowed for are a compromise between the number of moments existing, the strength of the dependence, the growth rate of the dimension $K$ and the sparsity of the parameters.\footnote{For an illustration of the complexities of these relations, we refer to Figure C.1 in \citet{adamek2022lasso} which visualises the feasible combinations with regards to lasso consistency.}

We can now state our first, and main, theoretical result. This establishes that the first stage regressions allow for sufficiently accurate variable selection and estimation such that the second stage is not affected by the variable selection performed.
\begin{theorem}\label{th:pds_cons}
Let $\hat{\bbeta}_p$ denote the OLS estimator of the first $p$ elements of $\bbeta_+$ (equal to $\bbeta$) in \eqref{eq:sec_stage_hatS}, and $\tilde{\bbeta}_p$ the corresponding estimator in \eqref{eq:sec_stage_trueS}. Then uniformly over a parameter space $\mathcal{B}$ for which Assumptions \ref{ass:basic} and \ref{as:hlev} holds for all elements in $\mathcal{B}$, we have that
\begin{equation*}
\sqrt{T} (\hat{\bbeta}_{+} - \bbeta) = \sqrt{T} (\tilde{\bbeta}_{+} - \bbeta) + o_p(1).
\end{equation*}
\end{theorem}

Theorem \ref{th:pds_cons} establishes the asymptotic equivalence of the estimator in the feasible second-stage regression \eqref{eq:sec_stage_hatS} based on the estimated active set $\hS$, and the estimator in the infeasible second-stage regression \eqref{eq:sec_stage_hatS} based on the true, unobserved, active set $\S$. Note that the equivalence is only established for the estimators of the coefficients of the lags of the Granger causing variable -- minus the augmented lags -- this is however all that is needed to establish the validity of the Granger causality tests. Before stating the result about the asymptotic distribution, we need another assumption on the asymptotic behaviour of the relevant variables. For a finite number of relevant variables -- measured by $s$ in Assumption \partref{as:hlev}{as:spar} -- this assumption follows directly from well-known results in the unit root and cointegration literature; see Appendix \ref{sec:verifyUR}. We here state the assumption more generally to also accommodate sparsity that increases with the sample size. 

\begin{assumption} \label{ass:URproperties}
Let $\{\bzeta_t\}_{t=1}^T$ denote an $n$-dimensional process, where $n=n_T$ may increase with $T$. Let $n_0$, $n_1$ and $n_2$ denote integers such that $n = n_0 + n_1 + n_2$ and partition $\bzeta_t = (\bzeta_{0,t}, \bzeta_{1,t}, \bzeta_{2,t})'$ comformably with $(n_0, n_1, n_2)$, and assume that $\E [\bzeta_t u_t] = \bzero$ for all $t=1,\ldots,T$. Define 
\begin{equation*}
\bD_{T} :=
\begin{bmatrix}
T^{1/2} \bI_{n_0} & \bzero & \bzero \\
\bzero & T \bI_{n_1} & \bzero \\
\bzero & \bzero & T^{2} \bI_{n_2}
\end{bmatrix}
=:
\begin{bmatrix}
\bD_{T, 0} & \bzero & \bzero \\
\bzero & \bD_{T, 1} & \bzero \\
\bzero & \bzero & \bD_{T, 2}
\end{bmatrix},
\end{equation*}
and let $\Delta_T, \delta_T$ denote deterministic sequences such that $\Delta_T, \delta_T\to 0$ as $T\to\infty$. For varying subscripts `$\cdot$' given below, let $\phi_{T,\cdot}$, $\kappa_{T,\cdot}$ and $\gamma_{T,\cdot}$ denote sequences of positive real numbers, such that with probability $1 - \Delta_T$ the following statements hold jointly:
\begin{enumerate}[(a)]
\item $\norm{\bD_{T, i}^{-1} \sum_{t=1}^T \bzeta_{i,t} \bzeta_{j,t}' \bD_{T,j}^{-1}}_2 \leq \phi_{T,ij}$ for $i,j = 0,1, 2$ and $i < j$;
\item $\lambda_{\min} \left(\bD_{T, 0}^{-1} \sum_{t=1}^T \bzeta_{0,t} \bzeta_{0,t}' \bD_{T,0}^{-1}\right) \geq \kappa_{T,0}$ and $\lambda_{\min} \left(\bD_{T, -0}^{-1} \sum_{t=1}^T \bzeta_{-0,t} \bzeta_{-0,t}' \bD_{T,-0}^{-1}\right) \geq \kappa_{T,12}$, where $\bzeta_{-0,t} = (\bzeta_{1,t}', \bzeta_{2,t}')'$ and $\bD_{t,-0} = \text{diag}(\bD_{T,1}, \bD_{T,2})$;
\item $\norm{\bD_{T, i}^{-1} \sum_{t=1}^T \bzeta_{i,t} u_t}_{2} \leq \gamma_{T,u,i}$ for $i = 0,1, 2$;
\item $\norm{T^{-1} \sum_{t=1}^T \left[\bzeta_{0,t} \bzeta_{0,t}' - \E \left(\bzeta_{0,t} \bzeta_{0,t}' \right) \right]}_2 \leq \delta_{T}$.
\end{enumerate}
Assume that the rates above are bounded as
\begin{enumerate}[(i)]
\item $\kappa_{T,12} - (\phi_{T,01}^2 + \phi_{T,02}^2) / \kappa_{T,0} \geq \mu_{T,(i)}$;
\item $\phi_{T,01} + \phi_{T,02} + 2(\phi_{T,01}^2 + \phi_{T,02}^2) / \kappa_{T,0} \leq \mu_{T,(ii)}$;
\item $\gamma_{T,u,1} + \gamma_{T,u,2} + (\phi_{T,01} + \phi_{T,02}) \gamma_{T,u,0} / \kappa_{T,0} \leq \mu_{T,(iii)}$;
\end{enumerate}
for sequences $\mu_{T,(i)}, \mu_{T,(ii)}, \mu_{T,(iii)}$ that satisfy
the condition
\begin{equation} \label{eq:rates_mu}
\mu_{T,(ii)} (\mu_{T,(ii)} + \mu_{T,(iii)}) \leq \delta_T \mu_{T,(i)}.
\end{equation}
In addition, let $\bR_m$ denote a deterministic $m \times n_0$ matrix where $m<\infty$ is not depending on $T$. Then we have that
\begin{equation*}
T^{-1/2} \sum_{t=1}^T \bR_m \bzeta_{0,t} u_t \xrightarrow{d} N(\bzero, \bOmega),
\end{equation*}
where $\bOmega = \lim_{T\rightarrow\infty} T^{-1} \bR_m \E(\bzeta_{0,t} u_t u_t' \bzeta_{0,t}') \bR_m' = \sigma_u^2 \bR_m \E(\bzeta_{0,t} \bzeta_{0,t}') \bR_m'$.
\end{assumption}

Assumption \ref{ass:URproperties} appears rather abstract with the sequences $\phi_{T,\cdot}$, $\kappa_{T,\cdot}$ and $\gamma_{T,\cdot}$ bounded in a complicated nonlinear way. However, note that the actual assumption that is required is that if we combine the blocks in (a)-(d) appropriately -- first through (i)--(iii), then via \eqref{eq:rates_mu} -- the terms become negligible. If we assume that the number of variables $n$ is finite, which in our setting follows from the sparsity $s$ not growing with the sample size, we can express $\phi_{T,\cdot}$, $\kappa_{T,\cdot}$ and $\gamma_{T,\cdot}$ in well-known rates needed to obtain limiting distributions of sums of products of integrated variables, see e.g.~\citet{hamilton1994time} for an overview. We show in Appendix \ref{sec:verifyUR} how Assumption \ref{ass:URproperties} can be verified in this case. 

One way to extend the result to a growing number of variables would be through a Gaussian approximation theorem; such an approach is considered in, e.g., \citet[Remark 3.4]{zhang2019identifying} and \citet[Theorem B.3]{smeekes2021automated}. This is a relatively crude approach that puts significant limitations on the allowed growth rate of $s$. However, as the growth rate of $s$ is anyway only allowed to be limited compared to the sample size $T$, such an approach would be feasible here without imposing strong additional restrictions (unlike in the general case such as for establishing a result like the deviation bound \partref{as:hlev}{as:ep} where it would significantly restrict the allowed growth rate of $K$). With this assumption in place we can now state our second theoretical result, which establishes the limit distribution of the Granger causality tests.

\begin{theorem}\label{th:asyGC}
Let $\Wald$ and $\LM$ be as defined in Algorithm \ref{alg:pdslm}. Assume that Assumption \ref{ass:URproperties} holds for $\bzeta_t = \bQ_{\S} \bv_{+,\S,t}$, where $\bv_{+,\S,t} = (x_{t-p-1}, \ldots, x_{t-p-d}, \bv_{\S,t})$. Then, uniformly over a parameter space $\mathcal{B}$  on which Assumptions \ref{ass:basic} and \ref{as:hlev} holds for all elements in $\mathcal{B}$, we have that
\begin{align*}
&\LM, \Wald \xrightarrow{d} \chi_p^2, \qquad \text{as } T \rightarrow \infty,
\end{align*}
under the null hypothesis that $\bbeta = \bzero$.
\end{theorem}

\begin{remark} \label{rmk:hetrobust}
Though we focus on homoskedastic error terms for simplicity, heteroskedasticity-robust versions of the test can easily be constructed and shown to be valid, as this would only affect the low-dimensional part of our results in Theorem \ref{th:asyGC}. Standard techniques can therefore be used for constructing heteroskedasticity-robust tests: for the Wald test, the OLS standard errors can be replaced by Eicker-White standard errors, while the LM test can be modified as in \citet{wooldridge1987regression}. We refer to HMS, Algorithm 2 for a full treatment. 
\end{remark}

\section{Monte-Carlo Simulations}\label{sec:sim}
We now evaluate the finite-sample performance of our proposed PDS-LA-LM Granger causality test. We consider the following Data Generating Processes (DGPs) in first differences inspired by \citet{kock2015oracle}:
\small
\begin{align*}
&\text{DGP1:}\qquad \Delta \bz_t=\begin{bmatrix}
0.5 & 0 & \ldots &0 \\
0 & 0.5&\ldots &0\\
\vdots &\vdots & \ddots & \vdots  \\
0&0&\ldots & 0.5 
\end{bmatrix}\Delta \bz_{t-1} + \bu_t,\\
&\text{DGP2:}\qquad \Delta \bz_t=\begin{bmatrix}
(-1)^{|i-j|}a^{|i-j|+1}  & \ldots &(-1)^{|i-j|}a^{|i-j|+1} \\
(-1)^{|i-j|}a^{|i-j|+1} &\ldots &(-1)^{|i-j|}a^{|i-j|+1}\\
\vdots  & \ddots & \vdots  \\
(-1)^{|i-j|}a^{|i-j|+1}&\ldots & (-1)^{|i-j|}a^{|i-j|+1} 
\end{bmatrix}\Delta \bz_{t-1} +\bu_t,
\end{align*}
\normalsize
with $a=0.3$.
The diagonal VAR(1) for DGP1 allows the sparsity assumption to be met. Instead, for DGP2 the coefficients decrease with exponential pace departing from the main diagonal and hence although the farthest coefficients are small, the exact sparsity assumption is not met. We report simulations for Granger causality tests from the first variable to the second variable. Therefore, also when integrated out to nonstationary, DGP1 automatically satisfies the null of no Granger causality from unit 2 to 1, however DGP2 does not. Therefore, for the power analysis of both DGP1 and DGP2 we set the coefficient in position $(2,1)$ equal to 0.2. We set the same coefficient equal to zero for DGP2 for the size analysis. 

Simulations are reported for different types of covariance matrices of the error terms. We employ a Toepliz-version for calculating the covariance matrix as $\Sigma_{i,j}=\rho^{|i-j|}$, where $(i,j)$ refer to row $i$, column $j$ of the matrix $\Sigma_{u}$. We cover two scenarios of correlation: $\rho=(0,0.7)$.

The lag length is fixed to $p=2$, while we employ a double ($d=2$) augmentation of the  Granger causing variable. Having $p\geq 2$ guarantees that no spurious results occur in the selection steps. Following the recommendation in HMS, we employ the BIC in selecting the tuning parameter $\lambda$ for the lasso.

Table \ref{tab:1} reports the size and power of the PDS-LA-LM test out of 1000 replications. We use different combinations of time series length $T=(50,100,200,500,1000)$ and number of variables in the system $K=(10,20,50,100)$. All the rejection frequencies are reported using a burn-in period of fifty observations.


Our PDS-LA-LM test shows good performance in terms of size and (unadjusted) power for all DGPs considered. The setting of no correlation is handled remarkably well by all DGPs and only moderate size distortion is visible in large systems for small samples. Whenever high correlation of errors is present, sizes are still in the vicinity of 5\% for DGP1 where the sparsity assumption is met. However, we notice how for DGP2, for which the sparsity assumption is not met, some residual size distortion remains visible even in large systems. However, the power of the test is always increasing with the sample size $T$ for all the considered cases.

\begin{remark}\label{rmk_bound}
As mentioned in Remark \ref{rmk_5}, in order to obtain the results for the size and power when $T\leq Kp$ we need to impose a lower bound on the lasso penalty $\lambda$ which guarantees to select at most $c\, T$ variables in each relevant equation of the VAR, for some $0 < c < 1$. The bound should be set as strict as the system requires and often there is not a universal constant $c$ that works in all settings, therefore this choice needs to be adaptive. For instance, if the lag length is $p=2$, this implies 3 selection steps plus $d$ augmented lags of Granger causing variables to be added. In some cases this might lead to too many variables being selected in order to perform least squares in the second step. In these cases we tighten the bound using either $c=0.33$ or $c=0.25$.
\end{remark}

\begin{center}
\begin{threeparttable}
\caption{Simulation results for the PDS-LA-LM Granger causality test}\label{tab:1}
\centering\small\singlespacing
\begin{tabular}{ccrccccccccccc}
\toprule
& & \multicolumn{1}{r}{$T$} & 50 & 100 & 200 & 500 & 1000 & 50 & 100 & 200 & 500 & 1000\\
\cmidrule(l){3-3} \cmidrule(l){4-8} \cmidrule(l){9-13} 
DGP & $\rho$ & \multicolumn{1}{l}{$K$} & \multicolumn{5}{c}{Size} & \multicolumn{5}{c}{Power}  \\
\cmidrule(l){1-3} \cmidrule(l){4-8} \cmidrule(l){9-13} 
1 &  0 &  10 & 10.4 & 7.2 & 5.5 & 4.2 & 5.5 & 23.3 & 41.0 & 74.2 & 99.4 & 100 \\
  &    &  20 & 12.5 & 9.5 & 7.3 & 6.3 & 5.9 & 19.7 & 40.1 & 74.6 & 99.5 & 100 \\
  &    &  50 & 10.7 & 7.1 & 7.2 & 5.4 & 4.7 & 17.7 & 38.9 & 66.5 & 99.2 & 100 \\
  &    & 100 &  9.9 & 8.4 & 6.8 & 6.1 & 5.3 & 16.0 & 34.1 & 68.1 & 98.5 & 100 \\
\cmidrule(l){1-3} \cmidrule(l){4-8} \cmidrule(l){9-13} 
2 &    &  10 &  9.0 & 6.2 & 4.7 & 4.2 & 5.4 & 19.7 & 36.0 & 69.3 & 99.2 & 100\\
  &    &  20 & 13.3 & 8.1 & 5.6 & 5.5 & 6.1 & 17.3 & 36.9 & 69.4 & 98.5 & 100\\
  &    &  50 &  8.6 & 7.3 & 6.9 & 5.9 & 4.7 & 18.2 & 34.4 & 63.9 & 98.8 & 100\\
  &    & 100 &  8.6 & 8.2 & 7.5 & 6.4 & 5.4 & 13.9 & 31.4 & 62.4 & 98.0 & 100\\
\cmidrule(l){1-3} \cmidrule(l){4-8} \cmidrule(l){9-13} 
1 & .7 &  10 & 13.6 &  9.4 &  6.2 &  5.0 & 5.5 & 20.9 & 20.9 & 44.1 & 87.5 & 99.3 \\
  &    &  20 & 12.8 & 11.4 &  9.1 &  6.2 & 5.8 & 17.8 & 23.7 & 41.2 & 84.9 & 99.4 \\
  &    &  50 & 11.7 &  9.6 & 10.3 &  7.2 & 5.2 & 15.5 & 22.6 & 39.8 & 80.5 & 99.2 \\
  &    & 100 & 12.6 & 12.4 &  7.2 & 10.1 & 6.3 & 13.3 & 22.2 & 37.6 & 72.7 & 98.7 \\
\cmidrule(l){1-3} \cmidrule(l){4-8} \cmidrule(l){9-13} 
2 &    &  10 & 12.1 &  9.0 & 6.0 & 4.8 & 6.1 & 17.3 & 19.7 & 40.1 & 84.1 & 98.3 \\
  &    &  20 & 12.4 &  8.8 & 8.0 & 6.4 & 5.9 & 15.9 & 19.9 & 39.0 & 81.6 & 98.8 \\
  &    &  50 & 11.7 &  9.6 & 9.3 & 7.4 & 5.6 & 14.6 & 19.4 & 34.0 & 74.9 & 98.2 \\
  &    & 100 & 11.1 & 10.1 & 9.1 & 9.5 & 5.7 & 12.8 & 21.0 & 33.3 & 66.0 & 97.9 \\
\bottomrule
\end{tabular}
\begin{tablenotes} 
\scriptsize 
\item Notes: Size and Power for the different DGPs are reported for 1000 replications. $T=(50,100,200,500)$ is the time series length, $K=(10,20,50,100)$ the number of variables in the system, the lag length is fixed to $p=2$ and BIC is used to select the tuninig parameter for the lasso. $\rho$ indicates the correlation employed to simulate the time series with the Toeplitz covariance matrix.  
\end{tablenotes}
\end{threeparttable}
\end{center}

\section{Lag Length Selection}\label{sec:choicep}
Up until this point, we considered the lag length $p$ as given. In reality, this is typically not the case. In this section we propose a simple, data-driven method to estimate $p$. Standard techniques for tuning the lag length such as information criteria or sequential testing fail when applied directly to the high-dimensional VAR. The standard approach with penalized regression methods would be to set $p$ as a generous upper bound, and let the method decide which lags are needed. This however provides complications for our approach as a large $p$ means many first-stage regressions need to be done, with the potential of selecting too many variables. On top of that, we need to augment the second stage with $d$ additional lags.

It is therefore important for our approach to have a reasonable data-driven selection of the lag length. We are essentially looking for an `informative upper bound' on the lag length; while mild over-specification is not a problem, under-specification breaks the lag augmentation, and must be avoided. We therefore base selection on univariate autoregressions for all time series in our dataset, and we apply an information criterion to those. This approach is motivated through the final equations representation of a VAR model, which implies that a VAR($p$) with $K$ variables generates individual ARMA models with maximal orders $(Kp,(K-1)p)$. As such, lag length selection based on individual autoregressions is likely to yield lag lengths larger than $p$ \citep{zellner1974time,cubadda2009studying}. On the other hand, \cite{cubadda2009studying} observed that individual autoregressions often need much smaller lag lengths than the large orders implied by the final equations representation, thus making it plausible that the orders found are not overly conservative.

We implement this method as follows. First, we estimate the autoregressions
\begin{equation*}
z_{i,t} = \sum_{j=1}^p \gamma_{j,p} z_{i,t-j} + \varepsilon_{i,p,t}, \qquad i = 1, \ldots, K,
\end{equation*}
by OLS. Then, letting $\hat\omega_i = \frac{1}{T} \sum_{t=1}^T \hat{\varepsilon}_{i,p,t}^2$, we set $\hat{\bOmega}_p = \text{diag}(\hat\omega_1, \ldots, \hat\omega_K)$ and choose the $p$ that minimizes
\begin{equation*}
\text{IC}^*(p) = \ln\left(\det \hat{\bOmega}_p \right) + C_T \frac{p K}{T} = \sum_{i=1}^N \log \hat\omega_i + C_T \frac{p K}{T},
\end{equation*}
where $C_T$ takes the standard values for well-known information criteria, e.g.~$C_T = \ln T$ yields BIC$^*$, while $C_T = 2$ yields AIC$^*$. Note that next to the dimensionality reduction offered by running individual autoregressions instead of a VAR, further reduction is achieved by letting $\hat{\bOmega}$ be diagonal.

We investigate the performance of this method in a small simulation study. We simulate VAR models using DGP2 with $\rho=0$ as in Section \ref{sec:sim} for the combinations of $K=(10,20,50,100)$ and $T=(50,100,200,500,1000)$.\footnote{We investigated DGP1 and $\rho=0.7$ as well. The method showed similar or better performance there. Results are available on request.} Knowing the true value of $p=2$, we apply the model selection procedure on a grid of 10 values for $p$. 

\begin{figure}
\centering
\includegraphics[width = \textwidth]{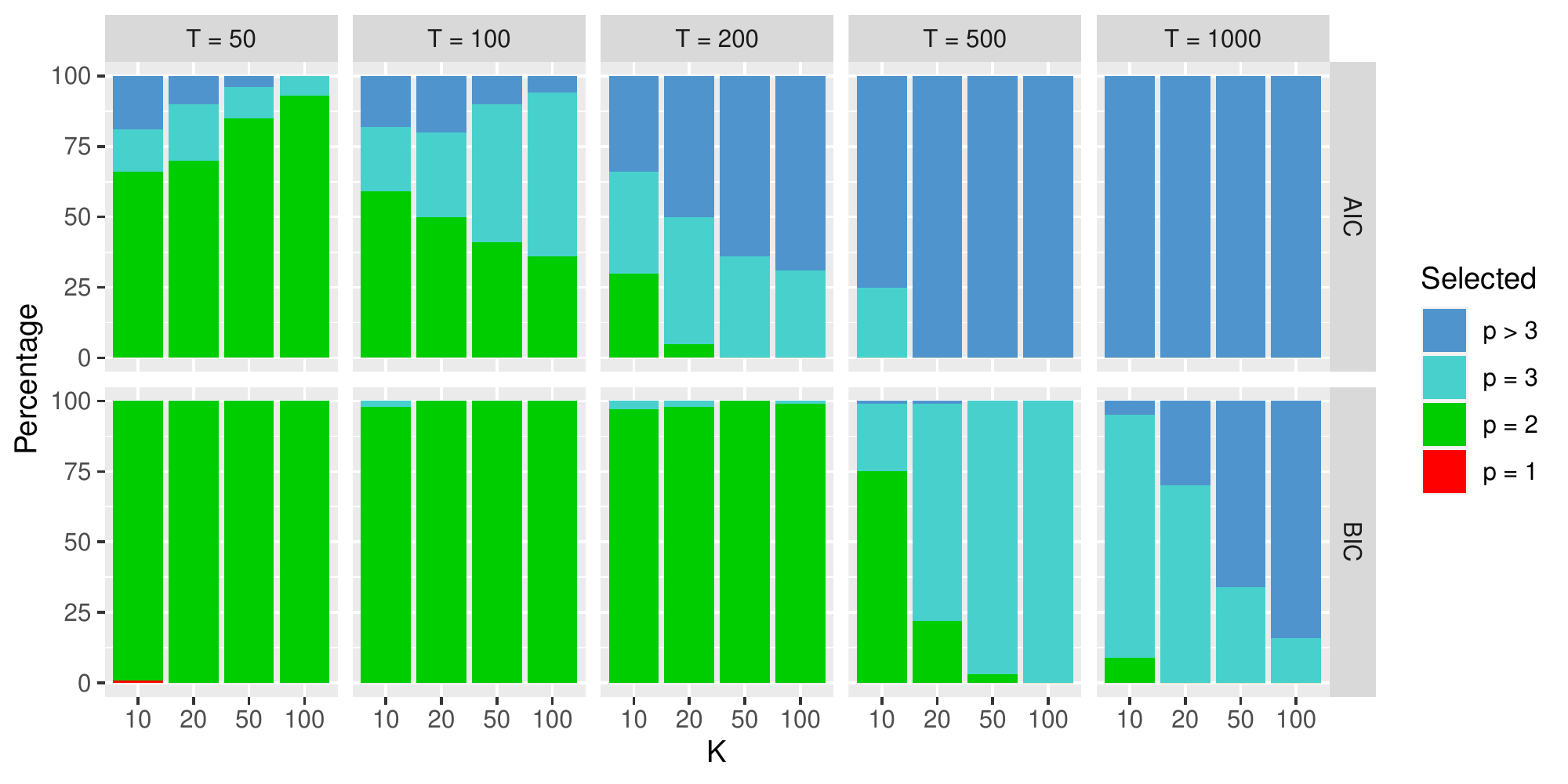} 
\caption{Lag length selection frequencies for AIC$^*$ and BIC$^*$}\label{fig:lagselection}
\end{figure}

We show the results in Figure \ref{fig:lagselection} for AIC$^*$ and BIC$^*$. Both criteria succeed in barely ever underestimating the lag length. Moreover, BIC$^*$ is quite accurate, only overestimating the lag length by more than one when $T=1000$, where larger values of $p$ are not problematic. AIC$^*$ is more liberal but still performs well in smaller samples, where it matters most. This simple way of selecting lag lengths therefore works remarkably well in finite examples, and is therefore recommended in combination with the PDS-LA method.

\section{Empirical Application: Causes and Effects of Economic Uncertainty}\label{sec:emp_appl}

The role of economic uncertainty is a hot topic in macroeconomics. In particular, it is a matter of much debate whether economic uncertainty should be seen as an exogenous shock, causing economic conditions such as business cycle fluctuations, or an endogenous response to economic fundamentals \citep{ludvigson2021uncertainty}. A complicating issue is that there is no universal definition or measurement of uncertainty. In this application we address both issues: we investigate whether uncertainty (Granger) causes or is caused by other variables in the economy, and whether different measurements contain the same information.

Since the seminal paper of \citet{bloom2009impact} economic uncertainty has been at the forefront of macroeconomic debate. \citet{bloom2009impact} built a measure of uncertainty from the Chicago Board Options Exchange (CBOE) S\&P 100 Volatility Index (VXO), arguing that stock market volatility expectations are a good proxy for overall economic uncertainty. The FRED-MD dataset of \citet{mccracken2016fred} has introduced the series VXOCLSx in its September 2015 vintage. This series is constructed by splicing a synthetic historical VXO series obtained from Nicholas Bloom’s website and VXOCLS from FRED.\footnote{In its December 2021 vintage, FRED-MD removed VXOCLSx and replaced it with VIXCLSx as the former has been discontinued from the source. VIX is a similar measure of implied volatility based on the S\&P 500. While I made analysis focus on a pre-Covid vintage incorporating VXOCLSx, in Appendix \ref{sec:vix} We extend our analysis to a recent post-Covid vintage with VIXCLSx included.}

An alternative, very popular measure of uncertainty was proposed by  \citet{baker2016measuring}, who constructed the Economic Policy Uncertainty (EPU) index based on newspaper coverage frequency. This accounts for the frequency with which certain strings of keywords related to economic uncertainty appear in the 10 leading U.S.~newspapers. The Federal Reserve Economic Division (FRED) itself issues a Global Economic Policy Uncertainty Index (GEPUCURRENT) which is a GDP-weighted average of national EPU indices for 20 countries. However, the FRED-MD dataset does not contain any EPU uncertainty index. This naturally raises the question whether EPU can serve as a better measure of economic uncertainty within FRED-MD or VXOCLSx is ``good enough". If EPU and VXOCLSx measure the same content, one would expect that adding EPU to FRED-MD would not lead to much predictive power for EPU if VXOCLSx is accounted for in the information set.

Our PDS-LA method can thus be used in this context with a twofold purpose. First, by estimating Granger causal relationships between FRED macroeconomic variables and the uncertainty index and counting the significant number of outgoing (index$\to$FRED) and incoming (FRED$\to$index) links, we contribute in the understanding of whether such indexes are respectively exogenous sources of business cycle (index$\to$FRED) or rather endogenous responses to economic fundamentals (FRED$\to$index). Second, by investigating the links between a second index (EPU) and the FRED variables \emph{conditional} on VXOCLSx in the dataset, we can investigate whether EPU measures information not already present in VXOCLSx.

In our analysis we are particularly going to pay attention to the effect of how to treat potential nonstationarity in the data. The FRED-MD dataset comes with a detailed appendix \citep{mccracken2016fred} and Matlab/R routines which allow not only to clean the data from missing values and outliers, but also to take the appropriate transformations to render all the time series stationary. This presents the practitioner with an easy and popular solution to deal with non-stationarity, however there are two potential issues with this approach. First, not every time series can be clearly categorized regarding the order of integration, and there are several variables for which a case can be made for two different orders. Indeed, as illustrated by \citet{SmeekesWijler20} and \citet{SmeekesWilms20}, unit root tests may give ambiguous results and depending on the type of test employed, a different classification may arise. Second, even with a correct classification, differenced time series loose long-run information on any cointegrating relations, which may change the Granger causal relations. Our method provides an alternative way of performing the tests directly on the levels of the data, thereby avoiding the issue of transformations altogether.\footnote{To run these analyses we used the authors R package \textit{HDGCvar} available at \url{https://github.com/Marga8/HDGCvar}}.

For the Economic Policy Uncertainty index we gather the data directly from \citet{baker2016measuring}.\footnote{The data is freely available on the website of Scott R. Baker, Nick Bloom and Steven J. Davis at \url{https://www.policyuncertainty.com/index.html}.} Specifically, we use the three component index which combines (i) news coverage about policy-related economic uncertainty, (ii) tax code expiration data and (iii) economic forecaster disagreement.\footnote{We also repeated the analysis using the News-based Policy Uncertainty index as described in \citet{baker2016measuring} but the results are nearly identical to those reported here.} For details on how the index is computed we refer to the given reference. As the EPU index is only available starting January 1985, we accordingly use the FRED-MD data from January 1985. We split the analysis considering two different endpoints of the sample. First we consider the series until November 2019, thus intentionally excluding from the sample both the Covid pandemic and the war in Ukraine, but including the 2008 financial crisis. This makes for a total of 117 variables (when EPU is included) and 419 observations. In Appendix \ref{sec:vix} we extend the analysis to a recent vintage which includes the aforementioned crises and spans until September 2022.  

First, we do not add yet US-EPU to the information set and in Figure \ref{figura_vox1} and \ref{figura_vox2} we instead just investigate the predictive relations to and from VXOCLSx with all the macroeconomic series of FRED-MD using PDS-LA-LM.
\begin{figure}
\begin{minipage}{0.5\textwidth}
\centering
\includegraphics[width=0.95\textwidth, trim = {3.4cm 3cm 2.2cm 2.6cm},clip]{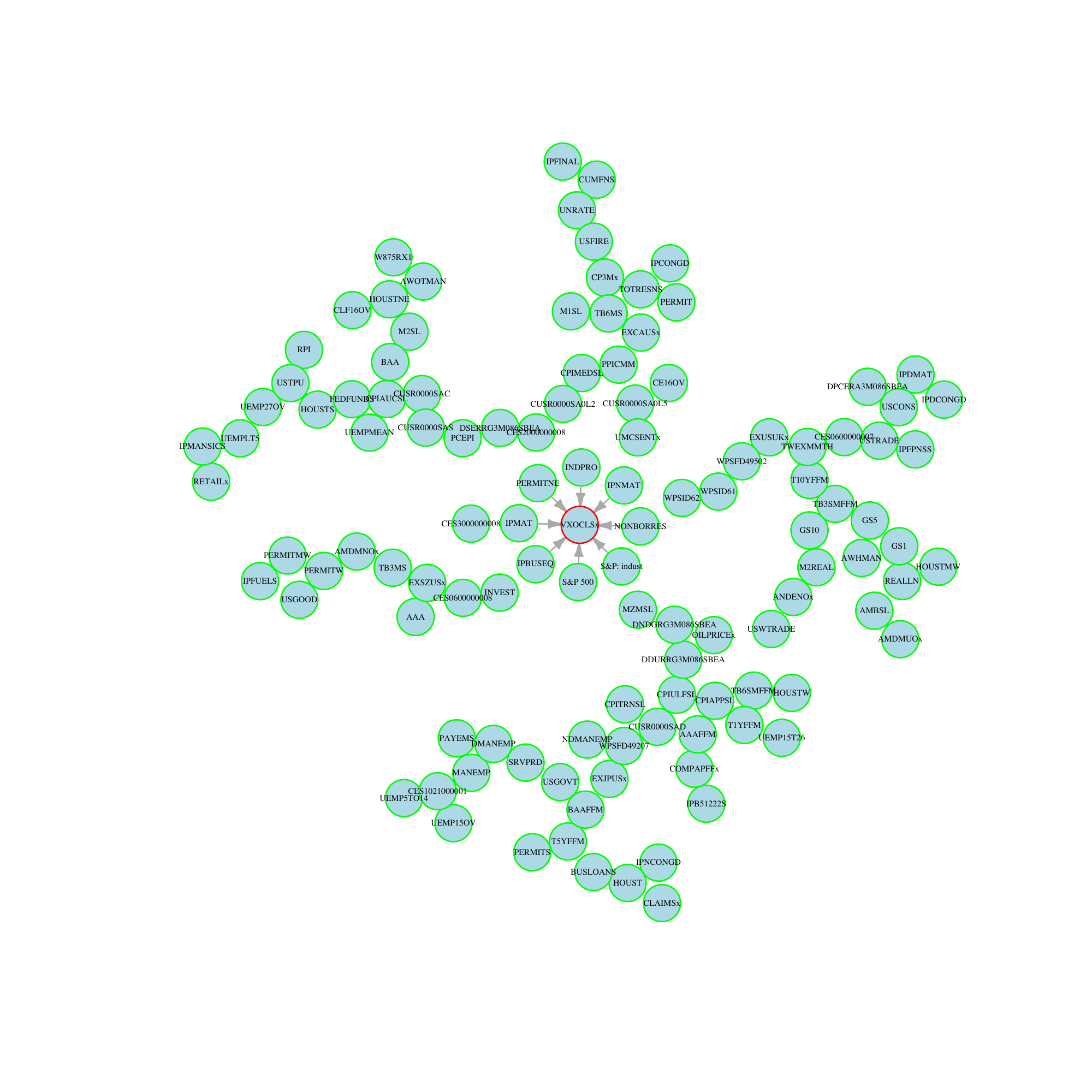}
\caption{PDS-LA-LM, Causes of VXOCLSx, $\alpha=0.01$}
\label{figura_vox1}
\end{minipage}
\begin{minipage}{.5\textwidth}
        \centering
\includegraphics[width=0.95\textwidth, trim = {3.4cm 3cm 2.2cm 2.6cm},clip]{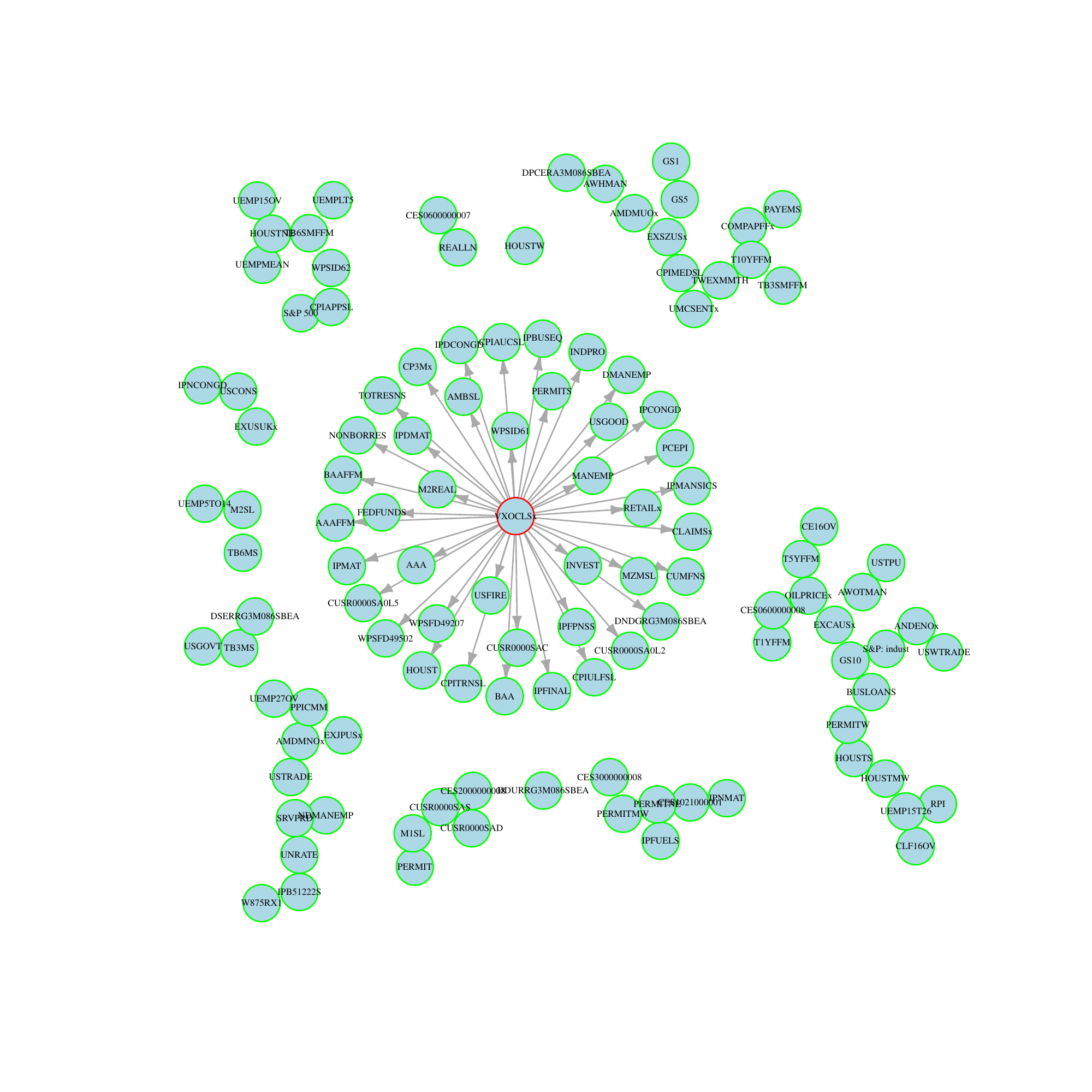}
\caption{PDS-LA-LM, Caused by VXOCLSx, $\alpha=0.01$}
\label{figura_vox2}
\end{minipage}
\end{figure}

At significance level $\alpha=0.01$, eight macroeconomic series are found to Granger cause VXOCLSx while the other way around VXOCLSx Granger-causes $41$ macroeconomic series. Before adding US-EPU to the dataset and investigate the changes, we repeat the analysis, this time applying the recommended stationary transformations (ST) from FRED-MD and we use the PDS-LM test of HMS instead of PDS-LA-LM to investigate the same relations.
\begin{figure}
\begin{minipage}{0.5\textwidth}
\centering
\includegraphics[width=0.95\textwidth, trim = {3.4cm 3cm 2.2cm 2.6cm},clip]{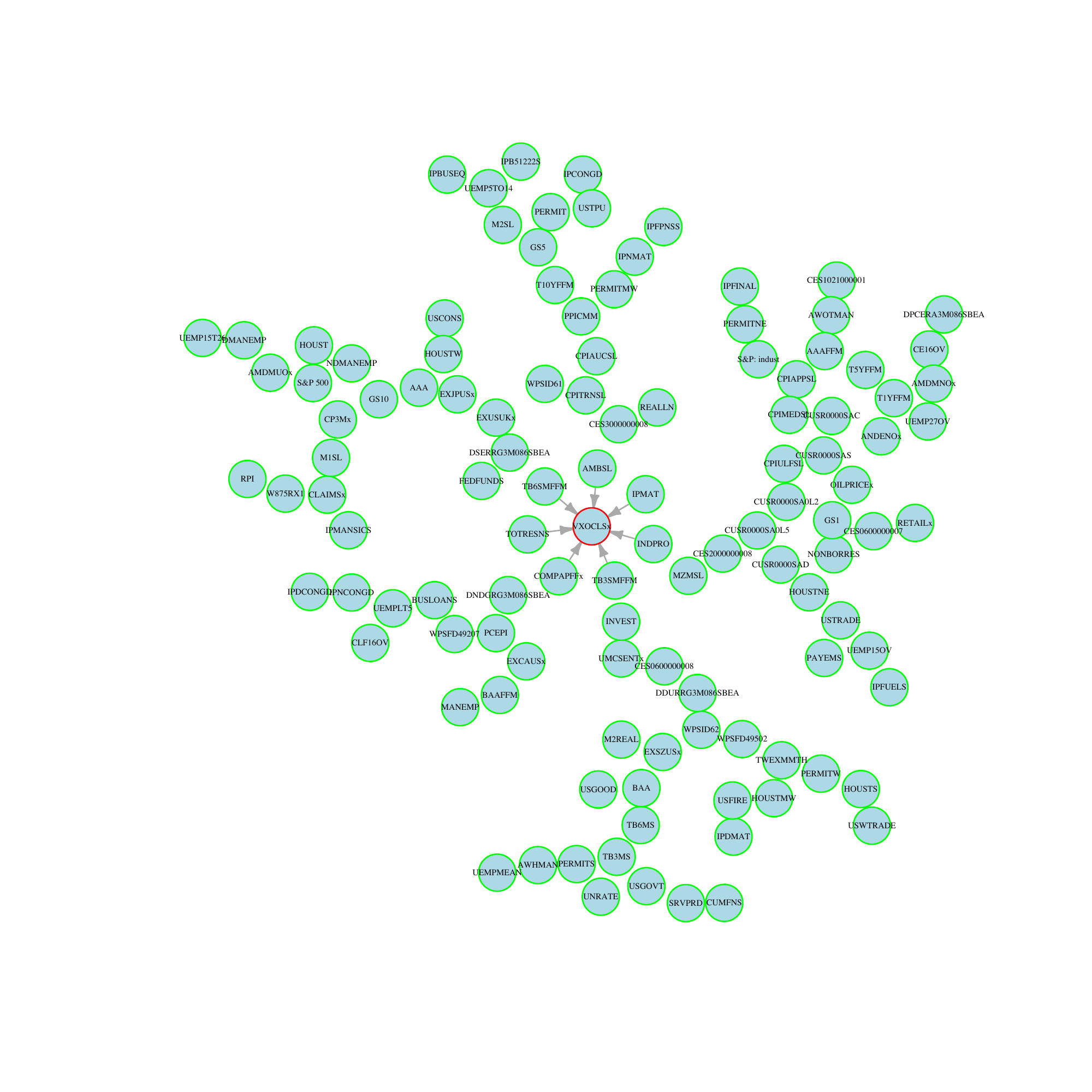}
\caption{PDS-LM, Causes of VXOCLSx, $\alpha=0.01$}
\label{figura_vxostat1}
\end{minipage}
\begin{minipage}{.5\textwidth}
        \centering
\includegraphics[width=0.95\textwidth, trim = {3.4cm 3cm 2.2cm 2.6cm},clip]{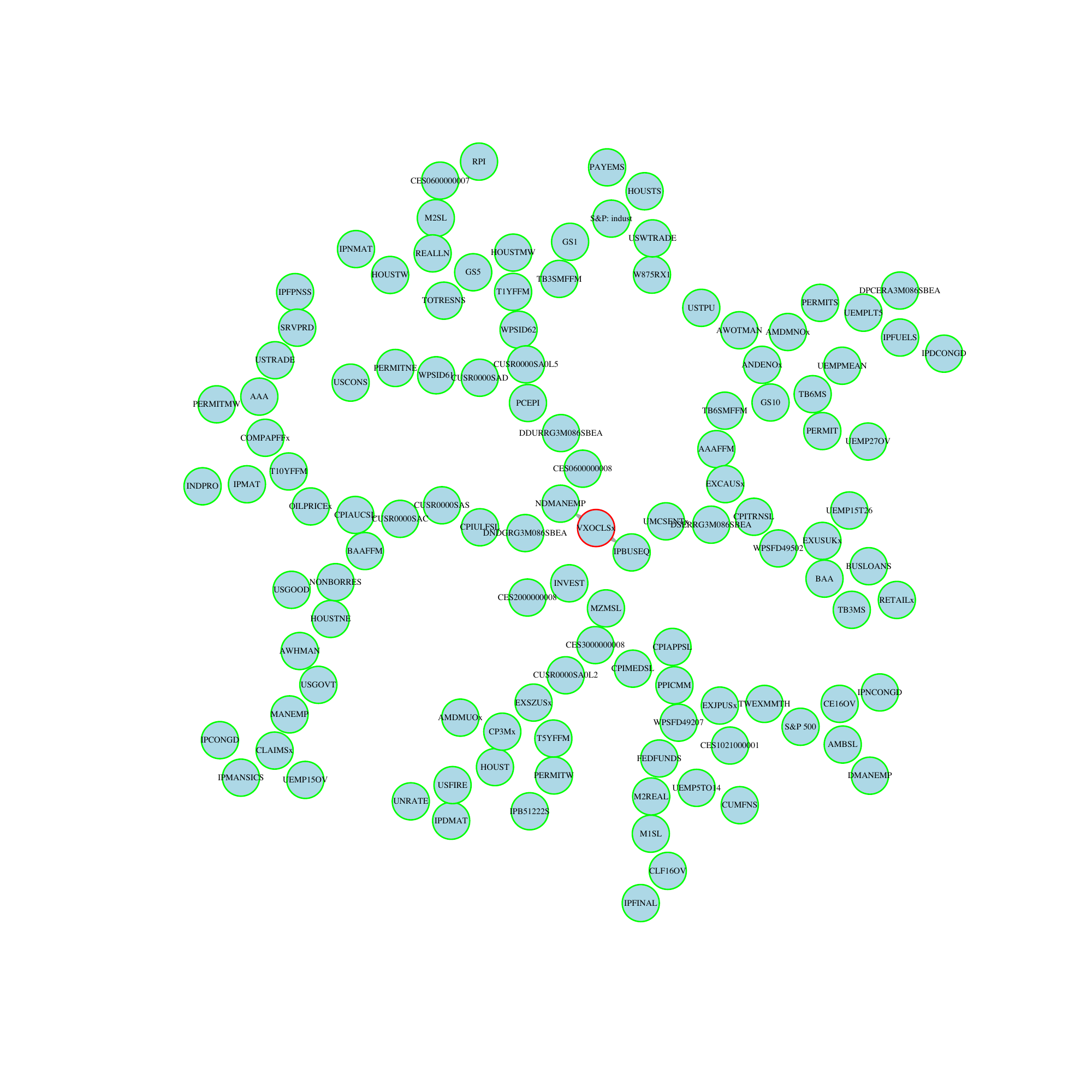}
\caption{PDS-LM, Caused by VXOCLSx, $\alpha=0.01$}
\label{figura_vxostat2}
\end{minipage}
\end{figure}

The difference between Figure \ref{figura_vox1}, \ref{figura_vox2} and Figure \ref{figura_vxostat1}, \ref{figura_vxostat2} is quite striking, especially for the case VXOCLSx$\rightarrow$  FRED. In fact, the number of significant connections drops dramatically from $39$ to just two and some differences can be found in the connections for the FRED$\rightarrow$VXOCLSx too. To zoom in on the connections and the difference between the lag-augmented (LA) and stationary transformed (ST) cases, in Figure \ref{grouped_0} we loosen the p-value threshold up to $10\%$ and group the variables by their FRED-MD sector classification. The bars now indicate the $p$-values of the test, with a full bar equaling a $p$-value of 0, and an empty bar indicating a $p$-value above 10\%.

\begin{figure}
\centering
\includegraphics[width = \textwidth]{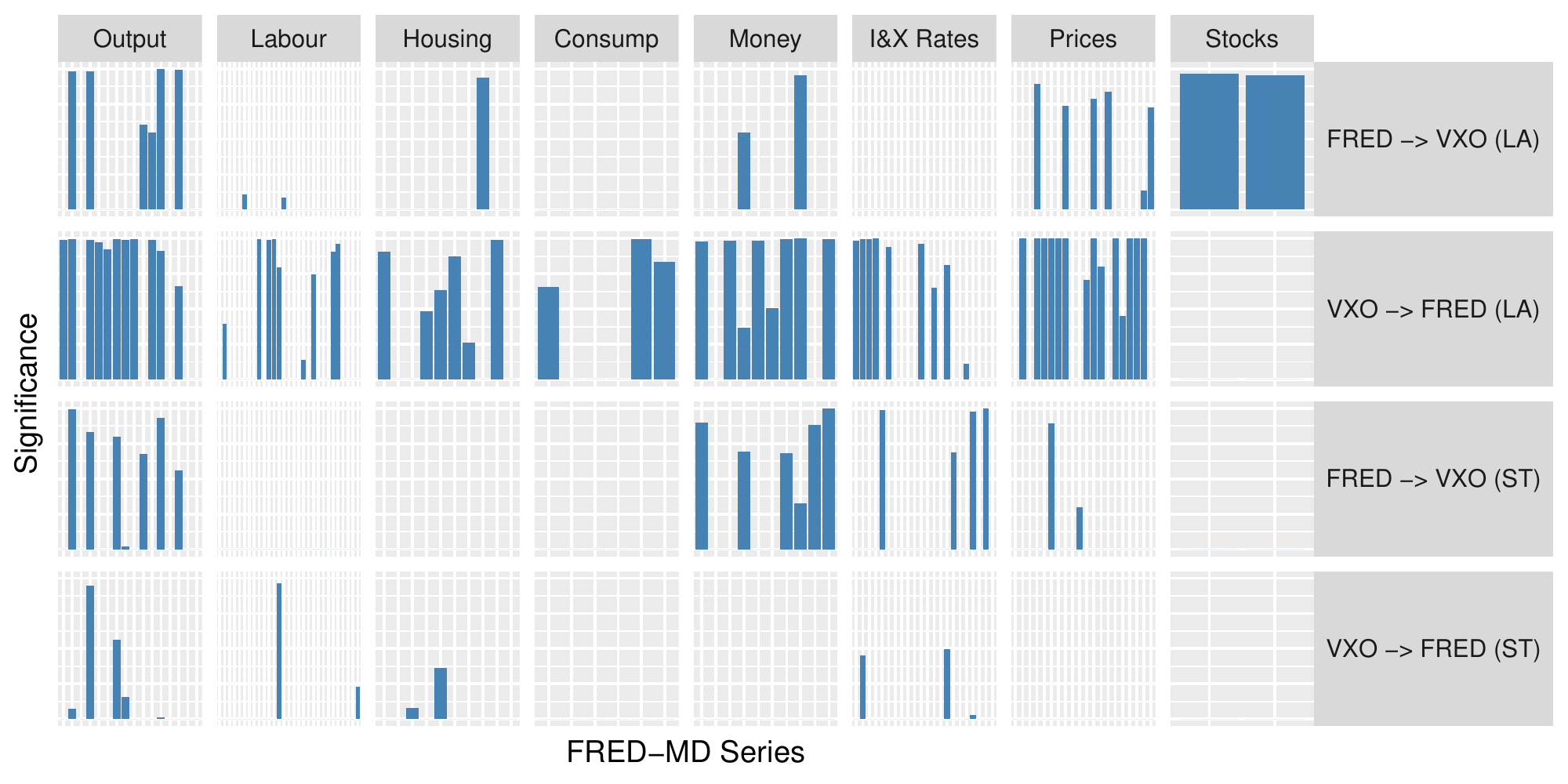}
\caption{ $p$-values for lag-augmented (LA) and stationary transformed (ST) Granger causality tests between FRED-MD and VXOCLSx. Full bars equal a $p$-value of 0, empty bars a $p$-value above 10\%}\label{grouped_0}
\end{figure}

The results again suggest how stationary transforming variables can have a profound impact on the inference performed, leading to very different pictures of Granger causality. It is particularly noticeable that sectors thought to be highly affected by uncertainty, such as the Output sector, barely has any significant connections left when testing for Granger causality from VXOCLSx. Similarly, the absence of causality from the Stocks category to VXOCLSx is surprising, given the latter's construction.

Next, we add US-EPU to the dataset and again repeat the analysis where now the focus is connections from and to US-EPU with all the macroeconomic series of FRED-MD using PDS-LA-LM. Importantly, now the results will be conditioned on VXOCLSx which remains in the information set. 

\begin{figure}
\begin{minipage}{0.5\textwidth}
\centering
\includegraphics[width=0.95\textwidth, trim = {3.4cm 3cm 2.2cm 2.6cm},clip]{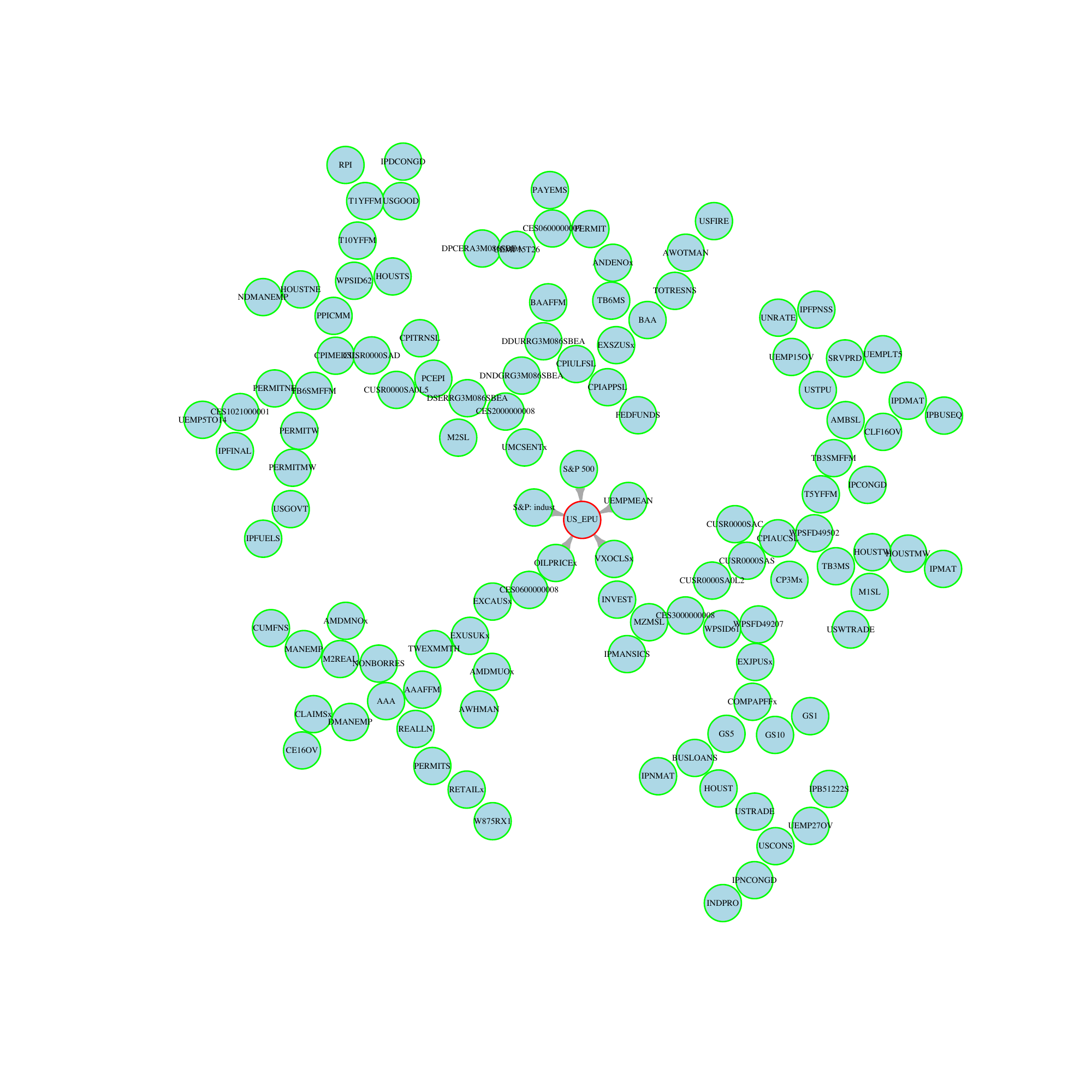}
\caption{PDS-LA-LM, Causes of US-EPU, $\alpha=0.01$}
\label{figura_epu3}
\end{minipage}
\begin{minipage}{.5\textwidth}
        \centering
\includegraphics[width=0.95\textwidth, trim = {3.4cm 3cm 2.2cm 2.6cm},clip]{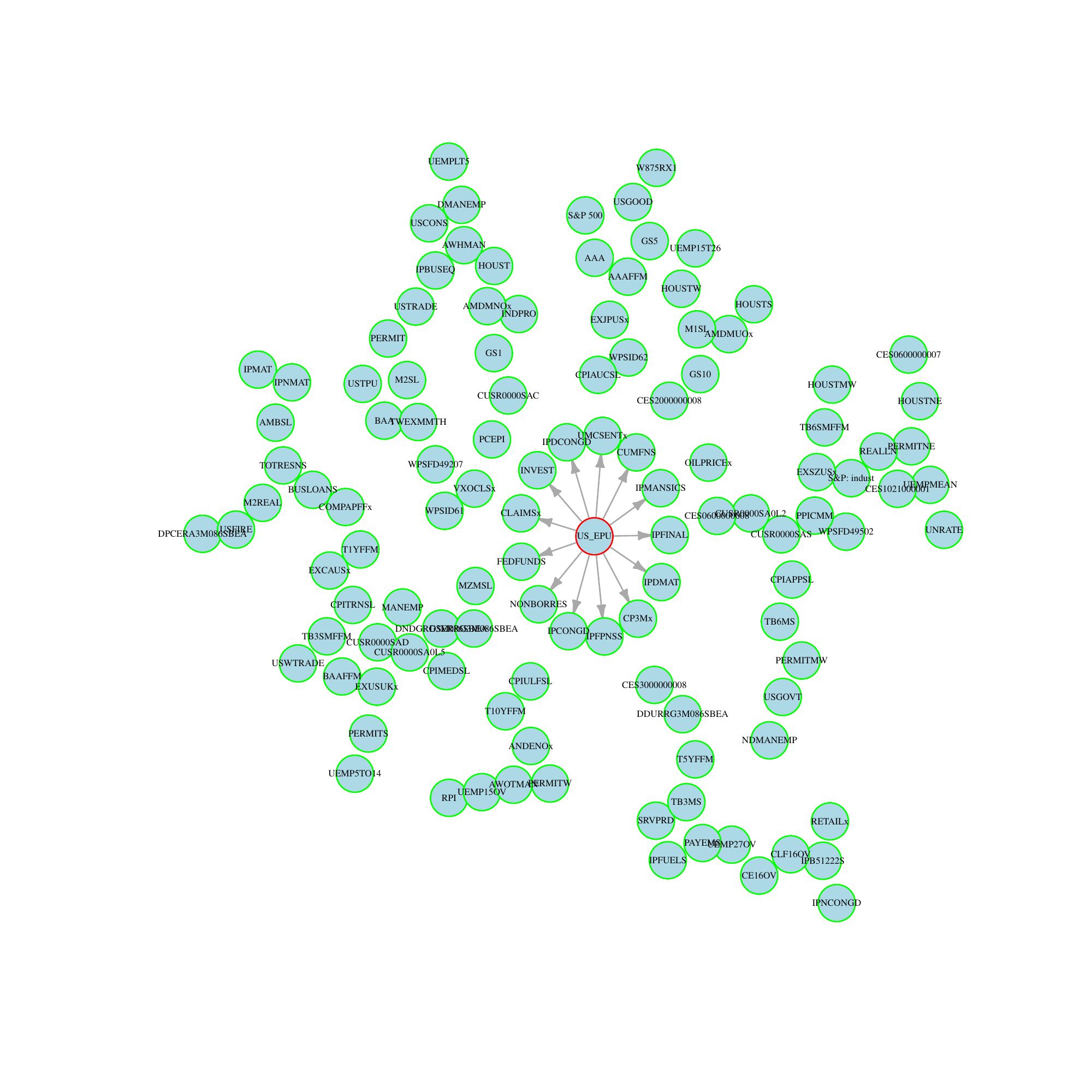}
\caption{PDS-LA-LM, Caused by US-EPU, $\alpha=0.01$}
\label{figura_epu4}
\end{minipage}
\end{figure}

In Figure \ref{figura_epu3} we only find five connections when it come to Granger causality relations from FRED-MD series to US-EPU. Vice-versa, in Figure \ref{figura_epu4} there are 13 causal paths from US-EPU to the other macroeconomic series. The results are quantitatively similar to FRED$\leftrightarrow$VXOCLSx in Figure \ref{figura_vox1}, \ref{figura_vox2}, though less pronounced. Both US-EPU and VXOCLSx are Granger caused by the S\&P500 and S\&P:indust but the remaining connections are not equal. Importantly, VXOCLSx is found to be Granger causal for US-EPU. Given the network US-EPU$\rightarrow$FRED is conditional on VXOCLSx, the non-overlapping connections found in Figure \ref{figura_epu4} that are not occurring in Figure \ref{figura_vox2} are new connections that could make US-EPU an interesting uncertainty index to add to FRED-MD. A cross-check reveals that only UMCSENTx is new, all others connections were already uncovered using VXOCLSx. This would seem to give credit to FRED-MD in using VXOCLSx as main index of economic uncertainty without the need to further include US-EPU. On the other hand, finding the connections from US-EPU despite the presence of VXOCLSx in the information set, indicates that US-EPU may still add additional predictive information about the variables predictable by VXOCLSx, suggesting that could still be worth including it in an empirical analysis.

\begin{figure}
\begin{minipage}{0.5\textwidth}
\centering
\includegraphics[width=0.95\textwidth, trim = {3.4cm 3cm 2.2cm 2.6cm},clip]{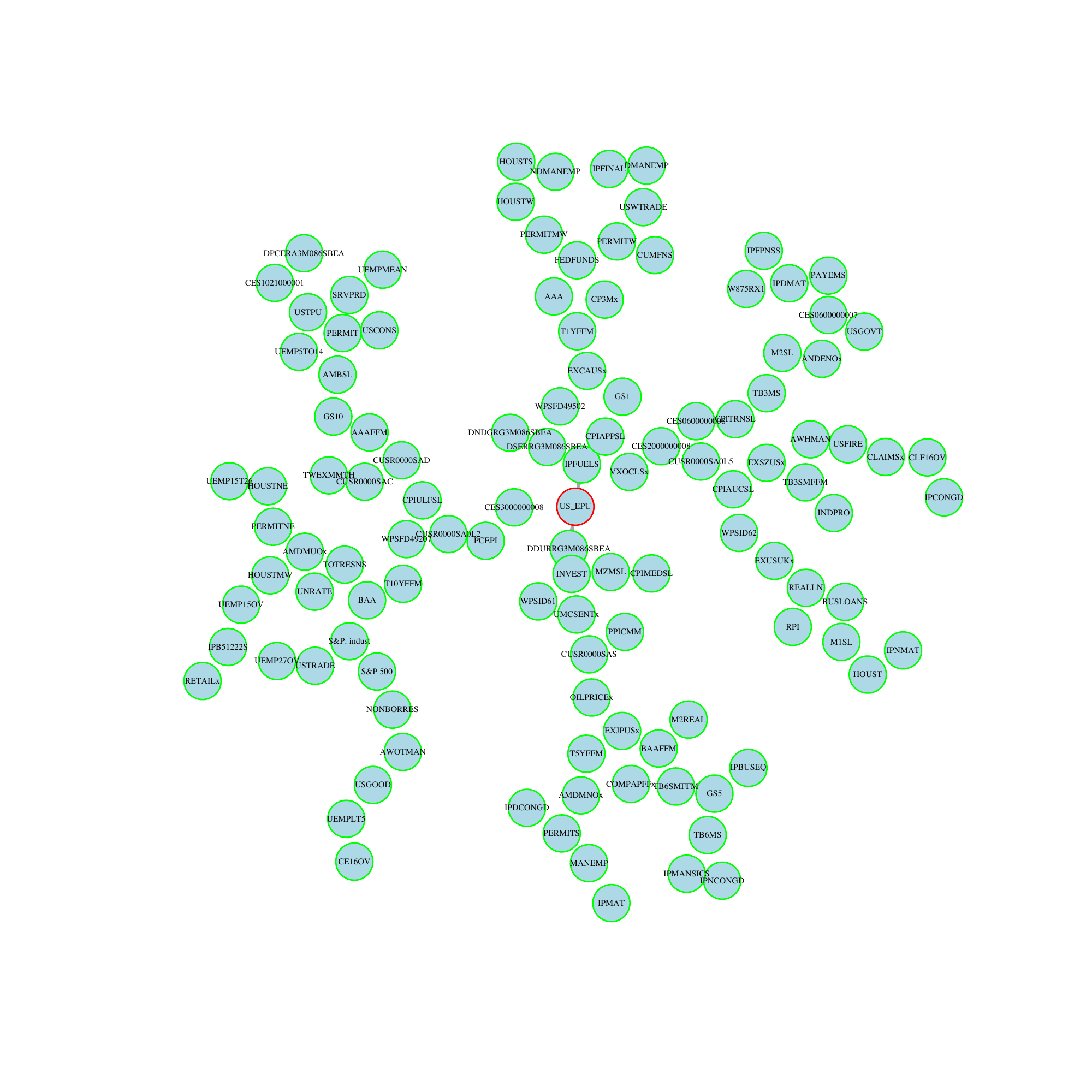}
\caption{PDS-LM, Causes of US-EPU, $\alpha=0.01$}
\label{figura_epustat1}
\end{minipage}
\begin{minipage}{.5\textwidth}
        \centering
\includegraphics[width=0.95\textwidth, trim = {3.4cm 3cm 2.2cm 2.6cm},clip]{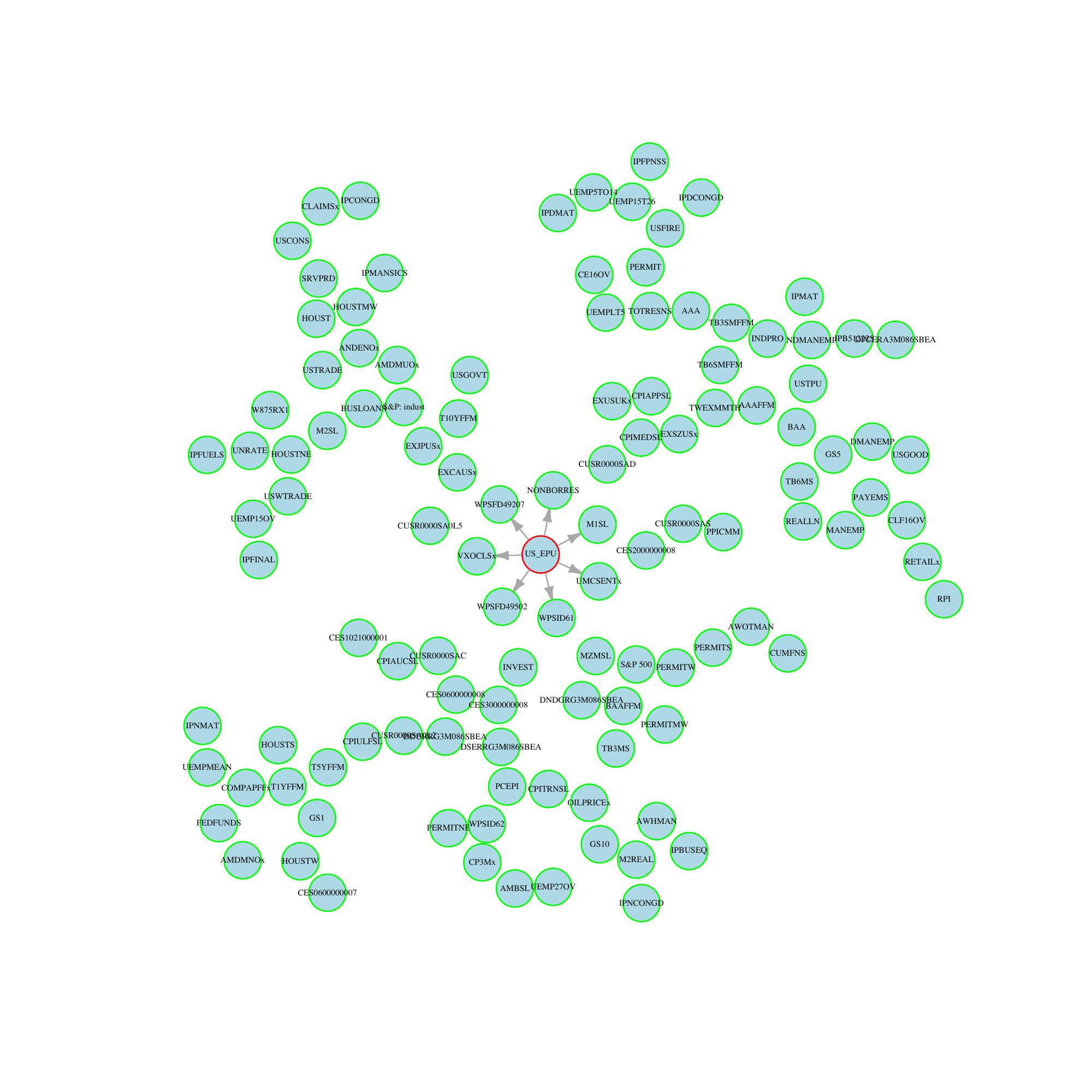}
\caption{PDS-LM, Caused by US-EPU, $\alpha=0.01$}
\label{figura_epustat2}
\end{minipage}
\end{figure}

As before, in Figure \ref{figura_epustat1}, \ref{figura_epustat2} we repeat the analysis, this time around however we apply the recommended stationary transformations (ST) from FRED-MD and we use \citet{hecq2021granger} PDS-LM test to investigate the same relations. Again, a decrease in the number of connections in both directions is evident from the network results. In Figure \ref{grouped_1} we again zoom in on the $p$-values and group the results by sector. The red bar refers to the VXOCLSx series which is found to be Granger causal for EPU in the analysis in levels. Interestingly, we find the opposite result with stationary transformed setting. Again, the differences in the relations  found indicate that transforming variables to stationarity may delete useful information about predictability.
\begin{figure}
\centering
\includegraphics[width = \textwidth]{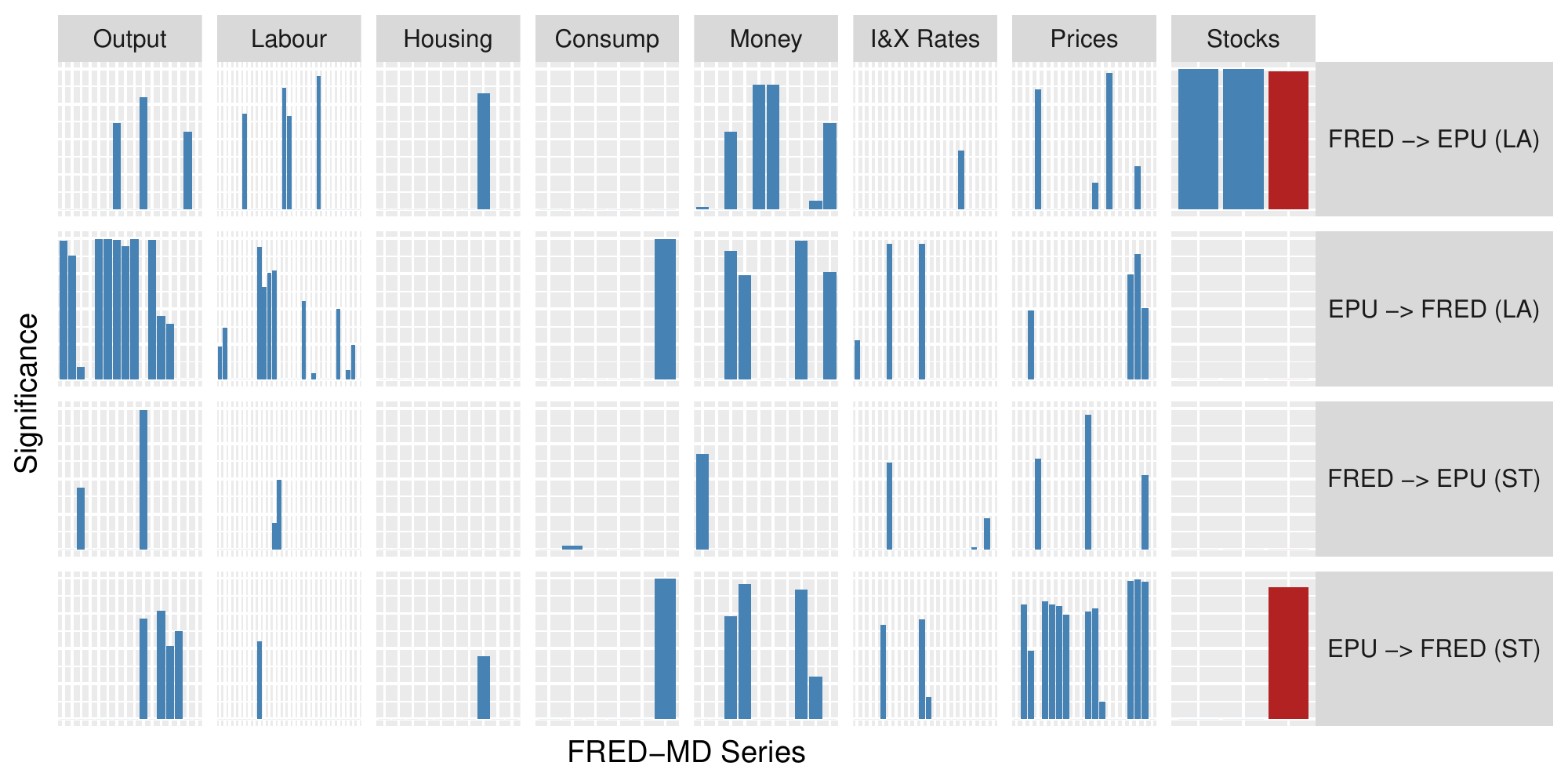}
\caption{ $p$-values for lag-augmented (LA) and stationary transformed (ST) Granger causality tests between FRED-MD and US-EPU. Full bars equal a $p$-value of 0, empty bars a $p$-value above 10\%}\label{grouped_1}
\end{figure}

Our analysis suggests that uncertainty accompanying a wide variety of global events, if measured in terms of expected volatility on the financial market, is primarily a cause rather than an effect of variations in the economic activity. This is in line with the recent work of \citet{ludvigson2021uncertainty} who find uncertainty about financial markets to be a source of output fluctuations. In addition, our analysis shows that even if different measures of uncertainty may be closely related, they are not identical and the choice of which one to use can affect further analysis. It may therefore be worthwhile to consider multiple measures of uncertainty in empirical applications.

\section{Conclusion}\label{sec:conclusion}
We propose an inferential procedure for Granger causality testing in high-dimensional non-stationary VAR models which avoids any knowledge or pre-tests about integration or cointegration in the data. To do so we adapt the \citet{toda1995statistical} approach of augmenting the lag length of the system and we show that by reducing this augmentation to only the variables of interest for the test, we are able to minimize parameter proliferation in high dimensions. We develop a post-double selection LM test which is based on penalized least squares estimators to partial-out those variables having no influence on the variables tested on while safeguarding from omitted variable bias using a double-selection mechanism.

We prove that the augmentation of the Granger causing variables has no effect on the null hypothesis tested, yet provides an automatic differencing mechanism letting the OLS estimator have standard asymptotic results. Also, we extend the relevant assumptions needed for the post-double selection estimator to work in the context of potential unit roots. We derive the asymptotics of the post-selection, augmented estimator, showing it attains standard asymptotic normality hence allowing for a valid test with standard $\chi^2$ limiting distribution.

Our proposed test shows good finite sample properties over different DGPs. We also give practical recommendations on both the optimal augmentation $d$ and on how to estimate the lag-length $p$. We argue that $d=2$ lags is the optimal augmentation in order to take into account possible I(2) as well as near I(2) variables that could compromise the size of the test. In order to estimate the lag length $p$ we propose to reduce the original VAR to a diagonal VAR. This reduces the high-dimensional system to a sequence of low-dimensional autoregressions, to which an information criteria is applied in order to select the correct lag length. 

Finally, we investigate how our test performs in practice by analysing the causes and effects of economic uncertainty. Using the FRED-MD dataset directly, without needing to apply their recommended stationary transformations, we compare two different ways of measuring economic uncertainty (VXO and EPU) and their relationship with all macroeconomic variables within the dataset. We also compare the analysis with the stationary transformation case, highlighting how such transformation can profoundly impact the results and give a very different picture of the causal structure. Our results suggest that uncertainty is primarily a cause rather than an effect of variations in the economic activity.

\bibliography{literature.bib}

\begin{thebibliography}{}

\bibitem[\protect\astroncite{Adamek et~al.}{2023}]{adamek2022lasso}
Adamek, R., Smeekes, S., and Wilms, I. (2023).
\newblock Lasso inference for high-dimensional time series.
\newblock {\em Journal of Econometrics}, 235(2):1114--1143.

\bibitem[\protect\astroncite{Baker et~al.}{2016}]{baker2016measuring}
Baker, S.~R., Bloom, N., and Davis, S.~J. (2016).
\newblock Measuring economic policy uncertainty.
\newblock {\em Quarterly Journal of Economics}, 131(4):1593--1636.

\bibitem[\protect\astroncite{Belloni et~al.}{2014}]{belloni2014inference}
Belloni, A., Chernozhukov, V., and Hansen, C. (2014).
\newblock Inference on treatment effects after selection among high-dimensional
  controls.
\newblock {\em Review of Economic Studies}, 81(2):608--650.

\bibitem[\protect\astroncite{Billio et~al.}{2012}]{billio2012econometric}
Billio, M., Getmansky, M., Lo, A.~W., and Pelizzon, L. (2012).
\newblock Econometric measures of connectedness and systemic risk in the
  finance and insurance sectors.
\newblock {\em Journal of Financial Economics}, 104(3):535--559.

\bibitem[\protect\astroncite{Bloom}{2009}]{bloom2009impact}
Bloom, N. (2009).
\newblock The impact of uncertainty shocks.
\newblock {\em Econometrica}, 77(3):623--685.

\bibitem[\protect\astroncite{Chang and Dickey}{1994}]{chang1994recognizing}
Chang, M.~C. and Dickey, D.~A. (1994).
\newblock Recognizing overdifferenced time series.
\newblock {\em Journal of Time Series Analysis}, 15(1):1--18.

\bibitem[\protect\astroncite{Cochrane}{1991}]{cochrane1991critique}
Cochrane, J.~H. (1991).
\newblock A critique of the application of unit root tests.
\newblock {\em Journal of Economic Dynamics and Control}, 15(2):275--284.

\bibitem[\protect\astroncite{Cubadda et~al.}{2009}]{cubadda2009studying}
Cubadda, G., Hecq, A., and Palm, F.~C. (2009).
\newblock Studying co-movements in large multivariate data prior to
  multivariate modelling.
\newblock {\em Journal of Econometrics}, 148(1):25--35.

\bibitem[\protect\astroncite{Dickey and Fuller}{1979}]{dickey1979distribution}
Dickey, D.~A. and Fuller, W.~A. (1979).
\newblock Distribution of the estimators for autoregressive time series with a
  unit root.
\newblock {\em Journal of the American Statistical Association},
  74(366a):427--431.

\bibitem[\protect\astroncite{Dolado and L{\"u}tkepohl}{1996}]{dolado1996making}
Dolado, J.~J. and L{\"u}tkepohl, H. (1996).
\newblock Making wald tests work for cointegrated var systems.
\newblock {\em Econometric Reviews}, 15(4):369--386.

\bibitem[\protect\astroncite{Friston et~al.}{2013}]{friston2013analysing}
Friston, K., Moran, R., and Seth, A.~K. (2013).
\newblock Analysing connectivity with granger causality and dynamic causal
  modelling.
\newblock {\em Current opinion in neurobiology}, 23(2):172--178.

\bibitem[\protect\astroncite{Fuller}{2009}]{fuller2009introduction}
Fuller, W.~A. (2009).
\newblock {\em Introduction to statistical time series}, volume 428.
\newblock John Wiley \& Sons.

\bibitem[\protect\astroncite{Granger}{1969}]{granger1969investigating}
Granger, C.~W. (1969).
\newblock Investigating causal relations by econometric models and
  cross-spectral methods.
\newblock {\em Econometrica}, pages 424--438.

\bibitem[\protect\astroncite{Granger}{1980}]{granger1980testing}
Granger, C.~W. (1980).
\newblock Testing for causality: a personal viewpoint.
\newblock {\em Journal of Economic Dynamics and control}, 2:329--352.

\bibitem[\protect\astroncite{Hamilton}{1994}]{hamilton1994time}
Hamilton, J.~D. (1994).
\newblock {\em Time Series Analysis}, volume~2.
\newblock Princeton University Press.

\bibitem[\protect\astroncite{Hecq et~al.}{2021}]{hecq2021granger}
Hecq, A., Margaritella, L., and Smeekes, S. (2021).
\newblock {Granger Causality Testing in High-Dimensional VARs: A
  Post-Double-Selection Procedure}.
\newblock {\em Journal of Financial Econometrics}, 21(3):915--958.

\bibitem[\protect\astroncite{Inoue and Kilian}{2020}]{inoue2020uniform}
Inoue, A. and Kilian, L. (2020).
\newblock The uniform validity of impulse response inference in
  autoregressions.
\newblock {\em Journal of Econometrics}, 215(2):450--472.

\bibitem[\protect\astroncite{Johansen}{1992}]{johansen1992representation}
Johansen, S. (1992).
\newblock A representation of vector autoregressive processes integrated of
  order 2.
\newblock {\em Econometric Theory}, 8(2):188--202.

\bibitem[\protect\astroncite{Johansen}{1995}]{johansen1995likelihood}
Johansen, S. (1995).
\newblock {\em Likelihood-based inference in cointegrated vector autoregressive
  models}.
\newblock OUP Oxford.

\bibitem[\protect\astroncite{Kiviet}{1986}]{kiviet1986rigour}
Kiviet, J.~F. (1986).
\newblock On the rigour of some misspecification tests for modelling dynamic
  relationships.
\newblock {\em The Review of Economic Studies}, 53(2):241--261.

\bibitem[\protect\astroncite{Kock and Callot}{2015}]{kock2015oracle}
Kock, A.~B. and Callot, L. (2015).
\newblock Oracle inequalities for high dimensional vector autoregressions.
\newblock {\em Journal of Econometrics}, 186(2):325--344.

\bibitem[\protect\astroncite{Leeb and P{\"o}tscher}{2005}]{leeb2005model}
Leeb, H. and P{\"o}tscher, B.~M. (2005).
\newblock Model selection and inference: Facts and fiction.
\newblock {\em Econometric Theory}, 21:21--59.

\bibitem[\protect\astroncite{Ludvigson et~al.}{2021}]{ludvigson2021uncertainty}
Ludvigson, S.~C., Ma, S., and Ng, S. (2021).
\newblock Uncertainty and business cycles: exogenous impulse or endogenous
  response?
\newblock {\em American Economic Journal: Macroeconomics}, 13(4):369--410.

\bibitem[\protect\astroncite{L{\"u}tkepohl}{2005}]{lutkepohl2005new}
L{\"u}tkepohl, H. (2005).
\newblock {\em New Introduction to Multiple Time Series Analysis}.
\newblock Springer Science \& Business Media.

\bibitem[\protect\astroncite{Masini et~al.}{2022}]{masini2022regularized}
Masini, R.~P., Medeiros, M.~C., and Mendes, E.~F. (2022).
\newblock Regularized estimation of high-dimensional vector autoregressions
  with weakly dependent innovations.
\newblock {\em Journal of Time Series Analysis}, 43(4):532--557.

\bibitem[\protect\astroncite{McCracken and Ng}{2016}]{mccracken2016fred}
McCracken, M.~W. and Ng, S. (2016).
\newblock {FRED-MD}: A monthly database for macroeconomic research.
\newblock {\em Journal of Business \& Economic Statistics}, 34(4):574--589.

\bibitem[\protect\astroncite{Mei and Shi}{2022}]{mei2022lasso}
Mei, Z. and Shi, Z. (2022).
\newblock On {LASSO} for high dimensional predictive regression.
\newblock arXiv e-print 2212.07052.

\bibitem[\protect\astroncite{Miao et~al.}{2023}]{miao2020high}
Miao, K., Phillips, P.~C., and Su, L. (2023).
\newblock High-dimensional vars with common factors.
\newblock {\em Journal of Econometrics}, 233(1):155--183.

\bibitem[\protect\astroncite{Montiel~Olea and
  Plagborg-M{\o}ller}{2021}]{montiel2021local}
Montiel~Olea, J.~L. and Plagborg-M{\o}ller, M. (2021).
\newblock Local projection inference is simpler and more robust than you think.
\newblock {\em Econometrica}, 89(4):1789--1823.

\bibitem[\protect\astroncite{Park and Phillips}{1988}]{park1988statistical}
Park, J.~Y. and Phillips, P.~C. (1988).
\newblock Statistical inference in regressions with integrated processes: Part
  1.
\newblock {\em Econometric Theory}, 4(3):468--497.

\bibitem[\protect\astroncite{Phillips and Durlauf}{1986}]{phillips1986multiple}
Phillips, P.~C. and Durlauf, S.~N. (1986).
\newblock Multiple time series regression with integrated processes.
\newblock {\em The Review of Economic Studies}, 53(4):473--495.

\bibitem[\protect\astroncite{Phillips and
  Hansen}{1990}]{phillips1990statistical}
Phillips, P.~C. and Hansen, B.~E. (1990).
\newblock Statistical inference in instrumental variables regression with i (1)
  processes.
\newblock {\em Review of Economic Studies}, 57(1):99--125.

\bibitem[\protect\astroncite{Phillips and Solo}{1992}]{phillips1992asymptotics}
Phillips, P.~C. and Solo, V. (1992).
\newblock Asymptotics for linear processes.
\newblock {\em The Annals of Statistics}, pages 971--1001.

\bibitem[\protect\astroncite{Seth et~al.}{2015}]{seth2015granger}
Seth, A.~K., Barrett, A.~B., and Barnett, L. (2015).
\newblock Granger causality analysis in neuroscience and neuroimaging.
\newblock {\em Journal of Neuroscience}, 35(8):3293--3297.

\bibitem[\protect\astroncite{Sims et~al.}{1990}]{sims1990inference}
Sims, C.~A., Stock, J.~H., Watson, M.~W., et~al. (1990).
\newblock Inference in linear time series models with some unit roots.
\newblock {\em Econometrica}, 58(1):113--144.

\bibitem[\protect\astroncite{Smeekes and Wijler}{2020}]{SmeekesWijler20}
Smeekes, S. and Wijler, E. (2020).
\newblock Unit roots and cointegration.
\newblock In Fuleky, P., editor, {\em Macroeconomic Forecasting in the Era of
  Big Data}, volume~52 of {\em Advanced Studies in Theoretical and Applied
  Econometrics}, chapter~17, pages --541--584. Springer.

\bibitem[\protect\astroncite{Smeekes and Wijler}{2021}]{smeekes2021automated}
Smeekes, S. and Wijler, E. (2021).
\newblock An automated approach towards sparse single-equation cointegration
  modelling.
\newblock {\em Journal of Econometrics}, 221(1):247--276.

\bibitem[\protect\astroncite{Smeekes and Wilms}{2023}]{SmeekesWilms20}
Smeekes, S. and Wilms, I. (2023).
\newblock bootur: An r package for bootstrap unit root tests.
\newblock {\em Journal of Statistical Software}, 106(12):1--39.

\bibitem[\protect\astroncite{Stern and Kaufmann}{2014}]{stern2014anthropogenic}
Stern, D.~I. and Kaufmann, R.~K. (2014).
\newblock Anthropogenic and natural causes of climate change.
\newblock {\em Climatic Change}, 122(1-2):257--269.

\bibitem[\protect\astroncite{Stock and Watson}{1989}]{stock1989interpreting}
Stock, J.~H. and Watson, M.~W. (1989).
\newblock Interpreting the evidence on money-income causality.
\newblock {\em Journal of Econometrics}, 40(1):161--181.

\bibitem[\protect\astroncite{Tibshirani}{1996}]{tibshirani1996regression}
Tibshirani, R. (1996).
\newblock Regression shrinkage and selection via the lasso.
\newblock {\em Journal of the Royal Statistical Society - Series B}, pages
  267--288.

\bibitem[\protect\astroncite{Toda and Yamamoto}{1995}]{toda1995statistical}
Toda, H.~Y. and Yamamoto, T. (1995).
\newblock Statistical inference in vector autoregressions with possibly
  integrated processes.
\newblock {\em Journal of Econometrics}, 66(1-2):225--250.

\bibitem[\protect\astroncite{Wijler}{2022}]{wijler2022restricted}
Wijler, E. (2022).
\newblock A restricted eigenvalue condition for unit-root non-stationary data.
\newblock (2208.12990).

\bibitem[\protect\astroncite{Wooldridge}{1987}]{wooldridge1987regression}
Wooldridge, J.~M. (1987).
\newblock A regression based lagrange multiplier statistic that is robust in
  the presence of heteroskedasticity.

\bibitem[\protect\astroncite{Zellner and Palm}{1974}]{zellner1974time}
Zellner, A. and Palm, F.~C. (1974).
\newblock Time series analysis and simultaneous equation econometric models.
\newblock {\em Journal of Econometrics}, 2(1):17--54.

\bibitem[\protect\astroncite{Zhang et~al.}{2019}]{zhang2019identifying}
Zhang, R., Robinson, P., and Yao, Q. (2019).
\newblock Identifying cointegration by eigenanalysis.
\newblock {\em Journal of the American Statistical Association},
  114(526):916--927.

\end{thebibliography}

\begin{appendices}
\numberwithin{table}{section}
\numberwithin{equation}{section}
\numberwithin{lemma}{section}
\numberwithin{assumption}{section}
\numberwithin{figure}{section}

\section{Main Theoretical Results} \label{sec:app_theory}
\subsection{Justification of Lag-Augmentation} \label{sec:LA}

By adding a ``free'' lag of $x_t$, we in essence allow the system to ``difference itself'' into the correct order to remove the unit roots. We will now make this argument more precise using some algebra manipulations similar to \citet{toda1995statistical}. We will focus on the bivariate case, but give some pointers where needed to extend the approach to multiple variables of interest. In the following, we will generally focus on the case where $d\leq 2$. While one can allow for higher orders $d$ in the estimation procedure, this would typically not be required for most economic applications of interest, where it is commonly accepted that the maximum order of integration of economic data hardly ever exceeds 2. Hence, we confine our presentation up to the case of $I(2)$ series.

Let us introduce the following $(p+d)\times (p+d)$ upper-triangular transformation matrices:
$$\boldsymbol{R}=\begin{bmatrix}
1&1&1&\cdots&1\\
0&1&1&\cdots&1\\
0&0&1&\cdots&1\\
\vdots&\vdots&\vdots&\ddots&\vdots\\
0&0&0&\cdots&1
\end{bmatrix}, \quad \bP_d=\bR^d=\begin{bmatrix}
1&2&3&\cdots&p+d\\
0&1&2&\cdots&p+d-1\\
0&0&1&\cdots&p+d-2\\
\vdots&\vdots&\vdots&\ddots&\vdots\\
0&0&0&0&\cdots&1
\end{bmatrix}.$$
$\bP_d$ can be considered as an ``integration matrix'' which integrates the vector it is multiplied with to order $d$. Similarly, we can consider the inverse matrix $\bP_d^{-1}$ as a difference matrix, which takes differences $d$ times. Now define the transformed data and parameters
\begin{equation}
\begin{split}
\label{eq_dsubs}
&\begin{bmatrix}
\bbeta_{+}^{*}\\ \bdelta
\end{bmatrix} = 
\begin{bmatrix}
\bP_d& \bzero \\
\bzero & \bI_{p(K-1)} 
\end{bmatrix} 
\begin{bmatrix}
\bbeta_{+}\\ \bdelta_{2} 
\end{bmatrix},\qquad
\begin{bmatrix}
\bX_{\la}^*\\ \bV 
\end{bmatrix}
=
\begin{bmatrix}
\bX_{\la}\\
\bV 
\end{bmatrix}
\begin{bmatrix}
\boldsymbol{P}_d &\bzero \\
\bzero & \boldsymbol{I}_{p(K-1)} 
\end{bmatrix}^{-1}.
\end{split}
\end{equation}
Then we can rewrite \eqref{eq:lag_augmented} as
\begin{equation}
\begin{aligned}
\label{eq_varpostalgebra}
\by &= \bX_{\la}\boldsymbol{P}_d^{-1}\boldsymbol{P}_d \bbeta_{+} + \bV \boldsymbol{P}_d^{-1}\boldsymbol{P}_d \bdelta + \bu, =\\
&= \bX_{\la}^* \bbeta_{+}^* + \bV \bdelta + \bu.
\end{aligned}
\end{equation}
$\bX_{\la}^* = \bP_d^{-1} \bX_{\la}$ now contains $d$-th order differences for the first $p$ entries, which are $I(0)$ if the original data are $I(d)$. Furthermore, performing the test of whether the first $p$ elements of $\bbeta_+$ are zero has a one-to-one correspondence to a test on the first $p$ elements of $\bbeta_{+}^*$. This then allows us to transform the test of Granger non-causality on $I(d)$ into an equivalent test on $I(0)$ data, thus leading to a test with a standard limit distribution regardless of the order of integration of the data.

We will now show that the null hypothesis of Granger non-causality in \eqref{eq_varpostalgebra} has a one-to-one correspondence to a test on the first $p$ elements of $\bbeta_{+}^*$, denoted as $\bbeta_p^*$. 
\begin{lemma} \label{lem:Ho_equiv}
Consider testing \eqref{eq:H0} in \eqref{eq:dgp_beta}. This is equivalent to testing
\begin{equation}\label{eq_H0rev}
H_0: \bbeta_p^*= \bzero \quad \text{against}\quad H_1: \bbeta_p^* \neq \bzero
\end{equation}
in \eqref{eq_varpostalgebra}
\end{lemma}

\begin{proof}[Proof of Lemma \ref{lem:Ho_equiv}]
Define the selection matrix $\bM = \left(\begin{matrix}\bI_{p}& \bzero_{p\times d}\\ \bzero_{d\times p} & \bzero_{d\times d} \end{matrix}\right)$.
Let $\bR_{p\times p}$ denote the upper left $p \times p$ block of $\bR$, and note that
\begin{equation*}
\bM \bP_d = \bM \bR^d = \left(\begin{matrix}
\bR_{p\times p} & \bzero_{{p\times d}}\\ 
\bzero_{d\times p} & \bzero_{d\times d}
\end{matrix}\right)^d = (\bM \bR)^d.
\end{equation*}
Therefore, we can write
\begin{equation}\label{eq:bMPR}
\left(\begin{matrix}
\bbeta_p^* \\
\bzero_{d\times 1}
\end{matrix}\right)
= \bM \bbeta_+^* = \bM \bP_d \bbeta_+ = (\bM \bR)^d \bbeta_+ = \left(\begin{matrix}
\bR_{p\times p}^d \bbeta \\
\bzero_{d\times 1}
\end{matrix}\right).
\end{equation}
As $\bR_{p\times p}$ is invertible, it follows that the test of $\bbeta_p=\bzero \Longleftrightarrow \bbeta_p^* = \bR_{p\times p} \bzero = \bzero$.
\end{proof}

The following examples clarify how the lag augmentation allows the interest variables to be re-expressed as difference-stationary. Let $p=d=1$ and let $\mathbb{x}_t^*$ be a row of $\bX_{\la}^*$. Then
\begin{equation*}
    \bR=\boldsymbol{P}_1=\begin{bmatrix} 1&1\\
    0&1 \end{bmatrix};\; \boldsymbol{P}_1^{-1}=\begin{bmatrix*}[r] 1&-1\\
    0&1 \end{bmatrix*};\; \bP_{1}^{-1}\mathbb{x}_t^{*'}=\begin{bmatrix*}[r] 1&-1\\
    0&1 \end{bmatrix*}\begin{bmatrix}
    x_{t-1}\\ x_{t-2}
    \end{bmatrix}=\begin{bmatrix}
    \Delta x_{t-1}\\ x_{t-2}
    \end{bmatrix}.
\end{equation*}
For $p=2, d=1$ we get
\begin{equation*}
   \boldsymbol{R}=\bP_1=\begin{bmatrix}
   1&1&1\\0&1&1\\0&0&1
   \end{bmatrix};\;  \boldsymbol{P}_1^{-1}=\begin{bmatrix*}[r] 1&-1&0\\
    0&1&-1\\
    0&0&1\end{bmatrix*};\; \bP_1^{-1}\mathbb{x}_t^{*'}=\begin{bmatrix*}[r] 1&-1&0\\
    0&1&-1\\
    0&0&1\end{bmatrix*}\begin{bmatrix}
    x_{t-1}\\x_{t-2}\\x_{t-3}
    \end{bmatrix}=\begin{bmatrix}
    \Delta x_{t-1}\\ \Delta x_{t-2}\\x_{t-3}
    \end{bmatrix}.
\end{equation*}
As a final example, let $p=d=2$:
\begin{equation*}\begin{aligned}
   &\boldsymbol{R}=\begin{bmatrix}
   1&1&1&1\\0&1&1&1\\0&0&1&1\\0&0&0&1
   \end{bmatrix};\; \boldsymbol{P}_2=\bR \bR=\begin{bmatrix}
   1&2&3&4\\0&1&2&3\\0&0&1&2\\0&0&0&1
   \end{bmatrix};\;  \boldsymbol{P}_2^{-1}=\begin{bmatrix*}[r]
   1&-2&1&0\\0&1&-2&1\\0&0&1&-2\\0&0&0&1
   \end{bmatrix*};\; \\
   &\bP_2^{-1}\mathbb{x}_t^{*'}=\begin{bmatrix*}[r]
   1&-2&1&0\\0&1&-2&1\\0&0&1&-2\\0&0&0&1
   \end{bmatrix*}\begin{bmatrix}
    x_{t-1}\\x_{t-2}\\x_{t-3}\\x_{t-4}
    \end{bmatrix}=\begin{bmatrix}
\Delta^2 x_{t-1}\\ \Delta^2 x_{t-2}\\x_{t-3}-2x_{t-4}\\x_{t-4}
    \end{bmatrix}.
    \end{aligned}
\end{equation*}

We now show the numerical equivalence of the Wald and LM test statistic for testing \eqref{eq:H0} in \eqref{eq:lag_augmented} and \eqref{eq_H0rev} in \eqref{eq_varpostalgebra}. Let $\bW_S$ denote a subset $S$ of the columns of $\bW_{-p}$. In the low-dimensional case $S$ may simply contain all columns of $\bW_{-p}$, whereas in the high-dimensional case $S$ is determined through some sort of (double) selection procedure. Let $\bV_S = (\bY_{-p}, \bW_S)$ and $\bdelta_S = (\bdelta_1, \bdelta_{2,S})$. Let the coefficients of the regression
\begin{equation} \label{eq:regrS}
\bY = \bX_{\la} \bbeta_+ + \bV_S \bdelta_S + \bu,
\end{equation}
be denoted as $(\hat{\bbeta}_+', \hat{\bdelta}_S')'$ and the coefficients of the regression
\begin{equation} \label{eq:regrSstar}
\bY = \bX_{\la}^* \bbeta_+^* + \bV_S \bdelta_S + \bu,
\end{equation}
be denoted as $(\hat{\bbeta}_+^{*\prime}, \hat{\bdelta}_S^{*\prime})'$.

\begin{lemma} \label{lem:OLSequiv}
Let $\hat{\bbeta}_p$ denote the first $p$ elements of $\hat{\bbeta}_+$, the OLS estimator of $\bbeta_+$ in \eqref{eq:regrS}, and $\hat{\bbeta}_p^*$ denote the first $p$ elements of $\hat{\bbeta}_+^*$, the OLS estimator of $\bbeta_+^*$ in \eqref{eq:regrSstar}. Also define 
\begin{equation*}
\hat{\bu}_+ = \by - \bX_{\la} \hat{\bbeta}_+ - \bV_{S} \hat{\bdelta}_{S}, \qquad \hat{\bu}_+^* = \by - \bX_{\la}^* \hat{\bbeta}_+^* - \bV_S \hat{\bdelta}_{S}^*,
\end{equation*} 
as the corresponding residuals. Then consider the Wald and LM statistics as defined in Algorithm \ref{alg:pdslm} (with the generic set $S$ replacing $\hS$). Then $\hat{\bbeta}_p^* = \bR_{p\times p}^d \hat{\bbeta}_p$, $\Wald = \Wald^*$ and $\LM = \LM^*$.
\end{lemma}

\begin{proof}[Proof of Lemma \ref{lem:OLSequiv}]
Define the residual-maker $\M(\bX) := \bI - \bX (\bX'\bX)^{-1} \bX'$. Then,
\begin{equation*}
\begin{aligned}
\left(\begin{matrix}
\hat{\bbeta}_p^* \\
\bzero_{d\times 1}
\end{matrix}\right)
&= \bM \hat{\bbeta}_{+}^* = \bM \left(\bX_{\la}^{*\prime} \M(\bV_S) \bX_{\la}^* \right)^{-1}\bX_{\la}^{*\prime} \M(\bV_S)\by\\
&=\bM \bP_d \left(\bX_{\la}' M(\bV_S) \bX_{\la}\right)^{-1}\bX_{\la}' \M(\bV_S) \by =\bM \bP_d \hat{\bbeta}_{+} = \left(\begin{matrix}
\bR_{p\times p}^d \hat{\bbeta}_p \\
\bzero_{d\times 1}
\end{matrix}\right)
.
\end{aligned}
\end{equation*}
It therefore follows that $\hat{\bbeta}_p^* = \bR_{p\times p}^d \hat{\bbeta}_p$.

Next, let $\bZ_+ := (\bX_{\la}, \bV_S)$ and $\bZ_+^* := (\bX_{\la}^*, \bV_S)$, and define 
\begin{equation*}
\bP_d^* =\left(\begin{matrix}
\bP_d & \bzero_{(p+d)\times (K-1)p} \\
\bzero_{(K-1)p\times (p+d)} & \bI_{(K-1)p}
\end{matrix} \right).
\end{equation*}
It then follows directly that $\bZ_+^* = \bZ_+ \bP_d^*$ and consequently
\begin{equation*}
\begin{split}
\hat{\bu}_+^* = M(\bZ_+^*) \by = \left(\bI - \bZ_+ \bP_d^* (\bP_d^{*\prime} \bZ_+' \bZ_+ \bP_d^*)^{-1} \bP_d^{*\prime} \bZ_+ \right) \by = M(\bZ_+)\by = \hat{\bu}_+.
\end{split}
\end{equation*}

Letting $\bM_- = (\bI_p, \bzero_{p \times d})$, then it follows directly from \eqref{eq:bMPR} that $\bM_- \bP_d = \bR_{p\times p}^d \bM_-$. The test statistic can be written as
\begin{small}
\begin{equation*}
\begin{split}
\Wald^* &= \by' \M(\bV_S) \bX_{\la}^* \left(\bX_{\la}^{*\prime} \M(\bV_S) \bX_{\la}^* \right)^{-1} \bM_-' \left[\hat{\sigma}_u^{*2} \left(\bM_-\bX_{\la}^{*\prime} \M(\bV_S) \bX_{\la}^* \bM_-' \right)^{-1} \right]^{-1} \\
&\quad\times \bM_- \left(\bX_{\la}^{*\prime} \M(\bV_S) \bX_{\la}^* \right)^{-1} \bX_{\la}^{*\prime} \M(\bV_S) \bX_{\la}^* \by\\
&= \hat{\sigma}_u^{-2} \by' \M(\bV_S) \bX_{\la} \left(\bX_{\la}^{\prime} \M(\bV_S) \bX_{\la} \right)^{-1} \bP_d' \bM_-' \left(\bM_- \bP_d^{-1\prime} \bX_{\la}^{\prime} \M(\bV_S) \bX_{\la} \bP_d^{-1} \bM_-' \right) \\
&\quad\times\bM_- \bP_d \left(\bX_{\la}^{\prime} M(\bV_S) \bX_{\la} \right)^{-1} \bX_{\la}^{\prime} \M(\bV_S) \bX_{\la} \by \\
&= \hat{\sigma}_u^{-2} \by' \M(\bV_S) \bX_{\la} \left(\bX_{\la}^{\prime} \M(\bV_S) \bX_{\la} \right)^{-1} \bM_-' \bR_{p\times p}^{d\prime}\\
&\quad\times \left(\bR_{p\times p}^{d\prime}\right)^{-1} \bM_- \bX_{\la}^{\prime} \M(\bV_S) \bX_{\la} \bM_-' \left(\bR_{p\times p}^d \right)^{-1} \\
&\quad\times \bR_{p\times p}^d\bM_- \left(\bX_{\la}^{\prime} \M(\bV_S) \bX_{\la} \right)^{-1} \bX_{\la}^{\prime} \M(\bV_S) \bX_{\la} \by\\
&= \hat{\sigma}_u^{-2} \by' \M(\bV_S) \bX_{\la} \left(\bX_{\la}^{\prime} \M(\bV_S) \bX_{\la} \right)^{-1} \bM_-' \bM_- \bX_{\la}^{\prime} \M(\bV_S) \bX_{\la} \bM_-' \\
&\quad\times \bM_- \left(\bX_{\la}^{\prime} \M(\bV_S) \bX_{\la} \right)^{-1} \bX_{\la}^{\prime} \M(\bV_S) \bX_{\la} \by = \hat{\bbeta}_p^{\prime} \left( \widehat{\var} \hat{\bbeta}_p \right)^{-1} \hat{\bbeta}_p = \Wald.
\end{split}
\end{equation*}
\end{small}

Let $\bV_+ = (\bx_{-(p+1)}, \ldots, \bx_{-(p+d)}, \bV)$ and $\bV_+^* = (\bx_{-(p+1)}^*, \ldots, \bx_{-(p+d)}^*, \bV)$. Defining $\tilde{\sigma}_u^2 = \hat{\bxi}' \hat{\bxi}$, we can write the LM test as
\begin{equation*}
\begin{split}
\LM &= \tilde{\sigma}_u^{-2} \by^\prime \M(\bV_{+,S}) \bX_{-p} \left[\bX_{-p}^{\prime} \M(\bV_{+,\hS}) \bX_{-p} \right]^{-1} \bX_{-p}^{\prime} \M(\bV_{+,\hS}) \by \\
&= \Wald \frac{\tilde{\sigma}_u^{2}} {\hat{\sigma}_u^{2}}.
\end{split}
\end{equation*}
As we have already shown that $\Wald= \Wald^*$ and $\hat{\sigma}_u^{2} = \hat{\sigma}_u^{*2}$, it remains to show that $\tilde{\sigma}_u^{2} = \tilde{\sigma}_u^{*2}$. From the definition of $\bP_d$ we have that $\bx_{-(p+1)}^*, \ldots, \bx_{-(p+d)}^*$ are a linear combination of only $\bx_{-(p+1)}, \ldots, \bx_{-(p+d)}$. Therefore, $\M(\bV_{+,S}) \by = \M(\bV_{+,S}^*) \by$ from which the result follows. This completes the proof of the lemma.
\end{proof}

\subsection{Proofs of Main Results}\label{sec:main_proofs}

\begin{proof}[Proof of Theorem \ref{th:pds_cons}]
Let $\P(\bA) = \bA (\bA^\prime \bA)^{-1} \bA^\prime$ denote the projection on the space spanned by $\bA$ and let $\M (\bA) = I - \P(\bA)$ denote the corresponding residual-maker matrix. For any matrix $\bA$, let the norm $\norm{\cdot}_p$ represent the induced $l_p$-matrix norm $\norm{\bA}_p = \sup_{x\neq 0} \norm{\bA x}_p / \norm{x}_p$.

We are working on the model
\begin{equation}\label{dgp}
\by= \bX_{-p} \bbeta + \bV \bdelta + \bu = \bX_{\la} \bbeta_+ + \bV \bdelta + \bu = \bX_{\la}^* \bbeta^* + \bV \bdelta + \bu.
\end{equation}
Define the set $\hSmj$ as the union of $\hS$ and all lags of $\bx$ except the $j$-th, such that the matrix $\bZ_{\hSmj} = (\bx_{-1}, \ldots, \bx_{-j+1}, \bx_{-j-1},\ldots, \bx_{-(p+d)}, \bV_{\hS})$ contains all variables used in in the second-stage regression except $\bx_{-j}$. Similarly, define the set $\Smj$ as the union of $\S$ and all lags of $\bx$ except the $j$-th, such that the matrix $\bZ_{\Smj} = (\bx_{-1}, \ldots, \bx_{-j+1}, \bx_{-j-1},\ldots, \bx_{-(p+d)}, \bV_{\S})$ contains all variables used in an `oracle' version of the second-stage regression.
For the $j$-th coefficient in $\hat{\bbeta}_+$ ($\bbeta_+$), we can then write
\begin{equation*}
\begin{split}
\sqrt{T} (\hat{\beta}_{j,+} - \beta_{j+}) &=  \underbrace{\left(T^{-2a}(\bx_{-j}^{\prime}\mathcal{M}(\bZ_{\hSmj}) \bx_{-j}) \right)^{-1}}_{A_T^{-1}} \underbrace{T^{1/2-2a} \bx_{-j}^{\prime} \mathcal{M}(\bZ_{\hSmj})\left[\bZ_{-p\backslash\{-j\}} \boeta_{0,-j} + \bu\right]}_{B_T},
\end{split}
\end{equation*}
where $\boeta_{0,-j}$ is equal to $\boeta_0$ but with a zero inserted for the coefficient of $\bx_{-j}$, and $a=1/2$ if $\bx$ is I(0), $a=1$ if $\bx$ is I(1) and  $a=2$ if $\bx$ is I(2).

We are now going to show in turn that 
\begin{equation}\label{proofA}
\begin{aligned}
A_{T,j} = T^{-2a}\bx_{-j}^{\prime} \mathcal{M}(\bZ_{\Smj}) \bx_{-j} + o_p(1), \qquad j = 1, \ldots, p+d,\\
\end{aligned}
\end{equation}
and 
\begin{equation}\label{proofB}
\begin{aligned}
B_{T,j} = T^{1/2-2a}  \bx_{-j}^{\prime} \mathcal{M}(\bZ_{\Smj}) \left[\bV \bdelta + \bu \right] + o_p(1),\qquad j = 1, \ldots, p+d.
\end{aligned}
\end{equation}

We write the first-stage regressions as
\begin{equation*}
\begin{split}
\by &= \bZ_{-p} \boeta_0 + \be_0, \qquad
\bx_{-j} = \bZ_{-p} \boeta_j + \be_j, \quad j = 1, \ldots, p.
\end{split}
\end{equation*}
Note that here for notational convenience, in a slight abuse of notation we insert $\bx_j$ on the right-hand side into $\bZ_{-p}$ instead of considering $\bZ_{-p,\{-j\}}$ as we do in the actual regression. Thereby, we implicitly increase the vector $\boeta_j$ with one entry for $\bx_{-j}$, for which we insert a zero.  

We now prove \eqref{proofA}. We have that
\begin{equation*}
\begin{aligned}
A_{T,j} &= T^{-2a} \left(\bZ_{-p} \boeta_j + \be_j \right)^{\prime} \left[\M(\bZ_{\hSmj})-\M(\bV_{\Smj})\right] \left(\bZ_{-p} \boeta_j + \be_j \right)\\
&= \underbrace{T^{-2a} \boeta_j'\bZ_{-p}' \left[\M(\bZ_{\hSmj})-\M(\bZ_{\Smj})\right] \bZ_{-p}\boeta_j}_{A_{T,j,1}} \\
&\quad +\underbrace{T^{-2a} \be_j' \left[\M(\bZ_{\hSmj})-\M(\bZ_{\Smj})\right] \be_j }_{A_{T,j,2}} + \underbrace{2T^{-2a} \boeta_j'\bZ_{-p}' \left[\M(\bZ_{\hSmj}) - \M(\bZ_{\Smj})\right] \be_j}_{A_{T,j,3}}.
\end{aligned}
\end{equation*}

For $A_{T,j,1}$, observe how $\boeta_j$ is zero outside the active set $\S_j$ which is in turn a subset of $\Smj$, such that we may write $\bZ_{-p} \boeta_j = \bZ_{\Smj} \boeta_{j,\Smj}$. It then directly follows that $\M(\bZ_{\Smj}) \bZ_{-p} \boeta_j = \bzero$. Hence, $A_{T,j,1}$ reduces to
\begin{equation*}
\begin{aligned}
\abs{A_{T,j,1}} &= \abs{\boeta_j' \bZ_{-p}' \M(\bZ_{\hSmj}) \bZ_{-p} \boeta_j} = \norm{\M(\bZ_{\hSmj}) \bZ_{-p} \boeta_j}_2^2
\leq \norm{\M(\bZ_{\hS_j}) \bZ_{-p} \boeta_j}_2^2 \\
&\leq \min_{\underline{\boeta}: \underline{\boeta}_m = 0, m \notin \hS_j} \norm{\bZ_{-p} \boeta_j - \bZ_{\hS_j} \boeta}_2^2 
\leq \norm{\bZ_{-p} (\boeta_j - \hat{\boeta}_j) }_2^2\leq \delta_T^2 T^{1/2},
\end{aligned}
\end{equation*}
which follows from observing that $\hS_j \subseteq \hSmj$, and that the constraint in the minimization is satisfied given $\hS_j = \{m: \hat{\eta}_{j,m} = 0\}$ and where the last inequality follows with probability $1 - \Delta_T$ from Assumption \partref{as:hlev}{as:cons}. 

For $A_{T,j,2}$, note that the difference between the residuals is equal to the difference between the corresponding projections, such that we can simplify the expression to
\begin{equation*}
\begin{aligned}
\abs{A_{T,j,2}} &= \abs{\be_j'\left(\P(\bZ_{\Smj})-\P(\bZ_{\hSmj})\right)\be_j} \leq \underbrace{\abs{\be_j'\left(\P(\bZ_{\Smj}))\right)\be_j}}_{A_{T,j,2}^1} + \underbrace{\abs{\be_j'\left(\P(\bZ_{\hSmj})\right)\be_j}}_{A_{T,j,2}^2}.
\end{aligned}
\end{equation*}
For $j=1, \ldots, p$, we then find that
\begin{align*}
A_{T,j,2}^1 &= \abs{\be_j'\bZ_{\Smj} \bD_{T,\Smj}^{-1} \left[\bD_{T,\Smj}^{-1}\bZ_{\Smj}' \bZ_{\Smj} \bD_{T,\Smj}^{-1} \right]^{-1}  \bD_{T,\Smj}^{-1} \bZ_{\Smj}' \be_j}\\
&\overset{(I)}{\leq} \sqrt{\bar{s}_T} \norm{\bD_{T,\Smj}^{-1} \bZ_{\Smj}' \be_j}_{\infty} \norm{\left[\bD_{T,\Smj}^{-1} \bZ_{\Smj}' \bZ_{\Smj} \bD_{T,\Smj}^{-1} \right]^{-1}  \bD_{T,\Smj}^{-1} \bZ_{\Smj}' \be_j}_2\\
&\overset{(II)}{\leq} \bar{s}_T \norm{\bD_{T,\bZ}^{-1} \bZ_{-p}' \be_j}_{\infty} \norm{\left[\bD_{T,\Smj}^{-1} \bZ_{\Smj}' \bZ_{\Smj} \bD_{T,\Smj}^{-1} \right]^{-1} }_2 \norm{\bD_{T,\Smj}^{-1} \bZ_{\Smj}' \be_j}_\infty\\
&\overset{(III)}{\leq} \bar{s}_T \bkappa_{T,\min}^{-1} \norm{\bD_{T,\bZ}^{-1} \bZ_{-p}' \be_j}_{\infty}^2 \overset{(IV)}{\leq} \bar{s}_T \bar{\gamma}_T^2 \bkappa_{T,\min}^{-1} \leq \delta_T^2 T^{1/2},
\end{align*}
where (I) follows from the dual norm inequality and bounding the $\ell_1$ by the $\ell_2$ norm at the expense of the $\sqrt{\bar{s}_T}$ term; (II) follows from the fact that the columns of $\bD_{T,\bZ}^{-1}\bZ_{-p}$ contain all columns of $\bD_{T,\Smj}^{-1} \bZ_{\Smj}$, the definition of the induced norm and bounding the $\ell_2$ norm by the $\ell_{\infty}$ norm. (III) then follows from the eigenvalue condition in Assumption \ref{as:hlev}(\ref{as:rsev}) while (IV) invokes the empirical process bound of Assumption \ref{as:hlev}(\ref{as:ep}); the last inequality then follows from Assumption \ref{as:hlev}(\ref{as:rates}). For $A_{T,j,2}^2$ we can follow the same steps as above to again find that $A_{T,j,2}^2\leq \delta_T^2 T^{1/2}$.

We finally consider $A_{T,j,3}$. By the same arguments as for $A_{T,j,1}$ we get that $A_{T,j,3}$ reduces to $A_{T,j,3}=\abs{\boeta_j'\bZ_{-p}' \M(\bZ_{\hSmj})\be_j}$. Let us define the noiseless least squares estimator
\begin{equation*}
\tilde{\boeta}_{j,\hSmj} = \argmin_{\boeta:\eta_m = 0, m \notin \hSmj} \norm{\bZ_{-p} \boeta_j - \bZ_{-p} \boeta}_2^2, \qquad j = 1,\ldots,p,
\end{equation*}
such that $\M(\bZ_{\hSmj}) \bZ_{-p} \boeta_j = \bZ_{-p} (\boeta_j - \tilde{\boeta}_{j,\hSmj})$. Then, with probability $1- \Delta_T$,
\begin{align*}
A_{T,j,3} &= \abs{(\boeta_j - \tilde{\boeta}_{j,\hSmj})' \bD_{T,\bZ} \bD_{T,\bZ}^{-1} \bZ_{-p}' \be_j}\leq \norm{\bD_{T,\bZ}^{-1} \bZ_{-p}' \be_j}_\infty \norm{\bD_{T,\bZ} \left(\tilde{\boeta}_{j,\hSmj} - \boeta_j\right) }_1 \\
&\leq \bar{\gamma}_T \bar{s}_T \bkappa_{T, \min}^{-1} \norm{\bZ_{-p} \left(\tilde{\boeta}_{j,\hSmj} - \boeta_j\right) }_2 
\leq \bar{\gamma}_T \bar{s}_T \bkappa_{T, \min}^{-1} \norm{\bZ_{-p} \left(\hat{\boeta}_{j} - \boeta_j\right) }_2
\leq \delta_T^2 T^{1/2}.
\end{align*}
The first inequality follows from the dual norm inequality. The second inequality follows from combining the sparsity and the restricted eigenvalue conditions in Assumption \ref{as:hlev}(\ref{as:spar}) and \ref{as:hlev}(\ref{as:rsev}). In particular, letting $\tilde{\boeta}_{j,\hSmj}^{\dagger} = \bD_{T,\bZ} \tilde{\boeta}_{\hSmj}$ and $\boeta_{j}^{\dagger}= \bD_{T,\bZ} \boeta_j$, we bound $\norm{\tilde{\boeta}_{j,\hSmj}^{\dagger} - \boeta_{j}^{\dagger}}_2 \leq \bar{s}_T \norm{\bZ_{-p} \bD_{T,\bZ} \left(\tilde{\boeta}_{j,\hSmj}^{\dagger} - \boeta_{j}^{\dagger}\right)}_1$ as $\bD_{T,\bZ}$ is diagonal and therefore does not affect the sparsity. We then use that $\tilde{\boeta}_{j,\hSmj}$ was constructed to minimize the $\ell_2$ norm of the residuals. Finally, the last inequality follows the consistency Assumption \ref{as:hlev}(\ref{as:cons}) and the rate Assumption \ref{as:hlev}(\ref{as:rates}). 

Combining all results, we find that for all $j=1,\ldots,p$, we get that $\abs{A_{T,j}} \leq \delta_T^2 T^{1/2 - 2a} \leq \delta_T^2 T^{-1/2}$, as sufficient for establishing \eqref{proofA}.

Let us now prove \eqref{proofB}. By re-arranging terms we have
\begin{equation*}
\begin{aligned}
B_{T,j} &= T^{1/2-2a}  \underbrace{\bx_{-j}^{\prime} \left[\mathcal{M}(\bZ_{\hSmj}) - \mathcal{M}(\bZ_{\Smj}) \right] \bZ_{-p\backslash\{-j\}} \boeta_{0,-j}}_{B_{T,j,1}}\\
&\quad + T^{1/2-2a} \underbrace{\bx_{-j}^{\prime} \left[\mathcal{M}(\bZ_{\hSmj}) - \mathcal{M}(\bZ_{\Smj}) \right] \bu}_{B_{T,j,2}}.
\end{aligned}
\end{equation*}
We first deal with $B_{T,j,1}$. As $\boeta_0$ is only non-zero on $\S$, it follows directly that $\M(\bZ_{\Smj}) \bZ_{-p\backslash\{-j\}} \boeta_{0,-j} = \bzero$ and $B_{T,j,1}$ reduces to $B_{T,j,1} = \bx_{-j}^{\prime} \mathcal{M}(\bZ_{\hSmj}) \bZ_{-p} \boeta_0$. As above for \eqref{proofA}, using the projection $\bx_{-j} = \bZ_{-p} \boeta_j + \be_j$ we then get the following expression
\begin{equation}
\begin{aligned}
\label{second_eq}
B_{T,j,1} &= \underbrace{\boeta_j' \bZ_{-p}' \mathcal{M}(\bZ_{\hSmj}) \bZ_{-p\backslash\{-j\}} \boeta_{0,-j}}_{B_{T,j,1}^1} + \underbrace{\be_j' \mathcal{M}(\bZ_{\hSmj}) \bZ_{-p\backslash\{-j\}} \boeta_{0,-j}}_{B_{T,j,1}^2}.
\end{aligned}
\end{equation}
$B_{T,j,1}^1$ then follows analogously to $A_{T,j,1}$:
\begin{equation*}
\begin{aligned}
\abs{B_{T,j,1}^1} &\leq \sqrt{\norm{\M(\bZ_{\hSmj}) \bZ_{-p} \boeta_j}_2 \norm{\M(\bZ_{\hSmj}) \bZ_{-p\backslash\{-j\}} \boeta_{0,-j}}_2^2}\\
&\leq \sqrt{\norm{\M(\bZ_{\hS_j}) \bZ_{-p} \boeta_j}_2 \norm{\M(\bZ_{\hS_0}) \bZ_{-p\backslash\{-j\}} \boeta_{0,-j}}_2^2} \\
&\leq \sqrt{\min_{\boeta: \eta_{j,m} = 0, m \notin \hS_j} \norm{\bZ_{-p} \boeta_j - \bZ_{\hS_j} \boeta}_2^2 \min_{\boeta: \eta_{0,m} = 0, m \notin \hS_0} \norm{\bZ_{-p\backslash\{-j\}} \boeta_{0,-j} - \bZ_{\hS_0} \boeta}_2^2} \\
&\leq \norm{\bZ_{-p} (\boeta_j - \hat{\boeta}_j) }_2 \norm{\bZ_{-p\backslash\{-j\}} (\boeta_{0,-j} - \hat{\boeta}_0) }_2
\leq \delta_T T^{1/4} \norm{\bZ_{-p\backslash\{-j\}} (\boeta_{0,-j} - \hat{\boeta}_0) }_2.
\end{aligned}
\end{equation*}
Furthermore,
\begin{equation*}
\begin{aligned}
&\norm{\bZ_{-p\backslash\{-j\}} (\boeta_{0,-j} - \hat{\boeta}_0) }_2 \leq \norm{\bZ_{-p} (\boeta_{0} - \hat{\boeta}_0) }_2 + \norm{\bx_{j} (\eta_{0,j} - \hat{\eta}_{0,j}) }_2\\
&\quad\leq \delta_T T^{1/4} + \abs{\hat{\eta}_{0,j} - \eta_{0,j}} \sqrt{\sum_{t=p+d+1-j}^{T-j} x_{t-j}^2} \leq \delta_T T^{1/4} + \delta_T T^{-1/4} O_p (T^a),
\end{aligned}
\end{equation*}
by the consistency of $\hat{\boeta}_0$ in Assumption \partref{as:hlev}{as:cons} and since $\sqrt{\sum_{t=p+d+1-j}^{T-j} x_{t-j}^2} = O_p (T^a)$. We may then construct a sequence, say $\eta_T$, such that $\eta_T\to \infty$ but $\eta_T \delta_T \rightarrow 0$. It then follows that with probability $1-\Delta_T$, we have that $\sqrt{\sum_{t=p+d+1-j}^{T-j} x_{t-j}^2} \leq C \eta_T T^a$ and it then follows that $B_{T,j,1}^1 \leq \delta_T^2 T^{1/2} + \delta_T^2 \eta_T T^{a} \leq 2 \delta_T T^a$.

$B_{T,j,1}^2$ can be treated analogously to $A_{T,j,3}$. For $j = 1,\ldots,p$, define the estimators
\begin{equation*}
\tilde{\boeta}_{0,\hSmj} = \argmin_{\boeta:\eta_m = 0, m \notin \hSmj} \norm{\bZ_{-p\backslash\{-j\}} \boeta_{0,-j} - \bZ_{-p\backslash\{-j\}} \boeta}_2^2,
\end{equation*}
such that $\M(\bZ_{\hSmj}) \bZ_{-p\backslash\{-j\}} \boeta_{0,-j} = \bZ_{-p\backslash\{-j\}} (\boeta_{0,-j} - \tilde{\boeta}_{0,\hSmj})$. Then, with probability $1- \Delta_T$,
\begin{align*}
\abs{B_{T,j,1}^2} &= \abs{(\boeta_{0,-j} - \tilde{\boeta}_{0,\hSmj})' \bD_{T,\bZ,-j} \bD_{T,\bZ,-j}^{-1} \bZ_{-p\backslash\{-j\}}' \be_j}\\
&\leq \norm{\bD_{T,\bZ,-j}^{-1} \bZ_{-p\backslash\{-j\}} ' \be_j}_\infty \norm{\bD_{T,\bZ,-j} \left(\tilde{\boeta}_{0,\hSmj} - \boeta_{0,-j}\right) }_1 \\
&\leq \bar{\gamma}_T \bar{s}_T \bkappa_{T, \min}^{-1} \norm{\bZ_{-p} \left(\hat{\boeta}_{0} - \boeta_0\right) }_2 \leq \delta_T^2 T^{1/2}.
\end{align*}

For $B_{T,j,2}$, analogously to $A_{T,j,2}$ we have $\abs{B_{T,j,2}} \leq \underbrace{\abs{\be_j'\left(\P(\bZ_{\Smj})\right)\be_j}}_{B_{T,j,2}^1} + \underbrace{\abs{\be_j'\left(\P(\bZ_{\hSmj})\right)\be_j}}_{B_{T,j,2}^2}$.
Then
\begin{align*}
B_{T,j,2}^1 &\leq \bar{s}_T \norm{\bD_{T,\bZ}^{-1} \bZ_{-p}' \be_j}_{\infty} \norm{\bD_{T,\bZ}^{-1} \bZ_{-p}' \bu}_{\infty} \norm{\left[\bD_{T,\Smj}^{-1} \bZ_{\Smj}' \bZ_{\Smj} \bD_{T,\Smj}^{-1} \right]^{-1} }_2 \\
&\leq \bar{s}_T \bkappa_{T,\min}^{-1} \bar{s}_T \bar{\gamma}_T^2 \bkappa_{T,\min}^{-1} \leq \delta_T^2 T^{1/2},
\end{align*}
and it directly follows that $B_{T,j,2}^2 \leq \delta_T^2 T^{1/2}$ as well.

It then follows directly that $\abs{B_{T,j}} \leq \delta_T T^{1/2 - a} \leq \delta_T^2$, thereby establishing \eqref{proofB} for all $j=1,\ldots,p$.
\end{proof}

\begin{proof}[Proof of Theorem \ref{th:asyGC}]
We first show that the second stage regression delivers asymptotically normal estimators due to the lag augmentation. Let $\bV_{+,\hS} = (\bx_{-(p+1)}, \ldots, \bx_{-(p+d)}, \bV_{\hS})$ and $\bV_{+,\S} = (\bx_{-(p+1)}, \ldots, \bx_{-(p+d)}, \bV_{\S})$ and similarly let $\bV_{+,\hS}^*$ and $\bV_{+,\S}^*$ denote their counterparts with $\bx_{-(p+1)}^*, \ldots, \bx_{-(p+d)}^*$ replacing $\bx_{-(p+1)}, \ldots, \bx_{-(p+d)}$. By virtue of Lemmas \ref{lem:Ho_equiv} and \ref{lem:OLSequiv} as well as equations \eqref{proofA} and \eqref{proofB}, we have that
\begin{equation}\label{deviation_2}
\begin{split}
\sqrt{T} (\hat{\bbeta}_{p}^* - \bbeta_{p}^*)  &= \sqrt{T} \bR_{p\times p} (\hat{\bbeta}_{p} - \bbeta) = \sqrt{T} \bR_{p\times p} \left(\bX_{-p}^{\prime} \M(\bV_{+,\hS}) \bX_{-p} \right)^{-1} \bX_{-p}^{\prime} \M(\bV_{+,\hS})\bu \\
&= \sqrt{T} \bR_{p\times p} \left(\bX_{-p}^{\prime} \M(\bV_{+,\S}) \bX_{-p} \right)^{-1} \bX_{-p}^{\prime} \M(\bV_{+,\S})\bu + o_p(1) \\
&=\big(\underbrace{T^{-1} \bX_{-p}^{*\prime} \M(\bV_{+,\S}^*) \bX_{-p}^* \big)^{-1} }_{\bA_{T}} \underbrace{T^{-1/2} \bX_{-p}^{*\prime} \M(\bV_{+,\S}^*)\bu}_{\bb_{T}} + o_p(1).
\end{split}
\end{equation}
Let $\bS = \bV_{+, \S} \bQ_{\S}$ and $\bG_{T,\bQ} = \bQ_{\S} \bD_{T,\bS}$, , where $\bD_{T,\bS}$ is the matrix with $T^{1/2}, T$ and $T^2$ on the diagonal for the I(0), I(1) and I(2) elements, respectively. Then we can write
\begin{equation*}
\begin{split}
\bA_T &= T^{-1} \bX_{-p}^{*\prime} \bX_{-p}^* - T^{-1/2} \bX_{-p}^{*\prime} \bV_{+,\S}^* \bG_{T,\bQ}^{-1} \big(\bG_{T,\bQ}^{-1} \bV_{+,\S}^{*\prime} \bV_{+,\S}^* \bG_{T,\bQ}^{-1}\big)^{-1} \bG_{T,\bQ}^{-1} \bV_{+,\S}^{*\prime} \bX_{-p}^* T^{-1/2} \\
&= T^{-1} \bX_{-p}^{*\prime} \bX_{-p}^* - T^{-1/2} \bX_{-p}^{*\prime} \bS \bD_{T,\bS}^{-1} \big(\bD_{T,\bS}^{-1} \bS' \bS \bD_{T,\bS}^{-1} \big)^{-1} \bD_{T,\bS}^{-1} \bS' \bX_{-p}^* T^{-1/2}.
\end{split}
\end{equation*}
Let $\bS_{-0} = (\bS_1, \bS_2)$ and $\bD_{T,\bS_{-0}}$ the corresponding submatrix of $\bD_{T,\bS}$. It follows from inverting block matrices that
\begin{equation*}
\begin{aligned}
\bA_T &= \underbrace{T^{-1} \bX_{-p}^{*\prime} \M(\bS_{0}) \bX_{-p}^*}_{\bC_{T,1}} - \underbrace{T^{-1/2} \bX_{-p}^{*\prime} \M(\bS_{0}) \bS_{-0} \bD_{T,\bS_{-0}}^{-1}}_{\bR_{T,2}'} \\
& \quad \times \big(\underbrace{\bD_{T,\bS_{-0}}^{-1} \bS_{-0}' \M(\bS_{0}) \bS_{-0} \bD_{T,\bS_{-0}}^{-1} }_{\bR_{T,1}} \big)^{-1} \underbrace{\bD_{T,\bS_{-0}}^{-1} \bS_{-0}' \M(\bS_{0}) \bX_{-p}^* T^{-1/2}}_{\bR_{T,2}}
\end{aligned}
\end{equation*}
Then note that
\begin{equation*}
\begin{aligned}
\lambda_{\min} (\bR_{T,1}) &\geq \lambda_{\min} \left(\bD_{T,\bS_{-0}}^{-1} \bS_{-0}' \bS_{-0} \bD_{T,\bS_{-0}}^{-1} \right)\\
&\quad - \lambda_{\max} \left(\bD_{T,\bS_{-0}}^{-1} \bS_{-0}' \bS_{0} \bD_{T,\bS_{0}}^{-1} \big( \bD_{T,\bS_{0}}^{-1} \bS_{0}' \bS_{0} \bD_{T,\bS_{0}}^{-1} \big)^{-1} \bD_{T,\bS_{0}}^{-1} \bS_{0}' \bS_{-0} \bD_{T,\bS_{-0}}^{-1} \right) \\
&\geq \lambda_{\min} \left(\bD_{T,\bS_{-0}}^{-1} \bS_{-0}' \bS_{-0} \bD_{T,\bS_{-0}}^{-1} \right)\\
&\quad - \norm{\bD_{T,\bS_{-0}}^{-1} \bS_{-0}' \bS_{0} \bD_{T,\bS_{0}}^{-1} }_2^2 / \lambda_{\min}( \bD_{T,\bS_{0}}^{-1} \bS_{0}' \bS_{0} \bD_{T,\bS_{0}}^{-1} \big),
\end{aligned}
\end{equation*}
where the final inequality uses that, for compatible matrices $\bA$ and $\bB$, we have $\lambda_{\max} (\bA' \bB^{-1} \bA) \leq \lambda_{\max} (\bA' \bA) \lambda_{\max}(\bB^{-1}) = \norm{\bA}_2^2 / \lambda_{\min} (\bB)$. It then follows directly from Assumption \ref{ass:URproperties} that \begin{align*}&\lambda_{\min} \left(\bD_{T,\bS_{-0}}^{-1} \bS_{-0}' \bS_{-0} \bD_{T,\bS_{-0}}^{-1} \right) \geq \kappa_{T,12}, \quad \norm{\bD_{T,\bS_{-0}}^{-1} \bS_{-0}' \bS_{0} \bD_{T,\bS_{0}}^{-1} }_2^2 \leq \phi_{T,01}^2 + \phi_{T,02}^2,\\ &\lambda_{\min}( \bD_{T,\bS_{0}}^{-1} \bS_{0}' \bS_{0} \bD_{T,\bS_{0}}^{-1} \big) \geq \kappa_{T,0}.\end{align*} We may then conclude that
\begin{equation*}
\lambda_{\min} (\bR_{T,1}) \geq \kappa_{T,12} - (\phi_{T,01}^2 + \phi_{T,02}^2) / \kappa_{T,0} \geq \kappa_{T,-0}.
\end{equation*}

In similar fashion we can bound
\begin{equation*}
\begin{aligned}
\norm{\bR_{T,2}}_2 &\leq \norm{\bD_{T,\bS_{-0}}^{-1} \bS_{-0}' \bX_{-p}^* T^{-1/2}}_2 + \norm{\bD_{T,\bS_{-0}}^{-1} \bS_{-0}' \bS_{0} \bD_{T,\bS_{0}}^{-1}}_2 \\
&\quad \times \norm{ \bD_{T,\bS_{0}}^{-1} \bS_{0}' \bX_{-p}^* T^{-1/2}}_2 / \lambda_{\min}( \bD_{T,\bS_{0}}^{-1} \bS_{0}' \bS_{0} \bD_{T,\bS_{0}}^{-1} \big) \\
&\leq \phi_{T,01} + \phi_{T,02} + 2(\phi_{T,01}^2 + \phi_{T,02}^2) / \kappa_{T,0} \leq \phi_{T,-0}.
\end{aligned}
\end{equation*}
Letting 
\begin{equation}
\Sigma_{xx} := \E \left( T^{-1} \bX_{-p}^{*\prime} \bX_{-p}^* \right),\qquad
\Sigma_{sx} := \E \left(T^{-1} \bS_{0}' \bX_{-p}^* \right), \qquad
\Sigma_{ss} := \E \left(T^{-1} \bS_{0}' \bS_{0} \right),
\end{equation}
it follows directly from Assumption \ref{ass:URproperties} that $\norm{\bA_{T,1} - \bSigma_{x|w}}_2 \leq \delta_T$ with probability $1 - \Delta$, where 
\begin{equation}
\bSigma_{x|s} := \bSigma_{xx} - \bSigma_{sx}' \bSigma_{ss}^{-1} \bSigma_{sx}.
\end{equation}
Putting everything together, we then find that $\bA_T = \bSigma_{x|s} + o_p(1)$. 

Next we consider the numerator
\begin{equation*}
\begin{split}
\bb_T &= T^{-1/2}\bX_{-p}^{*\prime} \bu - T^{-1/2}\bX_{-p}^{*\prime} \bS \bD_{T,\bS}^{-1} \big(\bD_{T,\bS}^{-1} \bS' \bS \bD_{T,\bS}^{-1} \big)^{-1} \bD_{T,\bS}^{-1} \bS' \bu\\
&= \underbrace{T^{-1/2}\bX_{-p}^{*\prime} \M(\bS_{0}) \bu}_{\bb_{T,1}} - \underbrace{T^{-1/2} \bX_{-p}^{*\prime} \M(\bS_{0}) \bS_{-0} \bD_{T,\bS_{-0}}^{-1} }_{\bR_{T,2}}\\
&\quad\times \big(\underbrace{\bD_{T,\bS_{-0}}^{-1} \bS_{-0}' \M(\bS_{0}) \bS_{-0} \bD_{T,\bS_{-0}}^{-1}}_{\bR_{T,1}} \big)^{-1} \underbrace{\bD_{T,\bS_{-0}}^{-1} \bS_{-0}' \M(\bS_{0}) \bu}_{\br_{T, 3}}\\
\end{split}
\end{equation*}
The terms $\bR_{T,1}$ and $\bR_{T,2}$ were already treated for the denominator. For $\br_{T,3}$ we get
\begin{equation*}
\begin{aligned}
\norm{\br_{T,3}}_2 &\leq \norm{\bD_{T,\bS_{-0}}^{-1} \bS_{-0}' \bu}_2 + \norm{\bD_{T,\bS_{-0}}^{-1} \bS_{-0}' \bS_{0} \bD_{T,\bS_{0}}^{-1}}_2 \\
&\quad \times \norm{ \bD_{T,\bS_{0}}^{-1} \bS_{0}' \bu}_2 / \lambda_{\min}( \bD_{T,\bS_{0}}^{-1} \bS_{0}' \bS_{0} \bD_{T,\bS_{0}}^{-1} \big) \\
&\leq \gamma_{T,u,1} + \gamma_{T,u,2} + (\phi_{T,01} + \phi_{T,02}) \gamma_{T,u,0} / \kappa_{T,0} \leq \gamma_{T,u}.
\end{aligned}
\end{equation*}
Finally, it directly follows from Assumption \ref{ass:URproperties} that
\begin{equation*}
\norm{\bb_{T,1} - T^{-1/2}\bX_{-p}^{*\prime} \bu - \bSigma_{wx}' \bSigma_{ww}^{-1} \bS_{0}^{\prime} \bu}_2 \leq \delta_T. 
\end{equation*}
Then we can invoke the final statement of Assumption \ref{ass:URproperties} to conclude that
\begin{equation*}
\bb_{T,1} = T^{-1/2} 
\begin{bmatrix}
\bI_p & - \bSigma_{sx}' \bSigma_{ss}^{-1}
\end{bmatrix}
\begin{bmatrix}
\bX_{-p}^{*\prime} \bu \\ \bS_{0}^{\prime} \bu 
\end{bmatrix}
+o_p(1)\xrightarrow{d} N(\bzero, \sigma^2 \bSigma_{x|s}).
\end{equation*}

Plugging everything back into \eqref{deviation_2}, it then directly follows from the results above that
\begin{equation*}
\begin{split}
\sqrt{T} (\hat{\bbeta}_p^* - \bbeta_p^*) &\xrightarrow{d} N \left(\bzero, \sigma_u^2 \bSigma_{x|s}^{-1} \right).
\end{split}
\end{equation*}

We finally consider the test statistics. For the Wald test we get that
\begin{equation*}
\begin{split}
\Wald^* &= \hat{\sigma}_u^{*-2} \underbrace{T^{-1/2} \by'  \M(\bV_{+, \hS}^*) \bX_{-p}^* }_{\bF_{T}'}\big(\underbrace{T^{-1} \bX_{-p}^{*\prime} \M(\bV_{+,\hS}^*) \bX_{-p}^* \big)^{-1} }_{\bG_{T}} \underbrace{T^{-1/2} \bX_{-p}^{*\prime} \M(\bV_{+,\hS}^*)\by}_{\bF_{T}},
\end{split}
\end{equation*}
where $\hat{\sigma}_u^{*2} = \hat{\bu}_+^{*\prime} \hat{\bu}_+^*/T = \hat{\bu}_+^{\prime} \hat{\bu}_+/T$. It follows directly from the results obtained in the proof of Theorem \ref{th:pds_cons} as well as the current proof that $\bG_T = \bSigma_{x|w} + o_p(1)$ and $\bF_T \xrightarrow{d} N(\bzero, \sigma^2 \bSigma_{x|s})$. Furthermore,
\begin{equation} \label{eq:hatsigma}
\begin{split}
\hat{\sigma}_u^{*2} &= T^{-1} \bu' \bu - T^{-1} \bu' \M(\bV_{+,\hS}) \bX_{-p} (\hat{\bbeta}_p - \bbeta) + T^{-1} \bu' \M(\bV_{+,\hS}) \bV \bdelta \\
&\quad - T^{-1} (\hat{\bbeta}_p - \bbeta)' \bX_{-p}' \M(\bV_{+,\hS}) \bu + T^{-1} (\hat{\bbeta}_p - \bbeta)' \bX_{-p}' \M(\bV_{+,\hS}) \bX_{-p} (\hat{\bbeta}_p - \bbeta)\\
&\quad  + T^{-1} \bu' \M(\bV_{+,\hS}) \bV \bdelta - T^{-1} \bdelta' \bV' \M(\bV_{+,\hS}) \bu\\
&\quad  + T^{-1} \bdelta' \bV' \M(\bV_{+,\hS}) \bX_{-p} (\hat{\bbeta}_p - \bbeta) + T^{-1} \bdelta' \bV' \M(\bV_{+,\hS}) \bV \bdelta.
\end{split}
\end{equation}
All terms -- except for the leading $T^{-1} \bu'\bu$ term -- can be treated exactly as in the proof of Theorem \ref{th:pds_cons} to show that we may replace $\hS$ with $\S$, and as in the current proof to show their asymptotic negligibility. It then follows directly from the weak law of large numbers that $\hat{\sigma}_u^{*2} \xrightarrow{p} \sigma_u^2$. It then immediately follows that $\Wald \to \chi^2_p$.

The result for the LM test follows identically, except that we replace $\hat{\sigma}_u^{*2}$ by $\tilde{\sigma}_u^{*2} = \by' \M(\bV_{+,\hS}) \by$. As $\bbeta = \bzero$ under the null hypothesis, we can decompose $\tilde{\sigma}_u^{*2}$ as in \eqref{eq:hatsigma} with the terms involving  set to 0. It then directly follows that $\tilde{\sigma}_u^{*2} \xrightarrow{p} \sigma_u^2$ and $\LM \xrightarrow{d} \chi^2_p$.
\end{proof}

\section{
Verification of Assumption \ref{ass:URproperties}} \label{sec:verifyUR}
In this section we show how Assumption \ref{ass:URproperties} is satisfied using standard results from the unit root and cointegration literature. We will work under a set of assumptions commonly encountered in the cointegration literature \citep[see e.g.]{johansen1995likelihood}.\footnote{For simplicity we omit any deterministic components such as intercepts or linear trend terms in the analysis below. With suitable modifications however, the analysis goes through in a similar way if they are included.}

\begin{assumption}\label{ass:Johansen} 
Let $\bz_t$ denote an $n$-dimensional process, generated as $\bA(L) \bz_t = \bu_t$, where $\bA(z) = \bI - \sum_{j=1}^p \bA_j z^j$. Then $\bA(L)$ satisfies the following conditions:
\begin{itemize}
\item[(a)] $\abs{\boldsymbol{A}(z)}=0$ implies $|z|>1$ or $z=1$;
\item[(b)] Let $\bPi = -\bA(1)$, then we can write $\Pi = \boldsymbol{\mathcal{A}}\boldsymbol{B}'$, for some $\boldsymbol{\mathcal{A}}$ and $\boldsymbol{B}$, where $\boldsymbol{\mathcal{A}}$ and $\boldsymbol{B}$ are $n \times r$ matrices of rank $r$.
\end{itemize}
In addition, for an $n \times r$ matrix $\bC$ let $\bC_{\perp}$ be an $n\times (n-r)$ matrix of rank $n-r$ such that $\bC^{\prime} \bC_{\perp}=\bzero$, and define $\bar{\bC}_{\perp}=\bC_{\perp}\left(\bC_{\perp}^{\prime} \bC_{\perp}\right)^{-1}$. Then, one of the following conditions holds:
\begin{itemize}
\item[(c)] $\boldsymbol{\mathcal{A}}_{\perp}^{\prime} \tilde{\bA}(1) \boldsymbol{B}_{\perp}$ is nonsingular, where $\tilde{\bA}(z) = \bI_n -\sum_{j=1}^{p-1} \tilde{\bA}_j z^j$ and $\tilde{\bA}_j = \sum_{i=j+1}^{p} \bA_i$;
\item[(c$^\prime$)] $\bar{\boldsymbol{\mathcal{A}}}_{\perp}^{\prime} \tilde{\bA}(1) \bar{\boldsymbol{B}}_{\perp}=\bF \bG^{\prime}$, where and $\bF$ and $\bG$ are $(n-r)\times r'$ matrices of rank $r'$ ($0\leq r' < n - r$).
\end{itemize}
\end{assumption}

The rationale for these assumptions comes from rewriting the VAR model as a vector error correction model (VECM), such that
\begin{equation}\label{ecm}
\Delta \bz_t = \bPi \bz_{t-1} + \sum_{j=1}^{p-1} \tilde{\bA}_j \Delta z_{t-j} +  + \bu_t,
\end{equation}
where $\tilde{\bA}_j = \sum_{i=j+1}^{p} \bA_i$ and $\bPi = \bA(1)$. Assumption \ref{ass:Johansen}(a) rules out explosive processes but allows for unit roots; \ref{ass:Johansen}(b) allows for a reduced rank in $\bPi$, i.e., cointegration is present. If \ref{ass:Johansen}(c) holds, $\bz_t$ is $I(1)$ and linear combinations $\bB' \bz_t$ are $I(0)$. If \ref{ass:Johansen}(c') holds, $\bz_t$ is $I(2)$ and appropriate linear combinations are $I(1)$ or even $I(0)$. This can be seen by again rewriting the VECM as
\begin{equation}\label{ecm2}
\Delta^2 \bz_t = \bPi \bz_{t-1} + \bPhi \Delta \bz_{t-1} + \sum_{j=1}^{p-2} \bA^{*}_j \Delta^2 \bz_{t-j} + \bu_t,
\end{equation}
where $\bPhi = \tilde{\bA}(1)$ and $\bA_j^* = \sum_{i=j+1}^{p-1} \tilde{\bA}_i$.\footnote{Formally, we need another assumption to rule out that the process is $I(3)$, see \citet[Theorem 3]{johansen1992representation} for details.} 

Under these assumptions, we can find appropriate rotation matrices $\bQ$ such that $\bzeta_t = \bQ \bz_t$ can be partitioned into variables ``full-rank'' $\bzeta_{0,t}, \bzeta_{1,t}$ and $\bzeta_{2,t}$ that are $I(0)$, $I(1)$ and $I(2)$ respectively (where some sets may be empty). Limit results for appropriate products of variables with different orders of integration are well-known and have been derived by a variety of papers -- often under weaker assumptions than Assumptions \ref{ass:basic} and \ref{ass:Johansen} -- see for example \citet{phillips1986multiple}, \citet{park1988statistical}, \citet{sims1990inference}, \citet{phillips1992asymptotics}, or \citet{hamilton1994time} for an overview. We therefore state these results here without proof.

\begin{lemma} \label{lem:URfixed_n}
Let $\bz_t$ satisfy Assumptions \ref{ass:basic} and \ref{ass:Johansen} and let $\bzeta_t = \bQ \bz_t$. Let $\tilde{\bzeta}_t = (\bzeta_{0,t}, \Delta \bzeta_{1,t}, \Delta^3 \bzeta_{2,t})$ and define $\bSigma := \E \tilde{\bzeta}_t \tilde{\bzeta}_t'$, and partition it conformably such that $\bSigma_{ij} = \E \tilde{\bzeta}_{i,t} \tilde{\bzeta}_{j,t}'$. Let $\bLambda = \sum_{j=1}^{\infty}\E \tilde{\bzeta}_t \tilde{\bzeta}_{t+j}'$ and $\bOmega = \bSigma + \bLambda + \bLambda'$, where $\bOmega_{ij}$ refers to the same submatrix as $\bSigma_{ij}$, and assume that both $\bSigma$ and $\bOmega$ are full rank. Let $\bB(\cdot)$ denote an $n$-dimensional Brownian motion with covariance matrix $\bOmega$, which can be partitioned conformably with $\bzeta_{0,t}, \bzeta_{1,t}$ and $\bzeta_{2,t}$ as $\bB(\cdot) = (\bB_0(\cdot)', \bB_1(\cdot)', \bB_2(\cdot)')'$, $\Bar{\bB}_i(s)=\int_0^s \bB_i(u)du$ and $B_u(\cdot)$ a univariate Brownian motion independent of $\bB(\cdot)$. Then we have the following results (where convergence applies jointly if needed):
\begin{align*}
&(a) \quad T^{-1}\sum_{t=1}^T \bzeta_{0,t} \bzeta_{0,t}'\overset{p}{\to}\bSigma_{00}; &
&(b) \quad T^{-1/2}\sum_{t=1}^T \bzeta_{0,t-1} u_t \overset{d}{\to} N(\bzero, \sigma_u^2 \bSigma_{00}); \\
&(c) \quad T^{-1}\sum_{t=1}^T \bzeta_{1,t-1} u_t \overset{d}{\to}\int_0^1 \bB_1(s)'d\bB_{u}(s); &
&(d) \quad T^{-1}\sum_{t=1}^T \bzeta_{1,t} \bzeta_{0,t}' \overset{d}{\to} \int_0^1 \bB_1(s) d\bB_{0}(s)' + \bSigma_{10} + \bLambda_{10};\\
&(e) \quad T^{-2}\sum_{t=1}^T \bzeta_{1,t} \bzeta_{1,t}' \overset{d}{\to} \int_0^1 \bB_1(s) \bB_1(s)' ds; &
&(f) \quad T^{-2}\sum_{t=1}^T \bzeta_{2,t-1} u_t \overset{d}{\to}\int_0^1 \Bar{\bB}_2(s)'d\bB_{u}(s);\\
&(g) \quad T^{-2}\sum_{t=1}^T \bzeta_{2,t} \bzeta_{0,t}' \overset{d}{\to} \int_0^1 \Bar{\bB}_2\bB_{0}(s)'&
&(h) \quad T^{-3}\sum_{t=1}^T \bzeta_{2,t} \bzeta_{1,t}' \overset{d}{\to} \int_0^1 \Bar{\bB}_2 \bB_{1}(s)'\\
&(i) \quad T^{-4}\sum_{t=1}^T \bzeta_{2,t} \bzeta_{3,t}' \overset{d}{\to} \int_0^1 \Bar{\bB}_2 \Bar{\bB}_2(s)' 
\end{align*}
\end{lemma}

Using these results, we can now determine the appropriate rates in Assumption \ref{ass:URproperties}. Let $\eta_T$ denotes a sequence such that $\eta_T \to \infty$ at an arbitrarily slow rate.  First, it follows from (d), (f) and (h) that $\phi_{T,01} = \phi_{T,02} = \eta_T T^{-1/2}$. Next, we may directly deduce from (a) that $\kappa_{T,0} = \eta_T^{-1}$. Similarly, it follows from (e), (g), (h) and (i) that $\kappa_{T,12} = \eta_T^{-1}$. This was shown for $\kappa_{T,1}$ in \citet{phillips1990statistical}, but the same reasoning can be applied to the whole matrix. It follows from (b), (c) and (f) that $\gamma_{T,i} = \eta_T$ for $i=0,1,2$, and finally result (d) of Assumption \ref{ass:URproperties} follows from (a). 

It then remains to check the rates. As we chose $\eta_T$ growing arbitrarily slowly, it follows that $\eta_T^3 T^{-1/2} \to 0$. We then have that
\begin{enumerate}[(i)]
\item $\kappa_{T,0} -  (\phi_{T,01}^2 + \phi_{T,02}^2) / \kappa_{T,0} = \eta_T^{-1} - 2 \eta_T^{3} T^{-1} \geq \eta_T^{-1} / 2 = \kappa_{T,-0}$;
\item $\phi_{T,01} + \phi_{T,02} + 2(\phi_{T,01}^2 + \phi_{T,02}^2) / \kappa_{T,0} = 2\eta_T T^{-1/2} + 2 \eta_T^3 T^{-1} \leq 2 \eta_T T^{-1/2} = \phi_{T,-0}$;
\item $\gamma_{T,u,1} + \gamma_{T,u,2} + (\phi_{T,01} + \phi_{T,02}) \gamma_{T,u,0} / \kappa_{T,0} = 2 \eta_T + 2 \eta^3 T^{-1/2} \leq 2\eta_T = \gamma_{T,u}$;
\item $\phi_{T,-0}(\phi_{T,-0} + \gamma_{T,u}) / \kappa_{T,-0} \leq \eta_T^3 T^{-1/2} \leq \delta_T$.
\end{enumerate}
We complete this verification by noting that the final result simply follows from result (b).

\section{
Additional Empirical Analysis} \label{sec:vix}
Here we extend the analysis of Section \ref{sec:emp_appl} using the latest vintage of FRED-MD at the moment of writing, namely that of September $2022$ (453 observations). As mentioned, this vintage does include all latest crises too and VXOCLSx is replaced by VIXCLSx.
\begin{figure}[H]
\begin{minipage}{0.5\textwidth}
\centering
\includegraphics[width=0.95\textwidth, trim = {3.4cm 3cm 2.2cm 2.6cm},clip]{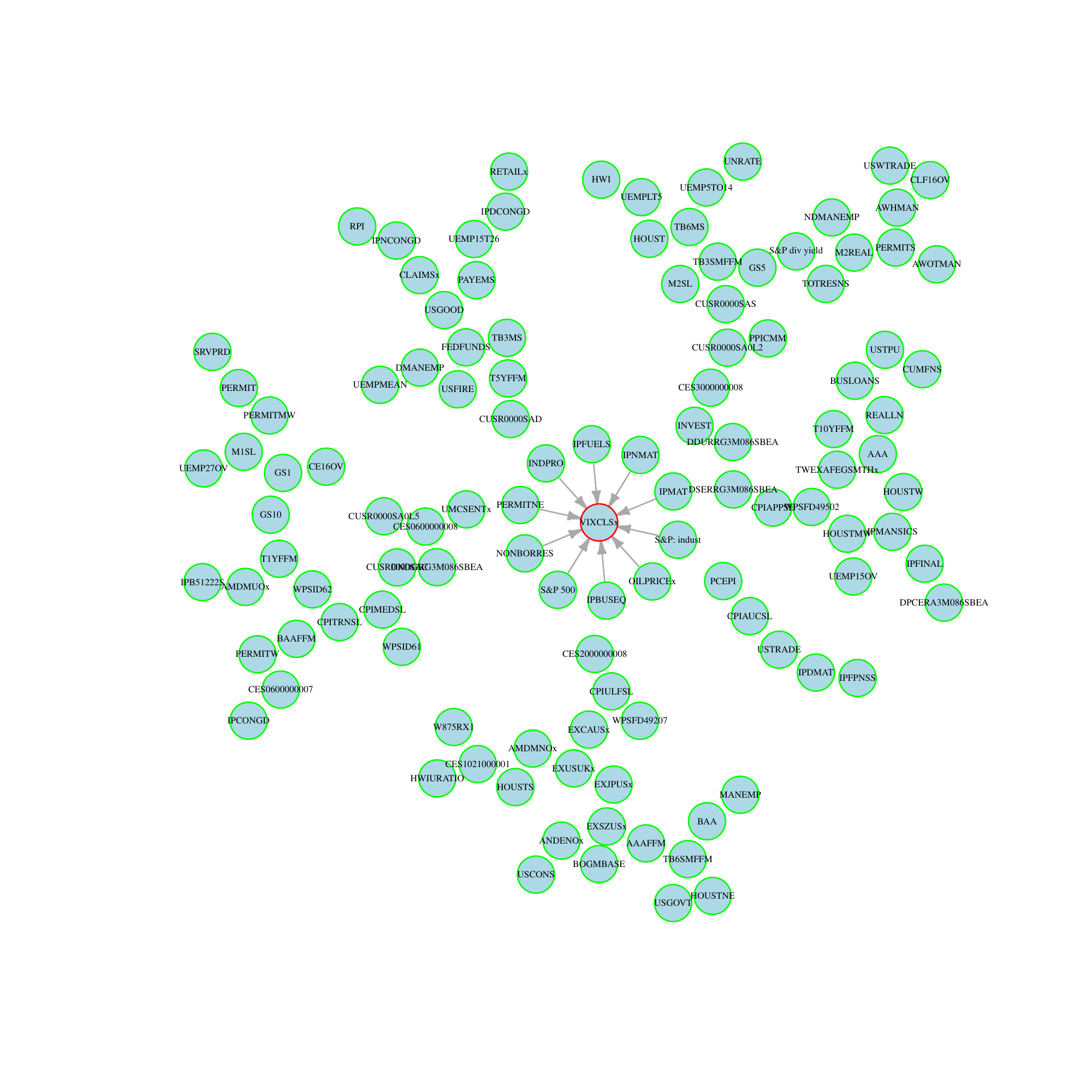}
\caption{PDS-LA-LM, Causes of VIXCLSx, $\alpha=0.01$}
\label{figura_vix1}
\end{minipage}
\begin{minipage}{.5\textwidth}
        \centering
\includegraphics[width=0.95\textwidth, trim = {3.4cm 3cm 2.2cm 2.6cm},clip]{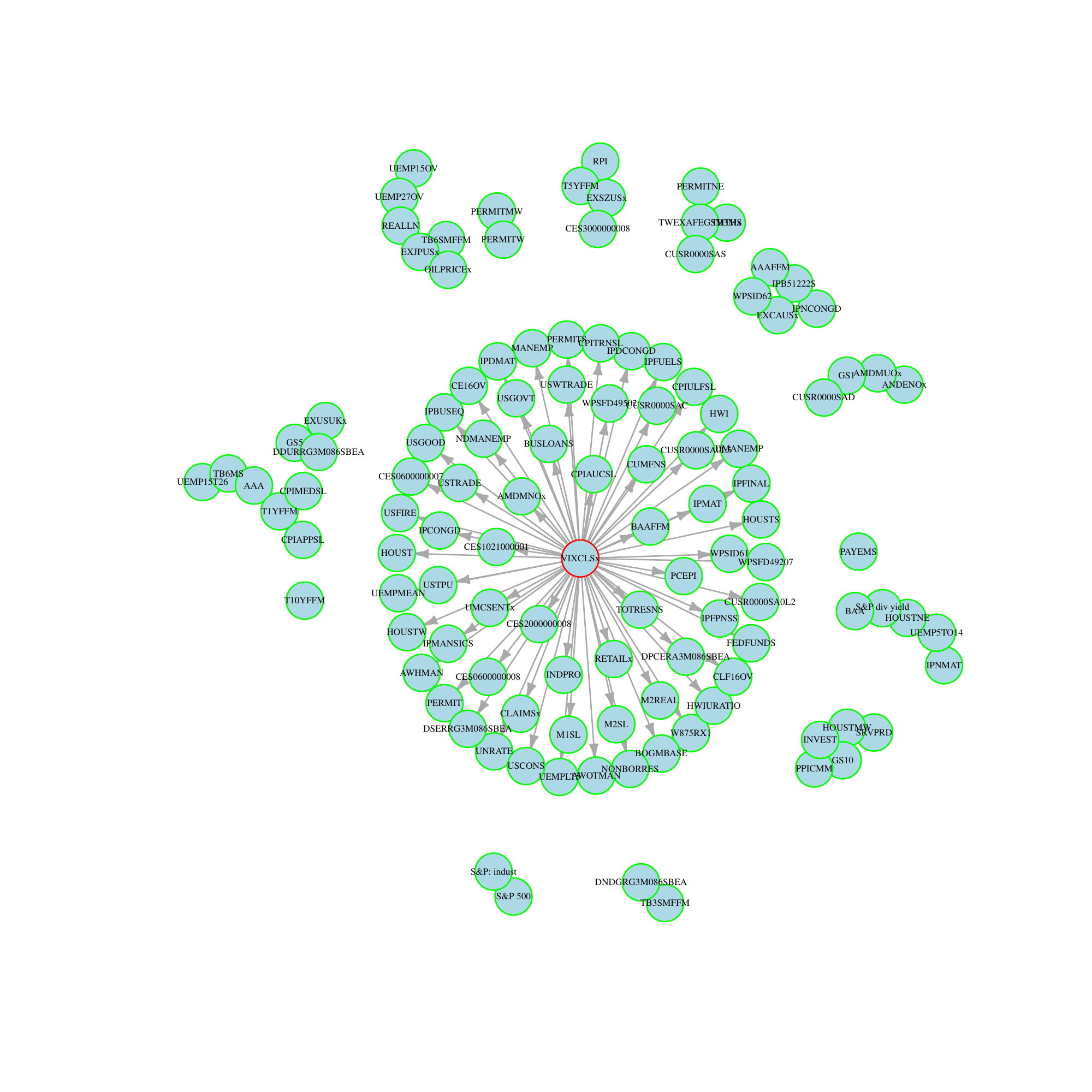}
\caption{PDS-LA-LM, Caused by VIXCLSx, $\alpha=0.01$}
\label{figura_vix2}
\end{minipage}
\end{figure}

Interestingly, a comparison of Figure \ref{figura_vox1} and Figure \ref{figura_vix1} shows how all the connections FRED$\rightarrow$VXOCLSx in the November 2019 FRED-MD vintage are found also in the FRED$\rightarrow$VIXCLSx October 2022 vintage plus two more: IPFUELS and OILPRICEx. This tells us two things: first, the change from VXOCLSx to VIXCLSx has not changed the causal infrastructure among the FRED-MD variables. This is sensible as VIXCLSx has been built from VXOCLSx and precisely with the aim to replace it. Second, the various global crises occurred between 2019 and 2022 have also not impacted the amount of insourcing Granger causal connections onto VIXCLSx but have increased substantially more the outsourcing amount of connections, as clear from Figure \ref{figura_vix2}. From 41 connections of VXOCLSx$\rightarrow$FRED in Figure \ref{figura_vox2} for the 2019 vintage, the amount of connections increased in the latest vintage to a staggering 65 connections (at $1\%$ significance level). Therefore, this contributes to the evidence that latest crises have increased substantially the perceived economic uncertainty and such uncertainty is outpouring towards a large number of macroeconomic variables.

\begin{figure}
\begin{minipage}{0.5\textwidth}
\centering
\includegraphics[width=0.95\textwidth, trim = {3.4cm 3cm 2.2cm 2.6cm},clip]{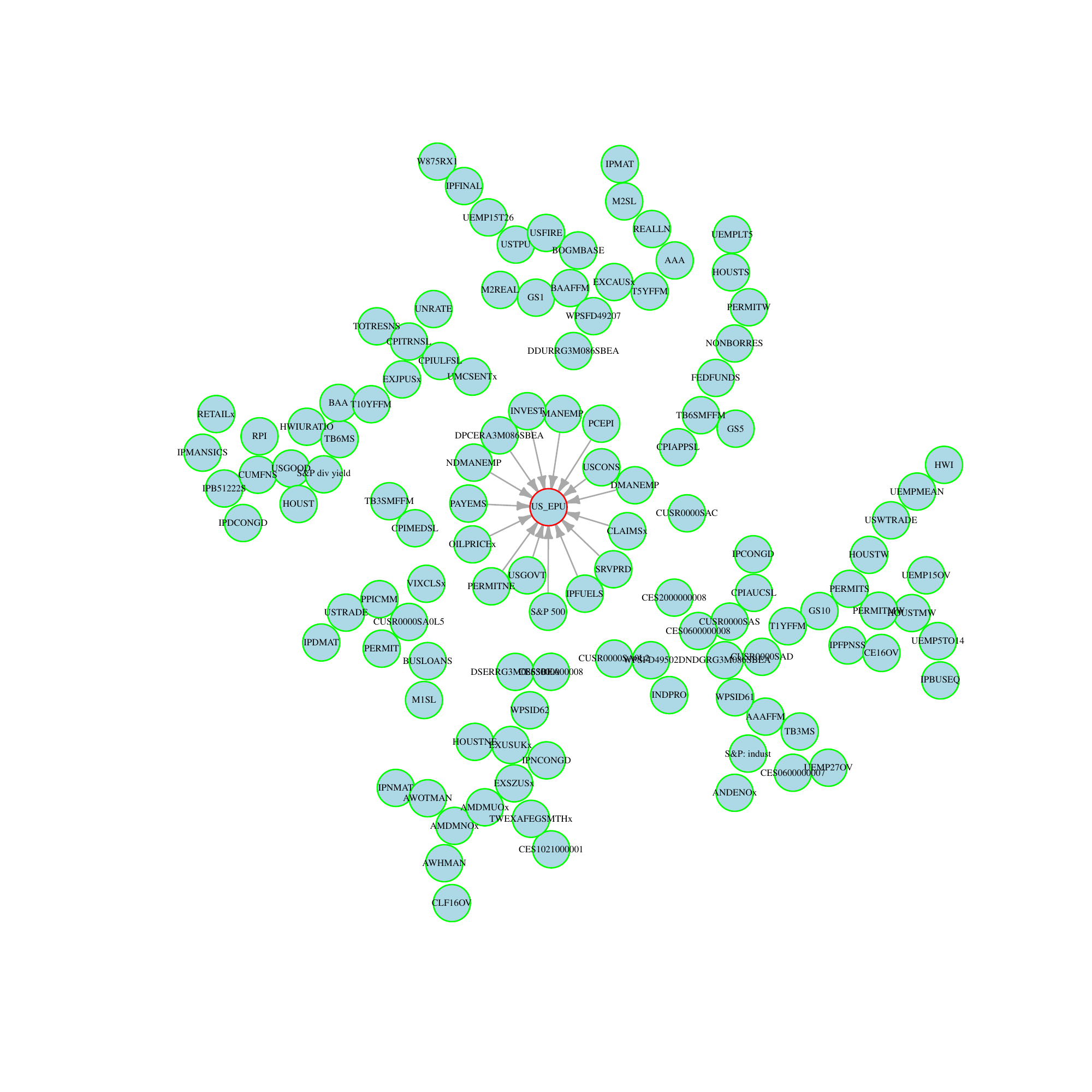}
\caption{PDS-LA-LM, Causes of US-EPU, $\alpha=0.01$}
\label{figura_vix3}
\end{minipage}
\begin{minipage}{.5\textwidth}
        \centering
\includegraphics[width=0.95\textwidth, trim = {3.4cm 3cm 2.2cm 2.6cm},clip]{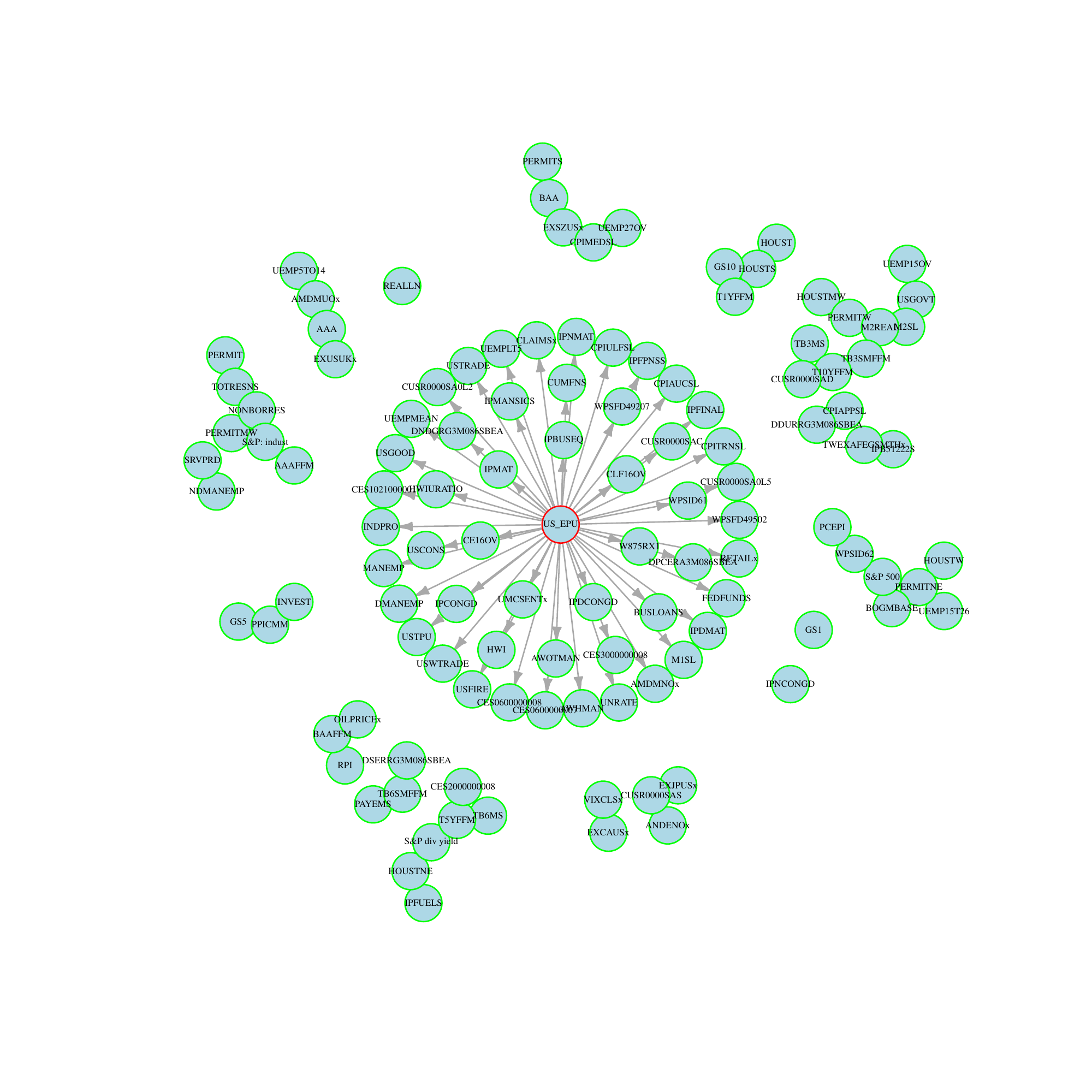}
\caption{PDS-LA-LM, Caused by US-EPU, $\alpha=0.01$}
\label{figura_vix4}
\end{minipage}
\end{figure}
\end{appendices}

\end{document}